\documentclass[aps,prb,notitlepage,nofootinbib]{revtex4-2}
\usepackage{graphicx}
\usepackage[dvipsnames]{xcolor}
\usepackage{amsfonts}
\usepackage{amssymb}
\usepackage{amsmath}
\usepackage{mathtools}
\usepackage{afterpage}
\usepackage[mathscr]{euscript}
\usepackage{braket}
\usepackage{subfigure}
\usepackage{multirow}
\usepackage{float}
\usepackage{qcircuit}
\usepackage{caption}
\usepackage{enumitem}
\usepackage{color}   
\usepackage{hyperref}
\hypersetup{
    colorlinks=true, 
    linktoc=page,     
    linkcolor=blue,  
    urlcolor=blue, breaklinks=true
}

\usepackage[utf8]{inputenc} 
\usepackage{yfonts}
\usepackage{dsfont}
\usepackage{cancel}
\usepackage[scr=boondox]{mathalpha}

\usepackage[pdf]{pstricks}
\usepackage{leftidx}
\usepackage{pst-all}
\usepackage{booktabs}

\usepackage[bottom]{footmisc} 

\usepackage{amsthm,mathtools,amsmath}
\usepackage{tikz}

\newcommand{\Z}{\mathbb{Z}}
\newcommand{\R}{\mathbb{R}}

\newcommand{\ra}{\rightarrow}

\newcommand{\bface}{{\bf f}}

\newcommand{\bmsquare}{\boldsymbol{\msquare}}
\newcommand{\bDelta}{{\bf \Delta}}

\usepackage{environ}
\NewEnviron{eqs}{%
\begin{equation}\begin{split}
    \BODY
\end{split}\end{equation}
}

\usepackage{hyperref}
\hypersetup{
    colorlinks=true,
    linkcolor=blue,
    filecolor=magenta,
    urlcolor=cyan,
}

\newtheorem{defn}{Definition}
\newtheorem{thm}{Theorem}
\newtheorem{lem}[thm]{Lemma}
\newtheorem{prop}[thm]{Proposition}

\newcommand\restr[2]{{
  \left.\kern-\nulldelimiterspace 
  #1 
  \vphantom{\big|} 
  \right|_{#2} 
  }}

\usepackage{MnSymbol}
 
\usepackage{scalerel,amssymb}
\newcommand{\msquare}{{\mathord{{\scalerel*{\Box}{t}}}}}

\newcommand{\ihat}{\hat{\imath}}
\newcommand{\jhat}{\hat{\jmath}}

\begin{document}

\def\SU{\mathrm{SU}}
\def\SO{\mathrm{SO}}
\def\Spin{\mathrm{Spin}}

\title{Higher cup products on hypercubic lattices: application to lattice models of topological phases}
\author{Yu-An Chen$^{1,2}$}
\author{Sri Tata$^1$}
\affiliation{$^1$Condensed Matter Theory Center, Joint Quantum Institute, and Department of Physics, NIST/University of Maryland, College Park, Maryland 20472 USA}
\affiliation{$^2$Joint Center for Quantum Information and Computer Science, University of Maryland, College Park,
Maryland 20742, USA}

\begin{abstract}
In this paper, we derive the explicit formula for higher cup products on hypercubic lattices, based on the recently developed geometrical interpretation on the simplicial complexes. We illustrate how this formalism can elucidate lattice constructions on hypercubic lattices for various models and deriving them from spacetime actions. In particular, we demonstrate explicitly that the (3+1)D SPT $S = \frac{1}{2} \int w_2^2 + w_1^4$ (where $w_1$ and $w_2$ are the first and second Stiefel-Whitney classes) is dual to the 3-fermion Walker-Wang model constructed on the cubic lattice.
Other examples include the double-semion model, and also the ``fermionic'' toric code in arbitrary dimensions on hypercubic lattices. In addition, we extend previous constructions of exact boson-fermion dualities and the Gu-Wen Grassmann Integral to arbitrary dimensions. Another result which may be of independent interest is a derivation of a cochain-level action for the generalized double-semion model, reproducing a recently derived action on the cohomology level.
\end{abstract}

\maketitle

\twocolumngrid
\tableofcontents
\onecolumngrid

\section{Introduction}
In algebraic topology, the cup product and higher cup products are fundamental operations on cocycles and cohomology. Cup products defined on simplicial complexes were central to the early development of algebraic topology, and the higher cup products and Steenrod operations introduced in~\cite{S47} have also been valuable tools for topologists. These operations were shortly afterwards interpreted geometrically~\cite{thom1950} in terms of Poincaré duals of cohomology classes. While these operations have been central to the cohomology-level understanding of topology, cochain-level formulas on concrete simplicial complexes have often been considered opaque and unnecessary for the study of manifold topology. However, renewed interest in these cochain-level formulas has come from the physics community in the study of topological phases of matter. In particular, Gu-Wen~\cite{guWen2014Supercohomology} found that cochain-level formulas for higher-cup products and Steenrod squares pop out of fermionic path integral constructions of topological lattice models. And recently, geometric interpretations~\cite{T18,T20} of the cochain-level formulas have been found and related to combinatorial spin structures and fermions.

In this paper, we extend the construction of~\cite{T20} to the hypercubic lattice and define cochain-level formulas for the higher cup operations on hypercubic lattices. We then apply these formulas and give an expository account through the hypercubic lattice of how they can be used to \textit{systematically} derive Hamiltonians of some well-known lattice models of topological phases given cochain-level formulas of space-time partition functions of topological actions.

The organization of this paper can roughly be split into two parts; the first half Secs.~\ref{sec:Notation}-\ref{sec:cochainsOverZ} introduces our operations, while the second half Secs.~\ref{sec:FermionsAndExactBosonization}-\ref{sec:ConstructionOfLatticeModels} applies the operations to various physical models.

In Section~\ref{sec:reviewOfOlderNotations}, we review the simplicial higher cup operations and some previous partial results for conventions in the two- and three-dimensional cubical cases. In Sec.~\ref{sec:Notation}, we establish convenient notation for the hypercubic cell structures that we use throughout and define chains and cochains. In Sec.~\ref{sec:cupM_OverZ2}, we give systematic definitions of the operations in all dimensions over $\Z_2$. In Sec.~\ref{sec:cochainsOverZ}, we do the same over $\Z$. Here, we find that various sign factors that show up over $\Z$ have interpretations in terms of signed intersections of dual cells. 

In Sec.~\ref{sec:FermionsAndExactBosonization}, we discuss applications of our formalism to fermionic systems in hypercubic lattices of arbitrary dimensions. In particular, we extend the $d$-dimensional exact boson-fermion duality~\cite{chen2020BosonArbDimensions} and the Gu-Wen Grassmann Integral~\cite{guWen2014Supercohomology} in $d$-dimensions to the hypercubic case. In Sec.~\ref{sec:ConstructionOfLatticeModels}, we illustrate how to systematically construct Hamiltonian models of topological phases from topological actions in our hypercubic formalism.
Particular examples include the $\Z_n$ toric code in arbitrary dimensions, the double-semion model~\cite{levinGu2012}, the 3-fermion Walker-Wang model~\cite{BCFV14}, and fermionic toric code~\cite{kapustinThorngren2017,chen2020BosonArbDimensions}. On the 3-fermion Walker-Wang model, we briefly discuss how our formalism can simplify the recently-constructed Quantum Cellular Automaton (QCA) \cite{HFH18} \footnote{In the original construction, the algebraic relations are checked by computers using the polynomial method. In the cup product formalism, the construction can be verified by simple calculations.}. Details are deferred to a future work~\cite{myfuturework}. 

Many detailed and technical calculations are relegated to the appendices. However, Appendix~\ref{app:cochainActionGenDoubleSemion} may be of independent interest from the main threads of this paper; there we consider the generalized double-semion model~\cite{FH16} and give a cochain-level derivation of the topological action~\cite{debray2019} that describes it.

\section{Review of $\cup_i$ products on triangulations and previous conventions on a cubic lattice} \label{sec:reviewOfOlderNotations}

First, we review some concepts in algebraic topology and also introduce notations used in this paper, in particular reviewing notation of triangulations and previous constructions of hypercubic cochain operations.

We will always work with an arbitrary triangulation of a simply-connected $d$-dimensional triangulation manifold $M^d$ equipped with a branching structure. A branching structure is an assignment of arrows to each $1$-simplex that doesn't form any closed loops on triangles. This is equivalent to there being a total order on vertices of a $d$-simplex. In general, we'll write $\Delta_k = \braket{0 \cdots k}$ so that the branching structure locally orders edges as $0 \to \cdots \to k$. We often refer to vertices, edges, faces, and tetrahedra as $v,e,f,t$, respectively. A general $k$-simplex is denoted as $\Delta_k$, and we refer to $T^k(M^d)$ as the set of $k$-simplices of $M^d$.

Associated to the set of $k$-simplices of $M^d$ is the chain space $C_k(M^d,\Z)$ which is a free $\Z$-module generated by basis vectors $\{\ket{0 \cdots k}\}_{\braket{0 \cdots k} \in T^k(M^d)}$. A finite linear combination of $k$-simplices with coefficients in $\Z$ is called a $k$-chain. The $k$-chains form an abelian group denoted $C_k(M^d,\Z)$. One can also consider the reductions of these space modulo 2 to give the  $\Z_2$ chain spaces $C_k(M^d,\Z_2)$. A $\Z_2$-valued $k$-chain can be identified with a finite set of $k$-simplices (the set consisting of those simplices whose coefficients are nonzero).

Now we define cochains, over $\Z_2$. A $\Z_2$-valued $k$-cochain is a function from the set of $k$-simplices to $\Z_2$, and can be expressed in terms of the dual vectors $\{\bra{0 \cdots k}\}_{\braket{0 \cdots k} \in T^k(M^d)}$. The set of all $k$-cochains is an abelian group denoted $C^k(M^d,\Z_2)$. For example, a $0$-cochain assigns an element of $\{0,1\} \in \Z_2$ to all vertices, and a $1$-cochain assigns $0$ or $1$ to all edges. A $k$-cochain acts on $k$-chain by evaluating the cochain on each simplex of the $k$-chain and adding up the results modulo $2$. This gives the \textit{chain-cochain pairing}. For $a \in C_k(M^d,\Z_2)$ and $b \in C^k(M^d,\Z_2)$, we often use the notations
\begin{equation}
    \int_{M} (a)(b) = \int_M (b)(a) = \int_a b
\end{equation}
to denote this pairing.

For each simplex $\Delta_k$, it's useful to define the indicator cochain ${\bf \Delta}_k$. For example, a vertex $v$ is associated to a 0-cochain ${\bf v}$, which takes value $1$ on $v$, and $0$ otherwise, i.e. ${\bf v}(v') = \delta_{v,v'} $. Similarly, edges $e$ give us 1-cochains ${\bf e}$ with ${\bf e}(e') = \delta_{e,e'} $, and so forth, i.e., ${\bf f}(f) = \delta_{f, f'}$ for the face $f$ \footnote{In principle, vertices, edges, and faces are labelled by $\Delta_0$, $\Delta_1$, and $\Delta_2$, and their dual cochains are $\bDelta_0$, $\bDelta_1$, and $\bDelta_2$. For familiarity, these low dimensional simplices are often denoted as $v$, $e$, and $f$ for vertices, edges, and faces.}. All such indicator cochains will be denoted in bold. 

The boundary operator is denoted by $\partial$. For a $k$-simplex $\Delta_k = \braket{0 \cdots k}$, $\partial \Delta_k$ consists of all boundary $(k-1)$-simplices of $\Delta_k$. In other words,
\begin{equation}
    \partial \ket{0 \cdots k} = \sum_{i=0}^k \ket{0 \cdots \hat{i} \cdots k}
\end{equation}
where $\hat{i}$ means skipping over that index.

The coboundary operator $\delta$ (not to be confused with the previous Kronecker delta) acts on cochains and is the the adjoint of $\partial$. On indicator cochains, we'll have
\begin{equation}
    \delta {\bf \Delta}_k = \sum_{ \Delta_{k+1} \supset \Delta_k} {\bf \Delta}_{k+1}.
\end{equation}
where the notation $\Delta_{k+1} \supset \Delta_k$ means that $\Delta_k$ is a subsimplex of $\Delta_{k+1}$. For general cochains, the definition above extends by linearity to give
\begin{equation}
    \delta \alpha = \sum_{\Delta_k \in \alpha} \sum_{ \Delta_{k+1} \supset \Delta_k} {\bf \Delta}_{k+1}.
\end{equation}

Now, we review product operations over cochains on triangulations. These constructions require a branching structure on the triangulation in order to determine local vertex orderings, whereas the boundary operators did not.

The cup product $\cup \equiv \cup_0$ and the higher cup product $\cup_n$ were orignally defined by Steenrod \cite{S47} for an arbitrary simplicial complex, and interpreted geometrically in \cite{T18, T20}. In Ref.~\cite{CKR18, CK18}, similar definitions for $\cup$ and $\cup_1$ were given for 2D square lattices and 3D cubic lattices. Here, we'll review previous definitions on both a triangulated lattice and the cubic lattices and how they fit into the framework of the geometric interpretations.

On a triangulated space, the cup product $\cup$ of a $p$-cochain $\alpha_p$ and a $q$-cochain $\beta_q$ is a $(p+q)$-cochain defined as \cite{H00}:
\begin{equation}
    \begin{split}
        (\alpha_p \cup \beta_q ) (0, \dots, p+q) &= \alpha_p (0, 1, \dots, p) \beta_q (p, p+1, \dots, p+q) \\
        &= \alpha_p (0 \to p) \beta_q (p \to p+q).
    \end{split}
\end{equation}
where $i \to j = \{i, i+1, \cdots, j-1, j\}$ are the integers between $i$ and $j$. The cup product satisfies the important Leibniz rule:
\begin{equation}
    \delta (\alpha \cup \beta) = \delta \alpha \cup \beta + \alpha \cup \delta \beta.
\end{equation}

The $\cup_m$ product of a $p$-cochain with a $q$-cochain is a $(p+q-m)$-cochain defined as
\begin{equation} \label{eq: definition of higher cup product}
    (\alpha_p \cup_m \beta_q) (0,1,\cdots, p+q-m)= \sum_{0 \leq i_0 < i_1 < \cdots < i_m \leq p+q-m } \alpha_p(0 \to i_0,i_1 \to i_2, i_3 \to i_4, \cdots)  \beta_q(i_0 \to i_1, i_2 \to i_3, \cdots)
\end{equation}
where $\{ i_0,i_1, \dots, i_m \}$ are chosen such that the arguments of $\alpha_p$ and $\beta_q$ contain $p+1$ and $q+1$ numbers separately. For example,
\begin{equation}
    (\alpha_2 \cup_1 \beta_2) (0123)= \alpha_2(013) \beta_2(123) + \alpha_2 (023) \beta_2(012).
\end{equation}
For arbitrary $\Z_2$-cochains $\alpha$ and $\beta$, the higher cup products satisfy a recursive Leibniz-like identity.
\begin{equation}
    \delta (\alpha \cup_{m} \beta) = \delta \alpha \cup_m \beta + \alpha \cup_{m-1} \delta \beta + \alpha \cup_m \beta + \beta \cup_{m-1} \alpha.
\end{equation}
which agrees with the $\cup_0\equiv \cup$ rule setting $\cup_{-1} \equiv 0$.

Now, we describe the geometric interpretations. A $p$-cochain $\alpha$ in $M^d$ can be thought of as some $(d-p)$-dimensional surface $\alpha^\vee$ living on the dual cellulation. On the cohomology-level, it's well-known that for closed cochains, $\alpha \cup \beta$ is Poincaré-dual to the intersection of the dual submanifold $\alpha^\vee \cap \beta^\vee$. However, the formula above actually describes a certain cochain-level intersection between $\alpha^\vee, \beta^\vee$. In order to define an intersection between $\alpha^\vee, \beta^\vee$, one first needs to consider a \textit{shifted} version of the dual cellulation and also the dual submanifold $\beta^\vee_\text{shifted}$, where the shifting is determined by a vector field associated to the triangulation. Then, $\alpha \cup \beta$ is dual to $\lim_{\epsilon \to 0} \alpha^\vee \cap \beta^\vee_\text{shifted}$, where $\epsilon$ is some parameter determining the shifting. In Fig.~\ref{fig:review_cup0Int}, we illustrate this in two dimensions, and give a similar picture that corresponds to a previous $\cup_0$ definition on the square lattice, which can be interpreted by resolving the square lattice into a certain branched triangulated lattice. 

\begin{figure}[h!]
  \centering
  \includegraphics[width=\linewidth]{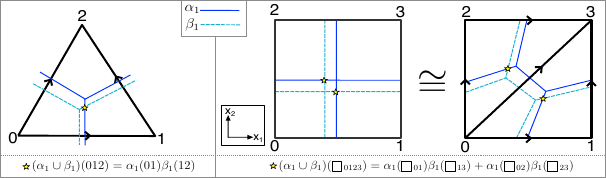}
  \caption[]{The $\cup$ product between $\alpha$ and $\beta$ can be interpreted as measuring the intersection of $\alpha^\vee$ with a shifted version of $\beta^\vee$. This can be used to define $\cup$ on a square lattice by resolving the square lattice into a triangular one.}
  \label{fig:review_cup0Int}
\end{figure}

The $\cup_m$ products are similar, except instead one considers a \textit{thickened} intersection product. In particular, one first thickens $\beta^\vee$ along $m$ vector fields and shifts along an $(m+1)^{th}$. Then, $\alpha \cup_m \beta$ is dual to $\lim_{\epsilon \to 0} \alpha^\vee \cap \beta^\vee_{\text{thickened,shifted}}$. Another way to think about this thickening is to instead locally \textit{project away} the thickening direction and measure the intersections in the projected version. This projection was indeed related to the previous definition \cite{CK18} of $\cup_1$ on a cubic lattice. This interpretation via a projection is depicted in Fig.~\ref{fig:review_cup1Projection}. 

\begin{figure}[h!]
  \centering
  \includegraphics[width=\linewidth]{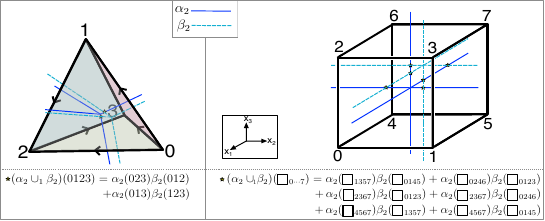}
  \caption[]{The $\cup_1$ product between $\alpha_2$ and $\beta_2$ in three dimensions can be interpreted measuring the intersections between the cochains in $\alpha_2$ on the dual 1-skeleton and $\beta_2$ on a shifted dual 1-skeleton after a certain projection onto 2D. Equivalently, it can be thought of as a thickened intersection, where on thickens $\beta$ in precisely the projection direction mentioned above. }
  \label{fig:review_cup1Projection}
\end{figure}

Although we primarily focus on $\Z_2$-valued cochains, in Sec.~\ref{sec:cochainsOverZ} we talk about how to define $\Z$-valued cochains on a hypercubic lattice and relevant product operations on them, as well how to interpret various signs that show up in the formulas.

\section{Notation for cell structure on $\Z^d$} \label{sec:Notation}

First, let's describe the notation we'll use to define cochains on a hypercube. We'll denote the hypercube in $d$-dimensions as $\msquare_d$. As a subset of $\R^d$, we'll use the following coordinates for the unit hypercube
\begin{equation}
\msquare_d = \{(x_1,\cdots,x_d) \in \R^d | x_i \in [-1,1]\}
\end{equation}
when describing our cochain operations.

\subsection{The Cells}
We will denote the set of $p$-cells of $\msquare_d$ as $F_p(\msquare_d)$. For example, the vertices will be $F_0(\msquare_d)$, and can be written as:
\begin{equation}
F_0(\msquare_d) = \{(\sigma_1, \cdots, \sigma_d) | \sigma_i = \pm 1\}
\end{equation}
We can think of such a label $(\sigma_1 \cdots \sigma_d)$ as equivalently describing a vertex's coordinate in $\R^d$. 

Now let's enumerate the $p$-cells, in $F_p(\msquare_d)$. A $p$-cell will be some hyperplane, spanned by the vectors in the directions $\{i_1 \cdots i_p\}$. But, this hyperplane doesn't completely specify the cell. There are $2^{d-p}$ total cells with the same directions, and they are labeled by which `side' of the center they're on, which is labeled by the other coordinates which do not vary. This is equivalent to a choice of $\pm 1$ for each element of $\{\ihat_1 \cdots \ihat_{d-p}\} := \{1 \cdots n\} \backslash \{i_1 \cdots i_p\}$. As such, we can choose to label the $p$-cells via a choice of $\{i_1 \cdots i_p\} \subset \{1 \cdots d\}$ and signs $\{\pm_{\ihat_1} \cdots \pm_{\ihat_{d-p}}\}$.

We can organize these labels and specify a hypercube as a tuple $(z_1 \cdots z_d)$ of length $d$ with the following entries. We'll say $z_j = \bullet$ if $j$ is a direction of the hyperplane, i.e. if $j \in \{i_1 \cdots i_p\}$. And we'll say $z_j = +$ or $z_j = -$ if $j \in \{\ihat_1 \cdots \ihat_{d-p}\}$, with the sign depending on what the corresponding coordinates of the cell (we'll sometimes write this equivalently as $z_j = \pm 1$). A cell described by $(z_1 \cdots z_d)$ will then be a $p$-cell if $p$ of the entries satisfy $z_i = \bullet$. We'll denote this cell by $P_{(z_1 \cdots z_d)}$. As a subset of $\R^d$, we can write $P_{(z_1 \cdots z_d)}$ as:

\begin{equation}
P_{(z_1 \cdots z_d)} := \{(x_1,\cdots,x_d) \in \R^d | 
x_i \in [-1,1] \text{ if } z_i = \bullet, \text{ or }  x_{\ihat} = z_{\ihat}  \text{ if } z_{\ihat} = \pm 1\}
\end{equation}

For example, we can enumerate the cells of $\msquare_2$ as:

\begin{equation}
\begin{split}
F_0(\msquare_2) &= \{(+,+),(+,-),(-,+),(-,-)\} \\
F_1(\msquare_2) &= \{(\bullet,+),(\bullet,-),(+,\bullet),(-,\bullet)\} \\
F_2(\msquare_2) &= \{(\bullet,\bullet)\}
\end{split}
\end{equation}

\subsubsection{Boundary of a cell}

Now let's describe what the boundary $(p-1)$-cells are of the $p$-cell $P_{(z_1 \cdots z_d)}$. Say $z_j = \bullet$ for exactly the $j \in \{i_1 \cdots i_p\}$. Then, the boundary $\partial P_{(z_1 \cdots z_d)}$ consists of the cells gotten by replacing exactly one $z_j = \bullet$ with either choice of $\tilde{z}_1 = +,-$. In particular,

\begin{equation}
\partial P_{(z_1 \cdots z_d)} = \{P_{(z_1 \cdots z_{j-1}, \tilde{z}_j, z_{j+1} \cdots z_d)} | j \in \{i_1 \cdots i_p\},  \tilde{z}_j \in \{+,-\} \}
\end{equation}

For example, we'd have that (using the notation $(z_1 \cdots z_d) \leftrightarrow P_{(z_1 \cdots z_d)}$)

\begin{equation}
\partial (\bullet,\bullet,+) = \{(\bullet,+,+),(\bullet,-,+),(+,\bullet,+),(-,\bullet,+)\}
\end{equation}

\subsubsection{The Dual Cells}
Now, let's describe the Poincaré dual cells of the cube. The dual cells are in bijection with the original cells, with $p$-cells mapping to $(d-p)$-cells and vice versa. We'll denote the dual of $P_{(z_1 \cdots z_d)}$ as $P^\vee_{(z_1 \cdots z_d)}$. As a subset of $\R^d$, we'll have that
\begin{equation}
P^\vee_{(z_1 \cdots z_d)} := \{(x_1,\cdots,x_d) \in \R^d | x_i = 0 \text{ if } z_i = \bullet,  x_{\ihat} \in [0,1] \text{ if } z_{\ihat} = +,  x_{\ihat} \in [-1,0] \text{ if } z_{\ihat} = -\}
\end{equation}

We will also explicitly parameterize the dual cells as follows, using the notation that $\mathbf{e}_i$ is the $i$th unit vector, with components $(\mathbf{e}_i)_j = \delta_{ij}$. Again denoting $\{\ihat_1 \cdots \ihat_{d-p}\}$ as the coordinates, $j$, with $z_j = \pm 1$,
\begin{equation} \label{dualCellParam}
P^\vee_{(z_1 \cdots z_d)} = \msquare_d \cap \{z_{\ihat_1} t_{\ihat_1} \mathbf{e}_{\ihat_1}  + \cdots +  z_{\ihat_{d-p}} t_{\ihat_{d-p}} \mathbf{e}_{\ihat_{d-p}} | t_{\ihat} \ge 0 \}
\end{equation}
Denote the set of $p$-cells as $F^\vee_p(\msquare_d)$. Then we'll have that, as sets, $F^\vee_p(\msquare_d) \equiv F_{d-p}(\msquare_d)$, and we can use the same labeling scheme to describe them. See Fig.~\ref{fig:cellDecomposition} for depictions in two and three dimensions. 

\begin{figure}[h!]
    \centering
    \begin{minipage}{\textwidth}
        \centering
        \includegraphics[width=\linewidth]{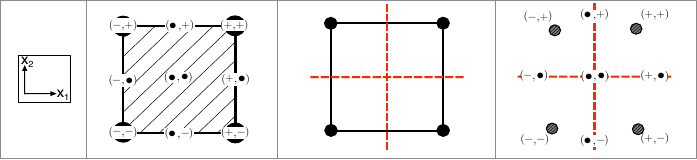}
    \end{minipage}
    \begin{minipage}{\textwidth}
        \centering
        \includegraphics[width=\linewidth]{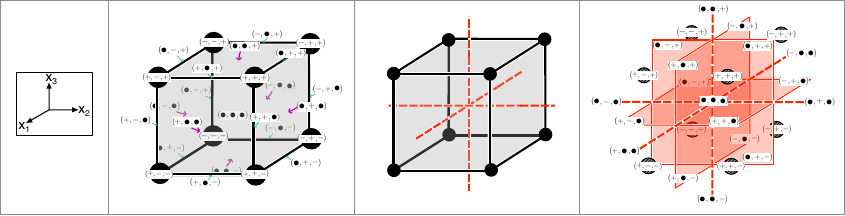}
    \end{minipage}
    \caption{The decomposition and labelings of the cells in a unit hypercube and its dual. The left panels are the original lattice with labels. The middle panels give the lattice (in black) and dual lattice (in red) superimposed on each other. The right panels give the labelings of the dual cells. (Top) Two dimensions. (Bottom) Three dimensions.}
    \label{fig:cellDecomposition}
\end{figure}

\subsection{Cochains}
We'll define a cochain $\alpha \in C^*(\msquare_d,\Z_2) = C^0(\msquare_d,\Z_2) \oplus \cdots \oplus C^d(\msquare_d,\Z_2)$ as the set of all functions from $F_*(\msquare_d) = F_0(\msquare_d) \sqcup \cdots \sqcup F_d(\msquare_d)$ to $\Z_2 = \{0,1\}$. And we'll say that $\alpha \in C^p(\msquare_d,\Z_2)$ for such functions that are only nonzero on $F_p(\msquare_d)$.
For a $p$-cell $P_{(z_1 \cdots z_d)}$ denoted as $\msquare_p$, we define the indicator $p$-cochain $\bmsquare_p$ with value $1$ on the cell $\msquare_p$ and 0 elsewhere.

Now, we define the coboundary operator, $\delta: C^*(\msquare_d) \to C^{*+1}(\msquare_d)$. Just like in simplicial cohomology, the value of $\delta \alpha$ on some $p$-cell $P_{(z_1 \cdots z_d)}$ is the sum of the values of $\alpha$ on the boundary $\partial P_{(z_1 \cdots z_d)}$, i.e.

\begin{equation}
(\delta \alpha)(z_1 \cdots z_d) = \sum_{(\tilde{z}_1 \cdots \tilde{z}_d) \in \partial(z_1 \cdots z_d)} \alpha(\tilde{z}_1 \cdots \tilde{z}_d).
\end{equation}
For example,
\begin{equation}
(\delta \alpha)(\bullet,\bullet,+) = \alpha(\bullet,+,+) + \alpha(\bullet,-,+) + \alpha(+,\bullet,+) + \alpha(-,\bullet,+).
\end{equation}

\section{$\cup_m$ products over $\Z_2$} \label{sec:cupM_OverZ2}

Let's describe the $\cup_k$ products. As described in \cite{T18,T20}, the $\cup_k$ product on cochains can be thought of as a `thickened intersection product' of cells on the dual cellulation. In particular, $\alpha \cup_k \beta$ is dual to $\lim_{\epsilon \to 0}\alpha^\vee \cap \beta^\vee_{\text{thickened,shifted}}$, where $\alpha^\vee,\beta^\vee$ are the duals of $\alpha,\beta$ respectively, $\beta^\vee_{\text{thickened,shifted}}$ is a thickened and shifted version of $\beta$, and $\epsilon$ is some scaling factor for the shifting that indicates we consider this thickened intersection in the limit that the shifting goes to zero.

In~\cite{T18,T20}, this characterization of the $\cup_k$ products was on a triangulation equipped with branching structure, where the branching structure essentially defined the vector fields. In turn, it's quite subtle to adequetely define the vector fields on a triangulation because triangulations and the associated theory of PL manifolds is subtle \footnote{In~\cite{T18}, this thickening and shifting prescription was \textit{conjectured} to reproduce Steenrod's $\cup_k$ formulas using the Halperin-Toledo vector fields~\cite{HalperinToledo}. In ~\cite{T20}, the thickening and shifting prescription was shown using a different set of vector fields defined on the interior of each simplex; explicit schemes to smooth the vector fields were not shown rigorously although it's expected that it can be done in the scope of PL manifold theory.}.
For us in the case of the hypercubic $\Z^d$ lattice embedded in $\R^d$, things are easier because we can simply choose a constant frame of vector fields to thicken and shift along. Similar to the case of triangulations in \cite{T20}, we choose vector fields that form a Vandermonde matrix. 

In the main text, we'll only describe the definitions and some key properties of our $\cup_m$ products. Although, we'll relegate many detailed calculations to the appendices. 

\subsection{Definitions and Examples}
As described above, our definitions of $\alpha \cup \beta$ and $\alpha \cup_m \beta$ were constructed based on what pairs of dual cells intersect each other upon thickening and shifting the dual cells of $\beta$, in the limit as the shifting goes to zero. In particular, the $\cup_m$ product will involve thickening $\beta$ along $m$ vector fields. By dimension-counting, if $\alpha$ a $p$-cochain and $\beta$ a $q$-cochain, then $\alpha^\vee,\beta^\vee$ consist of $(d-p)$-cells,$(d-q)$-cells respectively, and so $\alpha^\vee \cap \beta^\vee_{\text{thickened,shifted}}$ will consist of $(d-p-q+m)$-dimensional submanifolds so we'd want $\alpha \cup_m \beta$ to be a $(p+q-m)$-cochain.
With this, there's actually a subtlety that in the limit as the shifting goes to zero; in particular some of the intersecting cells degenerate to lower-dimensional cells in the limit. We do not consider these degenerate cells and instead define the intersection product in terms of those cells that survive as \textit{full $(d-p-q+m)$-dimensional cells} in the limit. The details of how to derive these formulas is given in Appendix~\ref{app:defininingCupMProducts}.

Consider a $d$-cell $(\cdots \bullet_1 \cdots \bullet_2 \cdots {\dots} \cdots \bullet_d \cdots)$, which consists of $d$ coordinates labeled $\bullet$ and all the rest $\cdots$ are in $\{+,-\}$.
Then we'll have 
\begin{equation*}
    (\alpha \cup_m \beta)(\cdots \bullet_1 \cdots {\dots} \cdots \bullet_d \cdots) = \sum_{(\text{cell-1, cell-2}) \in \mathrm{Int}_m(\cdots \bullet_1 \cdots \bullet_2 \cdots {\dots} \cdots \bullet_d \cdots) } \alpha(\text{cell-1}) \beta(\text{cell-2})
\end{equation*}
which is a a sum over all pairs $(\text{cell-1, cell-2})$ in the set $\mathrm{Int}_m(\cdots \bullet_1 \cdots \bullet_2 \cdots {\dots} \cdots \bullet_d \cdots)$ of pairs of cells whose duals survive as full-dimensional as described above.

Soon, we'll write out the explicit pairs that appear in $\mathrm{Int}(\bullet_1 {\dots} \bullet_d)$ when all coordinates are `$\bullet$'. This is actually all the information we need. In particular, all surviving pairs $(\text{cell-1, cell-2}) \in \mathrm{Int}_m(\cdots \bullet_1 \cdots \bullet_2 \cdots {\dots} \cdots \bullet_d \cdots)$ must have all of the $\cdots = \pm$ coordinates of $(\cdots \bullet_1 \cdots \bullet_2 \cdots {\dots} \cdots \bullet_d \cdots)$ shared by both of `$\text{cell-1}$' and `$\text{cell-2}$'. And, the $\bullet_1 \cdots \bullet_d$ coordinates that are possibly replaced by $\pm$ in $\text{cell-1}$, $\text{cell-2}$ will occur in the same way as they would've been in $\mathrm{Int}_m(\bullet_1 \cdots \bullet_d)$. For example, we'd have:
\begin{equation*}
    (\alpha \cup_m \beta)(\bullet, -, \bullet, \bullet , +) = \sum_{((z_1,z_2,z_3),(z'_1,z'_2,z'_3)) \in \mathrm{Int}_m(\bullet,\bullet,\bullet)} \alpha(z_1,-,z_2,z_3,+)\beta(z'_1,-,z'_2,z'_3,+)
\end{equation*}
Our descriptions of $\mathrm{Int}_m(\bullet_1 \cdots \bullet_d)$ will involve pairs of cells of dimensions $p,q$ with $p+q-m = d$. As such, these pairs correspond to the full product $C^{*_1}(\msquare_d) \times C^{*_2}(\msquare_d) \to C^{*_1 + *_2 - m}(\msquare_d)$ on the graded cochain ring on a single $\msquare_d$. To get a formula for a specific product $C^{p}(\msquare_d) \times C^{q}(\msquare_d) \to C^{p + q - m}(\msquare_d)$, one can restrict to pairs in $\mathrm{Int}_m(\bullet_1 \cdots \bullet_d)$ where $\text{cell-1}$,$\text{cell-2}$ have dimensions $p,q$ respectively.

Now we describe $\mathrm{Int}_m(\bullet_1 \cdots \bullet_d)$. In general, this set will consist of pairs $((z_1 \cdots z_d),(z'_1,\cdots,z'_d))$ such that for $m$ coordinates $1 \le i_1 < \cdots < i_m \le d$ both the $z,z'$ entries are $\bullet$; i.e. $z_{i_1} = z'_{i_1} = \cdots = z_{i_m} = z'_{i_m} = \bullet$. Given these $i_1 < \cdots < i_m$, we'll have that $(z_j,z'_j)$ for $j \not\in \{i_1, \cdots, i_m\}$ depend on the position of $j$ relative to the $i_1 < \cdots < i_m$. First, defining $i_0 := 0$ and $i_{d+1} = d+1$, we have that there is a unique $\ell(j) \in \{0, \cdots , d\}$ for which $i_{\ell(j)} < j < i_{\ell(j)+1}$. Given this, we'll have that $(z_j,z'_j) \in \{((-1)^{\ell(j)} + , \bullet), (\bullet, (-1)^{\ell(j)}-)\}$ where we mean ``$(-1)\cdot + = -$'' and ``$(-1) \cdot - = +$''. In short, the definition is:
\begin{defn}
    \begin{equation}
    \begin{split}
        \mathrm{Int}_m(\bullet_1 \cdots \bullet_d) = \bigcup\limits_{\{1 \le i_1 < \cdots i_m \le d\}} \bigg\{ ((z_1 \cdots z_d),(z'_1 \cdots z'_d)) \bigg\vert & z_{i_1} = z'_{i_1} = \cdots = z_{i_m} = z'_{i_m} = \bullet, \text{ and } \\
        & \text{ each } (z_j,z'_j) \in \{((-1)^{\ell(j)} + , \bullet), (\bullet, (-1)^{\ell(j)}-)\} \bigg\},
    \end{split}
    \end{equation}
    and 
    \begin{equation}
        (\alpha \cup_m \beta)(\bullet_1 \cdots \bullet_d) = \sum_{((z_1,\cdots,z_d),(z'_1,\cdots,z'_d)) \in \mathrm{Int}_m(\bullet_1 \cdots \bullet_d) } \alpha(z_1,\cdots,z_d) \beta(z'_1,\cdots,z'_d)
    \end{equation}
\end{defn}
In general, $(\alpha \cup_0 \beta)(z_1 \cdots z_d)$ will have $2^k$ terms, 
where $k$ is the number of elements in the set  $S \equiv \{i | z_i = \bullet \}$. For each $i \in S$, there are two choices for the arguments in $\alpha$ and $\beta$: $\alpha(\cdots, \bullet_i, \cdots) \beta (\cdots, -_i, \cdots)$ or $\alpha(\cdots, +_i, \cdots) \beta (\cdots, \bullet_i, \cdots)$
Similarly, $(\alpha \cup_m \beta)(\bullet_1 \cdots \bullet_d)$ will have ${d \choose m} \cdot 2^{d-m}$ terms which correspond to those choices $\{i_1 < \cdots < i_m\}$ and the two choices of $(z_j,z'_j)$ for $j \not\in \{i_1 \cdots i_m\}$. 
Although, only some of these terms appear in the restriction to each $C^p \times C^q \to C^{p+q-m}$. 

Now we spell out a few examples, starting with $\cup_0$. First in $d=3$, we'll have that: 
\begin{equation*}
\begin{split}
(\alpha \cup_0 \beta)(\bullet, \bullet, \bullet) = & \alpha(\bullet,\bullet,\bullet)\beta(-,-,-) + \alpha(\bullet,\bullet,+)\beta(-,-,\bullet) + \alpha(\bullet,+,\bullet)\beta(-,\bullet,-) + \alpha(+,\bullet,\bullet)\beta(\bullet,-,-) \\
+ & \alpha(\bullet,+,+)\beta(-,\bullet,\bullet) + \alpha(+,\bullet,+)\beta(\bullet,-,\bullet) + \alpha(+,+,\bullet)\beta(\bullet,\bullet,-) + \alpha(+,+,+)\beta(\bullet,\bullet,\bullet) \\
\end{split}
\end{equation*}
Note that every term $\alpha(z'_1,z'_2,z'_3)\beta(z''_1,z''_2,z''_3)$ above has exactly one of $z'_k$ or $z''_k$ equalling $\bullet$ for $k=1,2,3$. And, if $z'_k = \bullet$ then $z''_k = -$, and if $z''_k = \bullet$ then $z'_k = +$. Another example is:
\begin{equation*}
(\alpha \cup_0 \beta)(\bullet, \bullet, \pm) = \alpha(\bullet,\bullet,\pm)\beta(-,-,\pm) + \alpha(\bullet,+,\pm)\beta(-,\bullet,\pm) + \alpha(+,\bullet,\pm)\beta(\bullet,-,\pm) + \alpha(+,+,\pm)\beta(\bullet,\bullet,\pm)
\end{equation*}
where $\pm$ refers to either choice of $+$ or $-$. Note for any term $\alpha(z'_1,z'_2,z'_3)\beta(z''_1,z''_2,z''_3)$, we have that $z'_3 = z''_3 = \pm$, matching with the third term of the argument of $(\alpha \cup_0 \beta)(z_1,z_2,z_3)$. And since $z_1 = z_2 = \bullet$, the same pattern as the previous example works. That exactly one of $z'_k$ or $z''_k$ is $\bullet$ for $k=1,2$. And if $z'_k = \bullet$ then $z''_k = -$, and if $z''_k = \bullet$ then $z'_k = -$. 

As an example of the $\cup_1$ product, in $d=3$, we'll have:
\begin{equation*}
\begin{split}
(\alpha \cup_1 \beta)(\bullet, \bullet, \bullet) &
=  \alpha(\bullet,\bullet,\bullet)\beta(-,-,\bullet) + \alpha(\bullet,+,\bullet)\beta(-,\bullet,\bullet) + \alpha(+,\bullet,\bullet)\beta(\bullet,-,\bullet) + \alpha(+,+,\bullet)\beta(\bullet,\bullet,\bullet) \\
&+ \alpha(\bullet,\bullet,\bullet)\beta(-,\bullet,+) + \alpha(\bullet,\bullet,-)\beta(-,\bullet,\bullet) + \alpha(+,\bullet,\bullet)\beta(\bullet,\bullet,+) + \alpha(+,\bullet,-)\beta(\bullet,\bullet,\bullet) \\
&+ \alpha(\bullet,\bullet,\bullet)\beta(\bullet,+,+) + \alpha(\bullet,\bullet,-)\beta(\bullet,+,\bullet) + \alpha(\bullet,-,\bullet)\beta(\bullet,\bullet,+) + \alpha(\bullet,-,-)\beta(\bullet,\bullet,\bullet) 
\end{split}
\end{equation*}

See Fig.~\ref{fig:cupProductImages} for the geometric depiction of some of these examples.

\begin{figure}[h!]
  \centering
  \includegraphics[width=\linewidth]{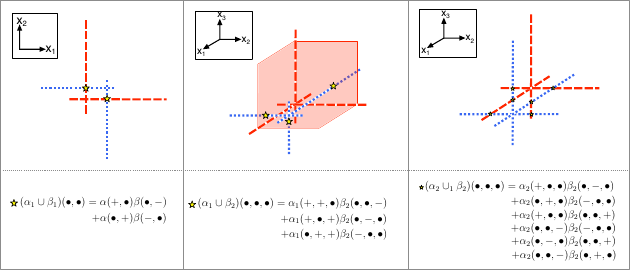}
  \caption[]{Geometric depictions of (Left) the $\cup_0$ product between 1-cochains in two dimensions, (Middle) the $\cup_0$ product between a 2-cochain and a 1-cochain in three dimensions, and (Right) the $\cup_1$ product of 2-cochains in three dimensions.} 
  \label{fig:cupProductImages}
\end{figure}

\subsection{Diagrammatics}

Now, we'll introduce a convenient way to diagrammatically represent the $\cup_m$ products. These diagrammatics will be helpful in establishing the identity $\delta (\alpha \cup_m \beta) = \delta\alpha \cup_m \beta + \alpha \cup_m \delta\beta + \alpha \cup_{m-1} \beta + \beta \cup_{m-1} \alpha$. \footnote{These diagrams are inspired by those presented in a lecture of John Morgan~\cite{morganHigherCupLecture} on simplicial higher cup products.}

First, we'll only need to care about the case of $(\alpha \cup_0 \beta)(\bullet \cdots \bullet)$ where all the arguments are all $\bullet$. This is because if any other arguments are $+$ or $-$, then the corresponding arguments for $\alpha$ and $\beta$ are fixed and equal to the $+$ or $-$. The diagrammatics are summarized in Fig.~\ref{fig:cup_mDiagrams}. 

\begin{figure}[h!]
    \centering
    \includegraphics[width=\linewidth]{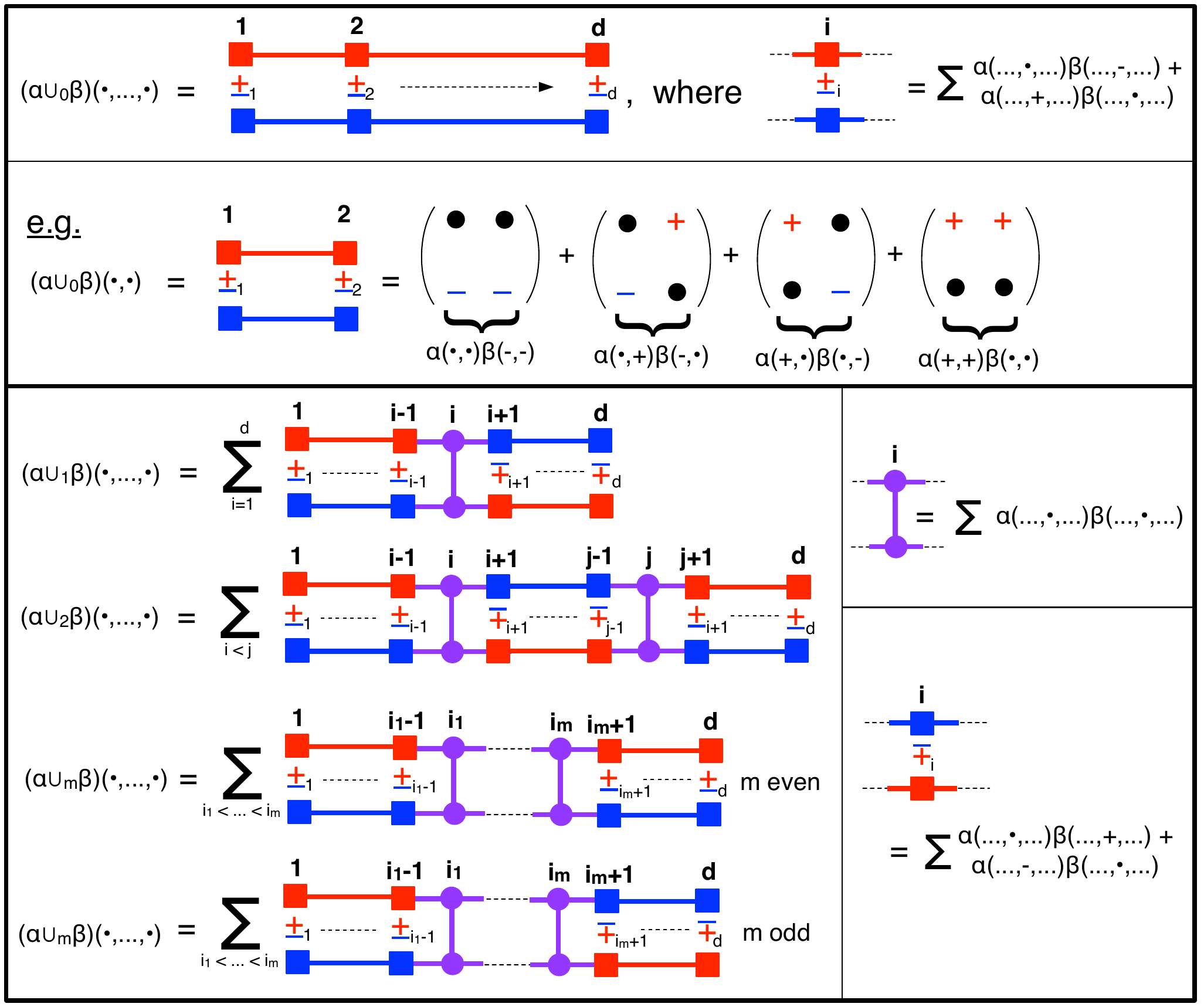}
    \caption{Diagrammatic representations of: (top) $\alpha \cup_0 \beta$ (Bottom) $\alpha \cup_m \beta$}
    \label{fig:cup_mDiagrams}
\end{figure}

In the Appendix~\ref{app:proofOfHigherCupId}, we show how to think about the diagrammatics with respect to the coboundary operators $\delta$ and show the fundamental cup and higher cup product identities:
\begin{prop} \label{prop:cup0_coboundary}
    $\delta(\alpha \cup \beta) = \delta \alpha \cup \beta + \alpha \cup \delta \beta$
\end{prop}

\begin{prop} \label{prop:cupM_coboundary}
    $\delta(\alpha \cup_{m} \beta) = \delta \alpha \cup_m \beta + \alpha \cup_m \delta \beta + \alpha \cup_{m-1} \beta + \beta \cup_{m-1} \alpha$.
\end{prop}

\section{Cochain operations over $\Z$} \label{sec:cochainsOverZ}
Here, we describe how to extend the definition of the boundary operator $\delta$ and the $\cup_0$ product to cochains over $\Z$.

\subsection{$\delta$ over $\Z$}
First, we define the coboundary operator $\delta$ acting on a $(d-1)$-cochain $\alpha$ as
\begin{equation}
(\delta \alpha)(\bullet_1, \cdots, \bullet_d) = \sum_{\substack{\ell=1, \\ \pm \in \{+,-\}}}^d \mp(-1)^\ell \alpha(\bullet, \cdots \underbrace{,\pm,}_{\substack{\ell^\text{th} \\ \text{coord.}}} \cdots, \bullet)
\end{equation}
where $\mp = - \pm$. In general for an $n$-cochain, $\delta\alpha$ can be defined similarly to the $\Z_2$ on $(\delta \alpha)(\cdots,\bullet_{i_1},\cdots,\bullet_{i_2},\cdots,\bullet_{i_n},\cdots)$ where there are $n$ coordinates $\{\bullet_{i_k}\}_{k=1}^n$ and $(d-n)$ coordinates $\pm$. Then we can define $\delta$ by keeping all the $\pm$ coordinates fixed and only altering the $\bullet$ coordinates one at a time 
\begin{equation} \label{coboundaryDefinitionOverZ}
(\delta \alpha)(\cdots,\bullet_{i_1},\cdots,\bullet_{i_2},\cdots,\bullet_{i_n},\cdots) = \sum_{\substack{\ell=1, \\ \pm \in \{+,-\}}}^n 
\mp (-1)^\ell \alpha(\cdots \underbrace{,\pm,}_{\substack{{i_\ell}^\text{th} \\ \text{coord.}}}\cdots)
\end{equation}
and the sign $\mp(-1)^\ell$ in front of each term only depends on $\pm$ and the coordinate $\ell$ that defines the position of $\bullet_{i_\ell}$ with respect to all the other $\bullet$ positions. Note that the equation
\begin{equation}
\delta \delta \alpha = 0
\end{equation}
will follow from essentially same reasoning as in the simplicial case. 

\subsection{$\cup_0 := \cup$ over $\Z$}
Now we'll define a cup-product operation $\cup_0 := \cup$ over $\Z$. 
It satisfies the following cochain-level Leibniz-rule, with respect to the $n$-cochain $\alpha^n$ and $m$-cochain $\beta^m$
\begin{equation*}
\delta(\alpha^n \cup \beta^m) = (\delta \alpha^n) \cup \beta^m + (-1)^n \alpha^n \cup (\delta \beta^m).
\end{equation*}

Now, we restrict our attention to case where $n+m=d$ to define the $\cup$ product. The other cases where $n+m < d$ can again be specified by keeping the signs fixed away from $(n+m)$ $\bullet$-coordinates in $(\alpha \cup \beta)(\cdots,\bullet_{i_1},\cdots,\bullet_{i_2},\cdots,\bullet_{i_n},\cdots)$ as in all previous discussion. Now define
\begin{equation}
(\alpha^n \cup \beta^{d-n})(\bullet_1, \cdots, \bullet_d) = \sum_{\substack{(z_1 \cdots z_d), (z'_1 \cdots z'_d) \\ z_i = + \text{ and } z'_i = \bullet \text{, OR}\\ z_i = \bullet \text{ and } z'_i = -}} (-1)^s \alpha(z_1 \cdots z_d) \beta(z'_1 \cdots z'_d)
\end{equation}
where the sign $(-1)^s$ will be defined in a moment. First note that the (mod 2) reduction of this exactly agrees with our formulas and diagrammatics over $\Z_2$. To define $(-1)^s$ on the term associated to $(z_1 \cdots z_d),(z'_1 \cdots z'_d)$, we first define $a_1 \cdots a_{d-n}$ as the $z_i$ that are labeled $+$ and define $b_1 \cdots b_n$ as the $z'_i$ that are labeled $-$. Then, we can define $(-1)^s$ as 
\begin{equation}
(-1)^s = \underbrace{(-1)^{d-n}}_{\substack{-1 \text{ for each} \\ + \text{ label}}} \text{sgn}
\begin{pmatrix}
1   & \cdots &  \cdots & \cdots & \cdots & d \\
b_1 & \to    & b_n     & a_1    & \to    & a_{d-n}
\end{pmatrix}
\cdot (-1)^{m}
\end{equation}
with respect to the sign of the permutation taking $(1 \to d)$ to $(b_1 \cdots b_n ; a_1 \cdots a_{d-n})$. 

We note that this sign is closely related to a \textit{signed intersection} between the cell dual to $\beta$ and the cells dual to $\alpha$. In particular, the dual cells corresponding to $(z_1 \cdots z_d),(z'_1 \cdots z'_d)$ will extend out in the $b_1, \dots, b_n$ directions for $\beta$'s dual and in the $a_1,\dots,a_{d-n}$ directions for $\alpha$. As such, the above sign can roughly be thought of as representing a ``$\beta^\text{dual} \cap \alpha^\text{dual}$" with signs in analog to how the $\cup$ product in cohomology is dual to signed intersections of the dual submanifolds. 

We state the Leibniz rule.
\begin{prop} \label{prop:cup0_coboundary_overZ}
$\delta(\alpha^n \cup \beta^m) = (\delta \alpha^n) \cup \beta^m + (-1)^n \alpha^n \cup (\delta \beta^m)$ where $\alpha^n$ and $\beta^m$ are $n$- and $m$-cochains respectively.
\end{prop}
The proof is given in Appendix~\ref{app:proofOfCupM_OverZ} together with an analogous one for the $\cup_k$ products, defined below.

\subsection{$\cup_k$ over $\Z$}
In general, a similar pattern works for the higher cup operations over $\Z$ which we'll call the $\cup_k$ products because of the minor differences in the formulas as compared to Steenrod's formulas. We give the definition for $\alpha^{n+k} \cup_k \beta^{d-n}$ analogous to our $\Z_2$ definitions
\begin{equation} \label{cupK_OverZ_Formula}
(\alpha^{n+k} \cup_k \beta^{d-n})(\bullet_1, \cdots, \bullet_d) := \sum_{1 =: i_0 \le i_1 < \cdots < i_k \le i_{k+1} := d} \sum_{\substack{(z_1 \cdots z_d), (z'_1 \cdots z'_d) \\ \text{each } z_{i_1} = z'_{i_1} = \cdots = z_{i_k} = z'_{i_k} = \bullet \\ z_j = (-1)^{\ell(j)}+ \text{ and } z'_j = \bullet \text{, OR}\\ z_j = \bullet \text{ and } z'_j = (-1)^{\ell(j)}- \\ \text{where } i_{\ell(j)} < j < i_{\ell(j)+1}}} (-1)^s \alpha(z_1 \cdots z_d) \beta(z'_1 \cdots z'_d)
\end{equation}
Here the sign $(-1)^s$ on each depends on $(i_1 \cdots i_k)$ and $(z_1 \cdots z_d),(z'_1 \cdots z'_d)$. Again we can define $b_1 < \cdots < b_{n}$ as the coordinates in $z'_1 \cdots z'_d$ that are labeled by $\pm$. And define $a_1 \cdots a_{d-n-k}$ as the coordinates in $z_1 \cdots z_d$ labeled as $\pm$. Then we'll set
\begin{equation}
(-1)^s = (-1)^{\text{\# of total + signs}} \text{sgn}
\begin{pmatrix}
1   & \cdots & \cdots & \cdots & \cdots & \cdots & \cdots & \cdots & d \\
i_1 & \to    & i_k    & b_1    & \to    & b_{n}  & a_1    & \to    & a_{d-n-k}
\end{pmatrix}
\cdot (-1)^{k+(d-n) + (n-k-1)k}
\end{equation}
where the ``$\text{\# of total + signs}$'' is the total number of $+$ coordinates among both $(z_1 \cdots z_d), (z'_1 \cdots z'_d)$. This can also be interpreted as a signed intersection of $\beta^\vee_\text{thickened,shifted} \cap \alpha^\vee$.
Then we'll have the cochain-level identity:
\begin{prop} \label{prop:cupM_coboundary_overZ}
Let $\alpha$ be an $(n-1+k)$-cochain and $\beta$ be an $(d-n)$-cochain. Then
\begin{equation}
\begin{split}
\delta(\alpha^{n-1+k} \cup_k \beta^{d-n}) = \delta\alpha \cup_k \beta + (-1)^{n-1+k} \alpha \cup_k \delta\beta + (-1)^{d-1} \alpha \cup_{k-1} \beta + (-1)^{(n-1+k)(d-n)-1+k+d}\beta \cup_{k-1} \alpha.
\end{split}
\end{equation}
\end{prop}
Note that defining $p=n-1+k$ and $q=d-n$, the above product is exactly the recursion relation $\delta(\alpha^{p} \cup_k \beta^{q}) = \delta\alpha\cup_k \beta + (-1)^{p} \alpha \cup_k \delta\beta + (-1)^{p+q-k} \alpha \cup_{k-1} \beta + (-1)^{pq+p+q}\beta \cup_{k-1} \alpha$ of Steenrod~\cite{S47}.
See Appendix~\ref{app:proofOfCupM_OverZ} for a proof.

\section{Fermions: Exact Bosonziation and Grassmann Integral} \label{sec:FermionsAndExactBosonization}
We will now apply the $\Z_2$ higher-cup formalism to express lattice models of fermions on a hypercubic lattice. 
We first `bosonize' fermionic operator algebras that are defined on a hypercubic lattice, generalizing the constructions of~\cite{chen2020BosonArbDimensions}. The paper~\cite{chen2020BosonArbDimensions} considered a triangulated $d$-manifold where there was one copy of fermion operators $\gamma_{\Delta_d}, \gamma'_{\Delta_d}$ for each $d$-simplex. It was shown that one can express all even-fermion-parity operators in that algebra in terms of an algebra of bosonic Pauli operators $X_{\Delta_{d-1}},Z_{\Delta_{d-1}}$ associated to each $(d-1)$-simplex of the triangulation in the presence of an additional gauge constraint. The bosonization map used the $\cup_m$ products on a triangulation. We'll show here that essentially the same map holds on a hypercubic lattice if one considers the hypercubic $\cup_m$ products defined here. 

After this, we will also extend the `Gu-Wen Grassmann Integral'~\cite{guWen2014Supercohomology} $\sigma(\alpha)$ to hypercubic lattices, which can be used in space-time path-integral constructions of fermionic phases of matter. We review two definitions of $\sigma(\alpha)$, via winding numbers on winding numbers on curves and via Grassmann variables.

In Sec.~\ref{sec:fermionicToricCode}, we will define a `fermionic toric code' model in arbitary dimensions on the hypercubic lattice, similar to~\cite{chen2020BosonArbDimensions}. This will use the exact bosonization map and also have a close relation to the Grassmann integral.

\subsection{Exact Bosonization}
Now, let's spell out the exact bosonization map on the hypercubic lattice $\Z^d$, which corresponds to maps between $(d+1)$-dimensional Hamiltonians. We'll be given a system of fermions generated by $\gamma_{\msquare_d}, \gamma'_{\msquare_d}$ living at the centers of the $d$-dimensional hypercubes (or equivalently on the vertices of the dual lattice). And on the bosonic side, we'll have Pauli matrices $X_{\msquare_{d-1}},Z_{\msquare_{d-1}}$ living on $(d-1)$-simplices. The duality is:
\begin{align}
    W_{\msquare_d} := \prod_{\msquare_{d-1} \subset \msquare_d} Z_{\msquare_{d-1}} 
    &\longleftrightarrow P_{\msquare_d} := -i \gamma_{\msquare_d} \gamma'_{\msquare_d}  \label{eq:exactbosonization1}\\
    {U}_{\msquare_{d-1}} := X_{\msquare_{d-1}} 
    \left( \prod_{\msquare'_{d-1}} Z_{\msquare_{d-1}}^{\int {\bmsquare}'_{d-1} \cup_{d-2} {\bmsquare}_{d-1}} \right)
    &\longleftrightarrow S_{\msquare_{d-1}} := i \gamma_{L(\msquare_{d-1})} \gamma'_{R(\msquare_{d-1})} \label{eq:exactbosonization2}\\
    G_{\msquare_{d-2}} := \underbrace{
    \prod_{\msquare_{d-1} \supset \msquare_{d-2}} X_{\msquare_{d-1}} 
    \left( \prod_{\msquare'_{d-1}} Z_{\msquare'_{d-1}}^{\int \delta\bmsquare_{d-2} \cup_{d-2} \bmsquare'_{d-1}} \right)}_{\text{Set equal to } 1 \text{ as gauge constraint}} = 1
    &\longleftrightarrow \underbrace{S_{\delta \bmsquare_{d-2}} \prod_{\msquare_d} P_{\msquare_d}^{\int \bmsquare_{d} \cup_{d-1} \delta\bmsquare_{d-2} }}_{\text{Equals } 1 \text{ as an identity}} = 1. \label{eq:exactbosonization3}\\
\end{align}
The first two Eqs.~(\ref{eq:exactbosonization1},\ref{eq:exactbosonization2}) above are identities between the operator algebras, with the $P_{\msquare_d}$ representing fermion parity inside $\msquare_d$ and $S_{\msquare_{d-1}}$ being a hopping operator across $\msquare_{d-1}$. The last Eq.~\eqref{eq:exactbosonization3} gives the required gauge constraint on the bosonic LHS while the fermionic RHS is an identity that always holds (that we'll show soon).
The rest of this section will demonstrate this bosonization, although we defer many details to~\cite{chen2020BosonArbDimensions} because the proofs are similar. We start by explaining some notation above, in particular what we mean by $L(\msquare_{d-1}),R(\msquare_{d-1})$ and $S_\lambda$ for a general cochain $\lambda$.

First, given a $(d-1)$-hypercube $\msquare_{d-1}$, there are two $d$-hypercubes $L(\msquare_{d-1}),R(\msquare_{d-1})$ that it's adjacent to. 
Supposing $\msquare_{d-1}$ is in the $\{1 \cdots \ihat \cdots d\}$ hyperplane, the two adjacent $\msquare_d$'s will be in the two $\pm x_{\ihat}$ directions relative to $\msquare_{d-1}$, which we'll call ${^+}(\msquare_{d-1}),{^-}(\msquare_{d-1})$ respectively. We'll define
\begin{equation} \label{eq:LR_assignment_dMinus1Cubes}
    \big( L(\msquare_{d-1}) , R(\msquare_{d-1}) \big) = 
    \begin{cases}
        \big( {^-}(\msquare_{d-1}),{^+}(\msquare_{d-1}) \big) \, \text{ if } i+d = 1 \text{ (mod 2)} \\
        \big( {^+}(\msquare_{d-1}),{^-}(\msquare_{d-1}) \big) \, \text{ if } i+d = 0 \text{ (mod 2)} \\
    \end{cases}
\end{equation}
which is depicted in Fig.~\ref{fig:majoranaLR_assignment}.

\begin{figure}[h!]
  \centering
  \includegraphics[width=0.8\linewidth]{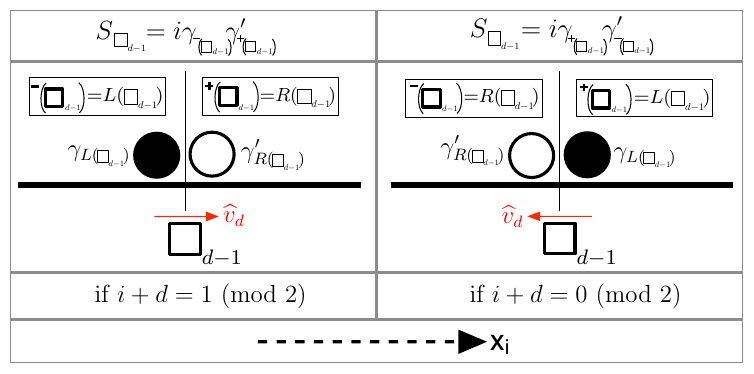}
  \caption[]{Definitions of $L(\msquare_{d-1})$ and $R(\msquare_{d-1})$ depending on $i+d \text{ (mod 2)}$. The $+x_i$ direction is to the right and determine the definitions of ${^\pm}(\msquare_{d-1})$. The $\gamma,\gamma'$ are given be the solid-black,open-white circles respectively. A useful reference is the red arrow labeled $\widehat{v}_d$ is the direction of the vector field $\widehat{v}_d$ that's parallel to the dual 1-skeleton, as defined in the Appendix~\ref{app:framingsOfCurvesAndWindings}; it is used to define the left/right sides of $\msquare_{d-1}$.} 
  \label{fig:majoranaLR_assignment}
\end{figure}

Now, to define $S_{\lambda}$ for a general $(d-1)$-cochain $\lambda$, we will need the following equality of the commutation relations of the $S_{\msquare_{d-1}}$ and ${U}_{\msquare_{d-1}}$ which match:
\begin{lem} \label{lem:exactBosonization_S_commRels}
\begin{equation}
    S_{\msquare_{d-1}} S_{\msquare'_{d-1}} = S_{\msquare'_{d-1}} S_{\msquare_{d-1}} (-1)^{\int \bmsquare_{d-1} \cup_{d-2} \bmsquare'_{d-1} + \bmsquare'_{d-1} \cup_{d-2} \bmsquare_{d-1}}
\end{equation}
and
\begin{equation}
    {U}_{\msquare_{d-1}} {U}_{\msquare'_{d-1}} = {U}_{\msquare'_{d-1}} {U}_{\msquare_{d-1}} (-1)^{\int \bmsquare_{d-1} \cup_{d-2} \bmsquare'_{d-1} + \bmsquare'_{d-1} \cup_{d-2} \bmsquare_{d-1}}
\end{equation}
\end{lem}
Note that the second equality for ${U}$ follows immediately from their definition, so we only need to check the first equality for the $S$. We save the proof to the end of this subsection.

Assuming this lemma, we can make a definition of $S_\lambda$ for general $(d-1)$-cochains $\lambda$. Say that $\lambda = \lambda_1 + \cdots + \lambda_N$ where each $\lambda_i$ is an indicator cochain on a single $\msquare_{d-1}$. Then, we'll define
\begin{equation}
    S_{\lambda} := (-1)^{\sum_{1 \le i < j \le N} \int \lambda_i \cup_{d-2} \lambda_j } \prod_{i=1}^N S_{\lambda_i}.
\end{equation}
where the products of the $\{S_{\lambda_i}\}$ are in the order $1 \to \cdots \to N$. The commutation relations of the $S$ as in Lemma~\ref{lem:exactBosonization_S_commRels} show that this definition is independent of the order of the $\lambda_i$, thus is well-defined.

Now, we discuss the fermionic identity on the RHS of Eq.~\eqref{eq:exactbosonization3}, given by the following lemma.
\begin{lem} \label{lem:exactBos_gaugeConstraintFermion}
For all $\msquare_{d-2}$, we have 
\begin{equation}
    S_{\delta \bmsquare_{d-2}} = \prod_{\msquare_d} P_{\msquare_d}^{\int \bmsquare_{d} \cup_{d-1} \delta\bmsquare_{d-2}} 
\end{equation}
\end{lem}
Again, we save the proof to the end of this subsection.

These are the main calculations needed to verify the exact bosonization. The rest of the arguments verifying the equalities of operators algebras carry over \textit{mutatis mutandi} from~\cite{chen2020BosonArbDimensions}. 

\subsubsection{Proofs of Lemmas~\ref{lem:exactBosonization_S_commRels},\ref{lem:exactBos_gaugeConstraintFermion}}
Now we discuss the proofs of Lemmas~\ref{lem:exactBosonization_S_commRels},\ref{lem:exactBos_gaugeConstraintFermion}. Before starting the proofs, it's helpful to state a formula for the $\cup_{d-2}$ product of two $(d-1)$-cochains. We'll use the notation for each choice of $\{i,\pm\}$
\begin{equation}
\ihat_{\pm} := P_{(\bullet, \cdots,\bullet, \underbrace{\pm}_\text{position `i'}, \bullet, \cdots, \bullet)}
\end{equation}
to denote both the $(d-1)$-hypercube above and also as an argument for cochains.\footnote{In the Appendix~\ref{app:framingsOfCurvesAndWindings}, we often use this to also refer to the \textit{dual} 1-cell.} The following formula holds
\begin{equation} \label{hypercubic_cup_dminus2_formula}
(\alpha \cup_{d-2} \beta)(\msquare_d) = \sum_{1\le i<j \le d}(\alpha(\ihat_{(-1)^{i+1}})\beta(\jhat_{(-1)^{j+1}})+\beta(\ihat_{(-1)^{i}})\alpha(\jhat_{(-1)^{j}}))
\end{equation}
where $\ihat_{\pm 1}, \jhat_{\pm 1}$ refer to the indices on the cells $\ihat_{\pm}, \jhat_{\pm}$. This can be shown using the $\cup_m$ formulas presented previously. For example, in $d=4$ we have
\begin{equation*}
\begin{split}
    (\alpha \cup_{d-2} \beta)(\msquare_{d}) = 
    &\alpha(\hat{1}_+)\beta(\hat{2}_-) + \alpha(\hat{1}_+)\beta(\hat{3}_+) + \alpha(\hat{1}_+)\beta(\hat{4}_-) + \alpha(\hat{2}_-)\beta(\hat{3}_+) + \alpha(\hat{2}_-)\beta(\hat{4}_-) + \alpha(\hat{3}_+)\beta(\hat{4}_-) \\
    + &\beta(\hat{1}_-)\alpha(\hat{2}_+) + \beta(\hat{1}_-)\alpha(\hat{3}_-) + \beta(\hat{1}_-)\alpha(\hat{4}_+) + \beta(\hat{2}_+)\alpha(\hat{3}_-) + \beta(\hat{2}_+)\alpha(\hat{4}_+) + \beta(\hat{3}_-)\alpha(\hat{4}_+)
\end{split}
\end{equation*}
and for all $d < 4$, the expression is simply the restriction to terms not involving $\ihat$ with $i > d$. The analogous pattern holds for all $d$. 

Now we prove the Lemma~\ref{lem:exactBosonization_S_commRels}.
\begin{proof}[Proof of Lemma~\ref{lem:exactBosonization_S_commRels}]
    Write 
    $S_{\msquare_{d-1}} = i \gamma_A \gamma'_B$ and $S_{\msquare'_{d-1}} = i \gamma_C \gamma'_D$
    where $A=L(\msquare_{d-1}),B=R(\msquare_{d-1})$ and $C=L(\msquare'_{d-1}),D=R(\msquare'_{d-1})$ are the $d$-hypercubes that $S_{\msquare_{d-1}}$ and $S_{\msquare'_{d-1}}$ respectively border. Note that the only ways they could possibly anticommute is if $A=C$ or if $B=D$. 
    
    The simplest case to check is when $\msquare_{d-1},\msquare'_{d-1}$ \textit{are not} contained in any common $d$-hypercube. Then $S_{\msquare_{d-1}},S_{\msquare'_{d-1}}$ must commute because then all $A,B,C,D$ would be different. Also we'd have 
    \begin{equation}
        \int \bmsquare_{d-1} \cup_{d-2} \bmsquare'_{d-1} + \bmsquare'_{d-1} \cup_{d-2} \bmsquare_{d-1} = 0,
    \end{equation}
    which is consistent with the commutation of $S_{\msquare_{d-1}},S_{\msquare'_{d-1}}$.
    Another simple case to check is that two $(d-1)$-hypercubes are identical, $\msquare_{d-1} = \msquare'_{d-1}$. We have
    \begin{eqs}
        \int \bmsquare_{d-1} \cup_{d-2} \bmsquare'_{d-1} + \bmsquare'_{d-1} \cup_{d-2} \bmsquare_{d-1} = 0,
    \end{eqs}
    and commuting $S_{\msquare_{d-1}}$ and $S_{\msquare'_{d-1}}$.
    
    So suppose WLOG that $\msquare_{d-1},\msquare'_{d-1}$ are both distinct subcells of $\msquare_d$ and that we can write $\msquare_{d-1} = \ihat_{\pm_1}$ and $\msquare'_{d-1} = \jhat_{\pm_2}$. First, the $\cup_{d-2}$ formula Eq.~\eqref{hypercubic_cup_dminus2_formula} gives that
    \begin{equation}
        \int \bmsquare_{d-1} \cup_{d-2} \bmsquare'_{d-1} + \bmsquare'_{d-1} \cup_{d-2} \bmsquare_{d-1} = 1 \text{ iff } 
        \begin{cases}
            i-j=0\text{ (mod 2) and } \pm_1 = \pm_2, \text{ OR} \\
            i-j=1\text{ (mod 2) and } \pm_1 = -\pm_2
        \end{cases}
    \end{equation}
    One can check using the characterization of $\{L,R\}(\{\msquare_{d-1},\msquare'_{d-1},\})$ from Eq.~\eqref{eq:LR_assignment_dMinus1Cubes} that these conditions also hold iff one of $A = L(\msquare_{d-1}) = L(\msquare'_{d-1}) = C = \msquare_d$ or $B = R(\msquare_{d-1}) = R(\msquare'_{d-1}) = D = \msquare_d$ hold, which verifies the commutation relation. 
    
    Note that whether $S_{\msquare_{d-1}},S_{\msquare'_{d-1}}$ anticommute commute can be expressed in terms of the vector $\widehat{v}_d$, which was depicted in Fig.~\ref{fig:majoranaLR_assignment} and described in Appendix~\ref{app:framingsOfCurvesAndWindings}. In particular, they anticommute iff on the common $d$-hypercube of $\msquare_{d-1},\msquare'_{d-1}$, the $\widehat{v}_d$ vectors both point towards the center or away from the center, like ${\color{red} \leftarrow\rightarrow}$ or ${\color{red} \rightarrow\leftarrow}$.
\end{proof}

\begin{proof}[Proof of Lemma~\ref{lem:exactBos_gaugeConstraintFermion}]
    Let's say that $\msquare_{d-2}$ is perpendicular to the directions $x_i,x_j$, so that it spans the directions $\{x_1 \cdots \hat{x}_i \cdots \hat{x}_j \cdots x_d\}$. Note that $\delta\bmsquare_{d-2}$ is nonzero on exactly four $(d-1)$-hypercubes, in the $\pm\{i,j\}$ directions of $\msquare_{d-2}$. As such, $\delta\bmsquare_{d-2}$ is dual to a square going around in the $i,j$ directions around $\msquare_{d-2}$.
    
    Now we want to handle the $\int \bmsquare_{d} \cup_{d-1} \delta\bmsquare_{d-2}$ factor of the integral. One can verify that for any $(d-1)$-cochain $\alpha$, we'd have
    \begin{equation}
        \int \bmsquare_d \cup_{d-1} \alpha = \sum_{k=1}^d \alpha(\hat{k}_{(-1)^{k}})
    \end{equation}
    where $\alpha(\hat{k}_{\pm 1})$ is $\alpha$ evaluated on the $\hat{k}_{\pm}$ subcell of $\msquare_d$. Note that $\delta \bmsquare_{d-2}$ is nonzero on two $(d-1)$-subcells of each of the four neighboring $d$-hypercubes. For the four $d$-hypercubes containing $\msquare_{d-2}$, the pairs of two $(d-1)$-subcells are labelled as $(\ihat_{+_1}, \jhat_{+_2})$, $(\ihat_{+_1}, \jhat_{-_2})$, $(\ihat_{-_1}, \jhat_{+_2})$, and $(\ihat_{-_1}, \jhat_{-_2})$. For these cases we'd end up having 
    \begin{equation}
        \int \bmsquare_d \cup_{d-1} \delta \bmsquare_{d-2} = 
        \begin{cases}
            0 \text{ if } i + j + p(1) + p(2) = 0 \text{ (mod 2)} \\
            1 \text{ if } i + j + p(1) + p(2) = 1 \text{ (mod 2)}
        \end{cases}
    \end{equation}
    where $(-1)^{p(1,2)} = \pm_{1,2}$. Note that opposite to the discussion at the end of the proof of Lemma \ref{lem:exactBosonization_S_commRels}, the cases where it returns $0$ correspond to the cases where the $\widehat{v}_d$ vector in the square looks like ${\color{red} \leftarrow \rightarrow}$ or ${\color{red} \rightarrow \leftarrow}$. 
    
\begin{figure}[h!]
  \centering
  \includegraphics[width=\linewidth]{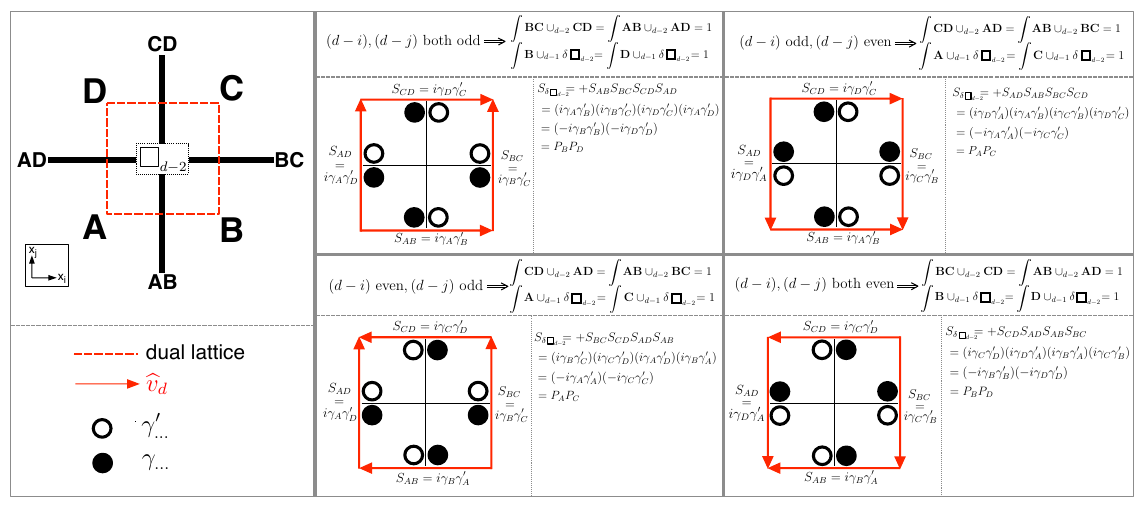}
  \caption[]{Various cases to check in the proof of Lemma~\ref{lem:exactBos_gaugeConstraintFermion}. (Left) $\delta \bmsquare_{d-2}$ defines a loop on the dual lattice in the $x_i,x_j$ plane going between four $d$-hypercubes, ranging between $\msquare_d = A,B,C,D$ in counter-clockwise order. The four $(d-1)$-hypercubes $AB,BC,CD,AD$ are labeled according to the $d$-hypercubes they border. We refer to $\{{\bf A},{\bf B},{\bf C},{\bf D}\}$ and $\{{\bf AB},{\bf BC},{\bf CD},{\bf AD}\}$ as the indicator cochains on the respective hypercubes. The vector $\widehat{v}_d$ as described in Appendix~\ref{app:framingsOfCurvesAndWindings} is a useful reference.
  (Right) The four cases of $(d-i),(d-j)$ being even/odd are checked. The nonzero integrals used to define $S_{\delta \bmsquare_{d-2}}$ and in the Lemma are shown in each case, together with the algebra verifying the equalities.}
  \label{fig:proofOfGaugeConstraint}
\end{figure}
    Now, we only have four cases to verify depending on whether $i,j$ are even/odd. These four cases are analyzed in Figure~\ref{fig:proofOfGaugeConstraint}.
\end{proof}

\subsection{Gu-Wen `Grassmann Integral'}

\subsubsection{Definition using winding numbers}

We start by explaining the how to give the ``winding number'' based defintion of the Grassmann integral $\sigma(\alpha)$. First we'll briefly talk about the case where $\alpha$ is dual to a single closed curve on the dual lattice, which is described in more detail in the Appendix~\ref{app:framingsOfCurvesAndWindings}. Then, we'll talk about how to extend the definition to more general $\alpha$.

Suppose $L$ is a cochain defining a set of dual edges that define a simple, non-self-intersecting, closed curve on the dual lattice. The main idea is that there's a canonical \textit{framing} of the curve and thus a canonical way to define an \textit{induced spin structure} on the framed curve. First, the space $\R^d$ is endowed with a `background framing' consisting of the frame used to define the $\cup_m$ products via the thickening/shifting prescription. The curve itself is endowed with a framing given by the first $(d-2)$ vector fields of the frame used to define the $\cup_m$ products via the thickening/shifting prescription; this corresponds to the `shared framing' discussed in Appendix~\ref{app:framingsOfCurvesAndWindings}. Then, the induced spin structure on the curve can be computed in terms of the number of times the other two vectors of the background frame wind with respect to the tangent of the curve after projecting away these `shared framing' directions. Call this winding number `wind$(L)$'. Then, the induced spin structure will be periodic if wind$(L)$ is even and it'll be anti-periodic if wind$(L)$ is odd. Then, we'll define $\sigma(L)$ for such loops as:
\begin{equation}
    \sigma(L) := -(-1)^{\text{wind(L)}} =
    \begin{cases}
        1 \text{ if anti-periodic} \\
        -1 \text{ if periodic}
    \end{cases}.
\end{equation}

Now we want to extend this definition for more general closed $\alpha$. For such $\alpha$, its dual will be a sum over closed loops: $\alpha = L_1 + \cdots + L_k$. As such, we would like to define $\sigma(\alpha)$ as a product of the $\sigma(L)$ the loop decomposition:
\begin{equation}
    \sigma(\alpha) := \prod_{i=1}^k \sigma(L_k) = (-1)^{\text{\# of loops}} \prod_{i=1}^k (-1)^{\text{wind}(L)}.
\end{equation}
However, a priori the specific loop decomposition $\alpha = L_1 + \cdots + L_k$ is ambiguous because the dual hypercubic lattice is a $2d$-valent graph. The way to deal with this, as in~\cite{T20}, is to introduce a \textit{trivalent resolution} of the dual 1-skeleton, so that any closed loop configuration would be endowed with a unique decomposition into distinct loops. The trivalent resolution we use is depicted in Fig.~\ref{fig:trivalentRes}. For comparison, we also give a trivalent resolution that can be used in the simplicial case to define $\sigma$.

\begin{figure}[h!]
    \centering
    \begin{minipage}{0.48\textwidth}
        \centering
        \includegraphics[width=\linewidth]{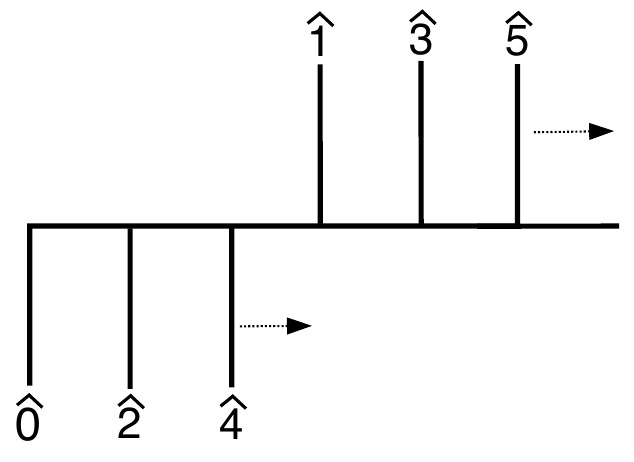}
    \end{minipage}
    \begin{minipage}{0.48\textwidth}
         \centering
         \includegraphics[width=\linewidth]{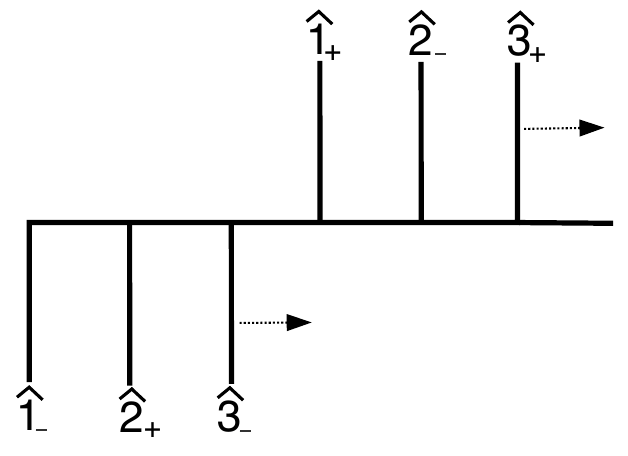}
    \end{minipage}
    \caption{(Left) Trivalent resolution for the dual 1-skeleton on a simplex $\braket{0 \cdots d}$, where $\ihat$ is the cell dual to $\braket{0 \cdots \ihat \cdots d}$. (Right) Trivalent resolution for the dual hypercubic lattice in a single $d$-hypercube, where $\ihat_\pm$ refers to the dual cell pointing in the $\pm x_i$ direction relative to the center. }
    \label{fig:trivalentRes}
\end{figure}

This function with the given trivalent resolution turned out to satisfy the quadratic refinement
\begin{equation} \label{eq:quadRefineSigma}
    \sigma(\alpha)\sigma(\beta)=\sigma(\alpha+\beta)(-1)^{\int \alpha \cup_{d-2} \beta}.
\end{equation}
The argument for this is essentially the same as the argument in~\cite{T20}, so we only briefly talk about it here. Quadratic refinement is a statement about how the windings of the collections of loops dual to $\alpha,\beta$ change when they reconnect and recombine in $\alpha + \beta$. This property is well-known in $d=2$~\cite{johnson1980} where the $\int \alpha \cup_{d-2} \beta$ factor reduces to the intersection number $\int \alpha \cup \beta$ between $\alpha^\vee,\beta^\vee$. In addition, one can use the thickening and shifting prescription of the higher-cup products to reduce the picture to project away the `shared framing' directions and reduce the picture to $d=2$. Then, the $\int \alpha \cup_{d-2} \beta$ factor reduces to the intersection number of the curves in a `projected' picture. 

We also mention a subtlety that the particular trivalent resolution use is important in determining how the fermion lines reconnect with each other, and that not every trivalent resolution will work to reproduce quadratic refinement (although there are many that do work). In Fig.~\ref{fig:trivalentRes}, we gave an example of resolutions that work in the simplicial case.\footnote{The resolution shown here is not the same as the one originally introduced in~\cite{T20}, although the methods there show that this one shown also works.} It's also interesting to compare our trivalent resolution to Fig.~\ref{fig:windingsGeneralDimensions} in the Appendix were we draw a projected version of the dual 1-skeletons and vector fields in illustrating how to compute the windings. 

Also, one can show that the values are all $+1$ on elementary loops dual to coboundaries of indicators $\msquare_{d-2}$ on a single $(d-2)$-simplex:
\begin{equation} \label{eq:sigmaElLoops}
    \sigma(\delta \bmsquare_{d-2}) = +1.
\end{equation}
In the simplicial case, the formula was actually $\sigma(\delta {\bf \Delta}_{d-2}) = (-1)^{w_2({\bf \Delta}_{d-2})}$ for indicator cochains ${\bf \Delta}_{d-2}$ and $w_2$ being the chain-representative dual to the second Stiefel-Whitney class. The geometric reason that these are all $+1$ here is that the background vector fields used to define the $\cup_{d-2}$ products are constant, so won't have any singularities, thus $w_2 \equiv 0$ for us.

Quadratic refinement and the values on elementary loops in Eqs.~(\ref{eq:quadRefineSigma},\ref{eq:sigmaElLoops}) actually fix the values of $\sigma(\delta \lambda)$ on all coboundaries $\delta \lambda$ (see~\cite{GK16}). In particular, we'll have:
\begin{equation} \label{eq:GrassmannIntegral_TrivialLoops}
    \sigma(\delta \lambda) = (-1)^{\int \lambda \cup_{d-4} \lambda + \lambda \cup_{d-3} \delta \lambda}.
\end{equation}
And since $\Z^d$ is a trivial cell structure, this would actually tell us all values of $\sigma(\alpha)$. However, the constructions for $\sigma$ still make sense if we compactify the hypercubic lattice, say into a torus, and the values on nontrivial loops are \textit{not} fixed in a similar way.

\subsubsection{Definition using Grassmann variables}
While the above does produce a function $\sigma(\alpha)$, it would be nice if one could produce the same function $\sigma(\alpha)$ using Grassmann variables in the spirit of~\cite{guWen2014Supercohomology,GK16}. In fact we can produce a function that does exactly that.

For a cochain $\alpha \in C^{d-1}(\text{cubical complex},\Z_2)$, define $\sigma^\text{gr}(\alpha)$ as follows:
\begin{equation}
\sigma^\text{gr}(\alpha) = \int \prod_{e|\alpha(e)=1} d\theta_e d{\bar{\theta}_e} \prod_{c}u(c)
\end{equation}
where for each hypercube $c$, $u(c)$ is a product of Grassmann variables of $\theta_e$ or $\bar{\theta}_e$ of the edges $e \in c$. We'll define $u(c)$ differently depending on when $d$ is even or odd:
\begin{equation}
u(c) = 
\begin{cases}
\theta_{\hat{1}_-}^{\alpha(\hat{1}_-)}\theta_{\hat{2}_+}^{\alpha(\hat{2}_+)}\theta_{\hat{3}_-}^{\alpha(\hat{3}_-)}\cdots {\bar{\theta}}_{\hat{1}_+}^{\alpha(\hat{1}_+)}{\bar{\theta}}_{\hat{2}_-}^{\alpha(\hat{2}_-)}{\bar{\theta}}_{\hat{3}_+}^{\alpha(\hat{3}_+)} \cdots  \text{ if } d \text{ even}  \\
\cdots \theta_{\hat{3}_+}^{\alpha(\hat{3}_+)} \theta_{\hat{2}_-}^{\alpha(\hat{2}_-)} \theta_{\hat{1}_+}^{\alpha(\hat{1}_+)} \cdots {\bar{\theta}}_{\hat{3}_-}^{\alpha(\hat{3}_-)}{\bar{\theta}}_{\hat{2}_+}^{\alpha(\hat{2}_+)}{\bar{\theta}}{\hat{1}_-}^{\alpha(\hat{1}_-)} \text{ if } d \text{ even} 
\end{cases}.
\end{equation}
Here, the notation $\theta_{\ihat_\pm}^{\alpha(\ihat_\pm)}$ or ${\bar{\theta}}_{\ihat_\pm}^{\alpha(\ihat_\pm)}$ means we include the Grassmann variable $\theta$ or ${\bar{\theta}}$ associated to the edge $\ihat_\pm$ if $\alpha(\ihat_\pm)=1$, i.e. if that edge is included in the cochain. See Fig.~\ref{fig:grassmannAssignments} for a depiction of the Grassmann assignments. 

\begin{figure}[h!]
  \centering
  \includegraphics[width=\linewidth]{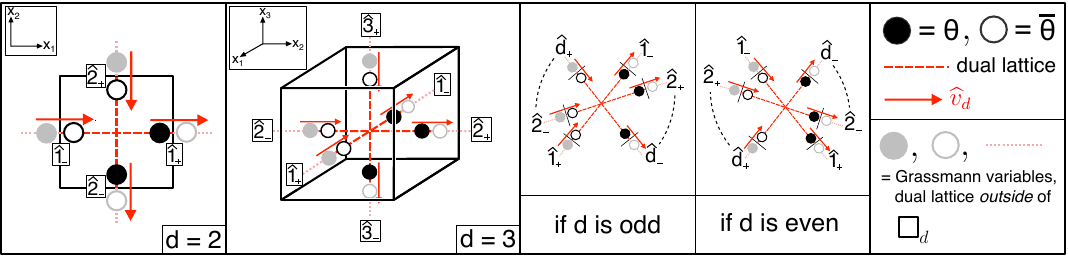}
  \caption[]{Assignment of Grassmann variables in a single $\msquare_d$ used in the definition of $\sigma^{\text{gr}}(\alpha)$ relative to the dual 1-skeleton and the $\widehat{v}_d$ vector (see Appendix~\ref{app:framingsOfCurvesAndWindings}). A solid, open circle represents a placement of $\theta,\bar{\theta}$ respectively and the red dashed line is the dual 1-skeleton. The grayed circles and gray-red line represent the analogous quantities \textit{outside} of $\msquare_d$.}
  \label{fig:grassmannAssignments}
\end{figure}

Note that the assignment of Grassmann variables on edges mimics the assignments of Majorana operators of the exact bosonization in the sense that $\theta,\bar{\theta}$ variables occur in the same configurations as $\gamma,\gamma'$ respectively (e.g. compare Figs.~\ref{fig:majoranaLR_assignment},\ref{fig:proofOfGaugeConstraint} to Fig.~\ref{fig:grassmannAssignments}). Later on, in Sec.~\ref{sec:fermionicToricCode}, we'll actually connect $\sigma(f)$ to the exact bosonization in how it shows up in amplitudes of ground state of the `fermionic toric code'.

It ends up that $\sigma^{\text{gr}}(\alpha)$ literally equals the winding-based definition $\sigma(\alpha)$. We don't reproduce the proof here but briefly describe it and how to map this hypercubic case onto a proof in an Appendix of~\cite{TKBB2021anomalies}, which shows the analogous statement in the simplicial case.
First, note that the order of the Grassmann variables is the same order as the edges that appear in the trivalent resolution Fig.~\ref{fig:trivalentRes} going from left to right. This can be used to show that both $\sigma(\alpha)$ and $\sigma^\text{gr}(\alpha)$ can be decomposed into products over the same set of loops that the dual of $\alpha$ gets decomposed into. Then, one has to check that indeed $\sigma(L)$ and $\sigma^\text{gr}(L)$ match on any cocycle $L$ dual to a single loop which can be done via some case-work and an inductive argument on the number of pairs of partial windings that occur throughout the loop. The map to the proof in~\cite{TKBB2021anomalies} can be guessed by matching orders of dual edges in Fig.~\ref{fig:trivalentRes}, where the trivalent resolutions used for the simplicial and hypercubic cases are shown side-by-side. In particular, the proof in $d$ dimensions for the hypercubic case maps onto the proof in $(2d -1)$ dimensions in the simplicial case.

\section{Construction of lattice models} \label{sec:ConstructionOfLatticeModels}
Given a $(d+1)$-dimensional action $S_{d+1}$ for some space-time topological action for some theory, the ground-state wavefunction $\ket{\Psi}$ of the theory on a $d$-dimensional spatial manifold $M$ can be constructed as:
\begin{equation}
\braket{\Phi|\Psi} \propto Z(\Phi) = \sum_{\text{fields}} \exp{ \Big( 2 \pi i \cdot S_{d+1}(\text{fields} ; \Phi) \Big)}
\end{equation}
where $\Phi$ is some state in the microscopic boundary Hilbert space that represents the boundary conditions of the topological action. Here, the partition function $Z(\Phi)$ is defined as a sum over some background `fields' compatible with boundary conditions, and the `action' $S_{d+1}(\text{fields} ; \Phi)$ depends on the boundary conditions $\ket{\Phi}$ and this background field $\Phi$.

A Hamiltonian formulation, or \textit{lattice model}, of a topological action is a gapped Hamiltonian that admits $\ket{\Psi}$ as a ground-state. Well-known examples of such a correspondence are the Turaev-Viro/Levin-Wen~\cite{turaevViro1992,levinWen2005} and Crane-Yetter/Walker-Wang~\cite{craneYetter1993,walkerWang2012} correspondence which give path-integral/Hamiltonian correspondences for (resp.) (2+1)-$d$ and (3+1)-$d$ TQFTs related to tensor categories.

To construct a $d$-dimensional Hamiltonian out of a $(d+1)$-dimensional path integral, one needs to consider the ground states of the manifold on a triangulated surface with boundary. In general, such a ground state will be a superposition over all possible boundary conditions, with amplitudes being proportional to the path integral evaluated in the presence of the boundary conditions. For example, the Levin-Wen and Walker-Wang models are paradigmatic examples of `string net' models where the ground states on a spatial manifold are configurations of strings with various colors with certain rules for branching at trijunctions. A basis of ground states is in correspondence with equaivalence classes of valid string-net configurations that can be transformed into each other via some set of \textit{local moves}. Such a ground state will be a sum over all configurations in the equivalence class, and each local move gives a relationship between the amplitudes of these configurations. In the space-time path-integral, these string-nets are actually configurations of two-dimensional sheets that restrict to a string-net on the spatial boundary.

One way to construct such a Hamiltonian with such a ground state, is to consider one with two kinds of \textit{mutually commuting} terms. The first kind are called `charge terms' (often called `star terms') which enforce certain constraints on what kinds of states are allowed in the ground state. In string-net states, these enforce the branching rules. The `flux terms' (often called `plaquette terms' in the context of string-net models) give the relations between amplitudes of configurations allowed by local moves. If all these terms mutually commute, then one can find the ground states exactly. And, a basis of grounds state is the equivalence classes of possible configurations modulo local moves. 

Of particular interest to us are actions that can be written as a sum over cochains or cocycles. In particular, we'll have that general states can look like $\ket{\Phi_1, \cdots,  \Phi_n}$ for cochains $\{\Phi_1, \cdots, \Phi_n\}$, and states appearing with nonzero amplitudes in a ground state $\ket{\Psi}$ will require a closed-cochain condition, that each $\delta \Phi_k = 0$ if $\braket{\Phi | \Psi} \neq 0$. This condition will correspond to the charge terms. 

The strategy to derive flux terms from a space-time action involves considering cochains $\{\Phi_k\}$ and their ground state amplitudes on a topologically trivial manifold, so that each $\Phi_k = \delta \tilde{a}_k$ is a coboundary. Generally on trivial manifolds, we will be able to derive exact expressions for $\braket{\{\delta \tilde{a}_k \}| \Psi}$ in terms of the $\tilde{a}_k$. Then, we'll show that the changes in amplitudes between $\braket{ \{ \delta \tilde{a}_k \} | \Psi}$ and $\braket{\{ \delta \tilde{a}_k + \delta \lambda_k  \}| \Psi}$ for indicator cochains $\lambda_k$ only depend on the $\delta \tilde{a}_k = \Phi_k$, and thus give local expressions for local changes in amplitudes.

And again, one would be able to derive the most general ground state on a manifold by considering the space of cocycle configurations modulo local changes, which in our examples will be in correspondence with cohomology classes, which by definition are cocycles modulo local changes.

\subsection{Warmup: $\Z_n$ toric code in arbitrary dimensions}
We start by revewing the simplest example of such an action is the so-called `toric code' following discussions in~\cite{BGK17,kapustinThorngren2017}. We start by reviewing the standard $\Z_2$ $(2+1)$D toric code and at the end note how this action generalizes to higher dimensions and $\Z_n$ coefficients, giving similar models.

The toric code's space-time partition function can be expresssed as a sum over $\Z_2$ cochains $a \in C^1(M,\Z_2)$ and chains $b \in C_2(M,\Z_2)$ with action:
\begin{equation}
    S_{t.c.} = \frac{1}{2} \int_{M^{2+1}} (a)(\partial b)
\end{equation}
where $\int (a)(\partial b)$ gives the cochain-chain pairing $C^1(M,\Z_2) \times C_1(M,\Z_2) \to \Z_2$. For us, we consider boundary conditions of wavefunctions on the manifold as $\restr{a}{\partial M}, \restr{b}{\partial M}$. Actually for us, the $\restr{b}{\partial M}$ represent `electric excitations' above the ground state corresponding to electric-particle excitations, and we will set $\restr{b}{\partial M} = 0$. Note that our \textit{microscopic} boundary Hilbert space is a tensor product of two-state systems $\{\ket{0}_e, \ket{1}_e\}$ on each edge $e$ (or equivalently dual edge $e^\vee$) on the boundary 2D triangulation, and a state on the boundary corresponds to a boundary condition $\restr{a}{\partial M}(e) = \Phi(e) \in \{0,1\}$ for each edge.

Note that by Stoke's theorem, above action reduces to
\begin{equation}
\begin{split}
    S_{t.c.} &= \frac{1}{2} \int_{M} (a)(\partial b) = \frac{1}{2} \int_{M} (\delta a)(b) +  \frac{1}{2}\int_{\partial M} (a)(b) \\
    &=  \frac{1}{2} \int_{M} (\delta a)(b)
\end{split}
\end{equation}
where the first equality is an integration-by-parts and the second equality uses $\restr{b}{\partial M} = 0$. Note that the partition function associated to the above action will give a ground-state wavefunction
\begin{equation}
\begin{split}
    \braket{\Phi | \Psi} \propto Z_{t.c.}(\Phi) = \sum_{\substack{a \in C^1(M,\Z_2) \,,\, \restr{a}{\partial M}(e) = \Phi(e) \\ b \in C_2(M,\Z_2)}}
    (-1)^{\int (\delta a) (b)}.
\end{split}
\end{equation}
Note that summing over $b$ above gives that the amplitude $\braket{\Phi | \Psi}$ is \textit{zero} unless $\delta a = 0$ everywhere throughout the 3-manifold $M$. This means that $a$ must be dual to closed worldsheets in $M$ that restrict to closed-loop configurations on $\partial M$, thus any nonzero amplitude in the $\ket{\Psi}$ wavefunction is associated to a closed-loop configuration on the dual cellulation of $\partial M$. This means that $b$ acts as a \textit{Lagrange multiplier}, and `integrating it out' enforces $\delta a = 0$ in all nonzero amplitudes. This general feature of a Lagrange multiplier will generally apply in all situations we consider, and in the future we'll often write the partition functions directly as sums over closed cochains.

A consequence is that the sum above reduces to
\begin{equation}
    Z_{t.c.}(\Phi) = |C_2(M,\Z_2)| \sum_{a \in C^1(M,\Z_2) | \delta a = 0} 1 =  |C_2(M,\Z_2)|  |Z^1(M,\Z_2)|
\end{equation}
which is a constant positive number for each $\Phi$. As such, we can write the full ground-state wavefunction on a simply-connected spatial manifold $N$ as:
\begin{equation}
\ket{\Psi} \propto \sum_{\tilde{a} \in C^1(N,\Z_2)} \ket{\delta \tilde{a}}
\end{equation}
where $\ket{\delta \tilde{a}}$ is the state with respect to the boundary conditions $\Phi = \delta \tilde{a}$, or equivalently, $\Phi(e) = (\delta \tilde{a})(e)$.

As such, the toric code is one of the simplest `topological orders' because the ground-state wavefunction on a simply connected space is a sum over all closed loop configurations weighted with the same amplitude. In particular, the \textit{local} constraint on these wavefunction amplitudes is trivial. Letting ${\bf v} \in C^0(N, \Z_2)$ be an indicator cochain that's nonzero on only one vertex, we'll have:
\begin{equation} \label{localConstraintToricCode}
    \braket{\Phi + \delta {\bf v} | \Psi} = \braket{\Phi | \Psi}.
\end{equation}
On a non-simply connected space, it is precisely this local constraint that we take to be the \textit{definition} of the toric-code topological order. One consequence of this is that for any closed spatial manifold $N$ there are really $|H_1(N,\Z_2)|$ independent wavefunctions $\ket{\Psi}$ that satisfy the above constraint, which correspond exactly to the $\Z_2$ functions on cohomology classes. 

Now, we're in a position to describe the toric code Hamiltonian. We'll express the Hamiltonian in terms of the \textit{dual cellulation} so that faces $f$ of the original cellulation are dual to vertices $f^\vee$ on the dual, and vertices $v$ of the original cellulation are dual to faces $v^\vee$
\begin{equation}
\begin{split}
    H_{t.c.} = - \sum_{f^\vee \in \substack{\text{dual} \\ \text{vertices}}} U^{t.c.}_{f^\vee}  - \sum_{v^\vee \in \substack{\text{dual} \\ \text{faces}}} W^{t.c.}_{v^\vee} \text{, where} \\
    U^{t.c.}_{f^\vee} := \prod_{e^\vee \supset f^\vee} Z_{e^\vee}, \quad\quad W^{t.c.}_{v^\vee} := \prod_{e^\vee \subset v^\vee} X_{e^\vee}.
\end{split}
\end{equation}
Here, $e^\vee \supset f^\vee$ are all the (dual) edges containing the vertex $f^\vee$, and $e^\vee \subset v^\vee$ are the (dual) edges at the boundary of $v^\vee$. And, in the basis $\{\ket{0}_{e^\vee},\ket{1}_{e^\vee}\}$ of each dual edge, we have $Z_{e^\vee} = \begin{pmatrix}1 & 0 \\ 0 & -1\end{pmatrix}$ and $X_{e^\vee} = \begin{pmatrix}0 & 1 \\ 1 & 0\end{pmatrix}$. 
Note that each term $U^{t.c.},W^{t.c.}$ in the Hamiltonian commute, which means that the above is a \textit{local commuting projector Hamiltonian}. As such, any ground state of $H_{t.c.}$ is in the shared ground state of each term $-U^{t.c.}_{f^\vee}$ and $-W^{t.c.}_{v^\vee}$. 

The first set of `star terms' enforce that $U^{t.c.}_{f^\vee} \ket{\Psi} = \ket{\Psi}$. As such, let $\Phi \in C^1(M,\Z_2)$ be such that $\braket{\Phi | \Psi} \neq 0$. Then,
\begin{equation}
    \bra{\Phi} U^{t.c.}_{f^\vee} \ket{\Psi} = \bra{\Phi} \prod_{e^\vee \supset f^\vee} Z_{e^\vee} \ket{\Psi} = \braket{\Phi | \Psi}
\end{equation}
while also
\begin{equation}
    \ket{\Phi} = \prod_{e^\vee \supset f^\vee} Z_{e^\vee} \ket{\Phi} = (-1)^{(\delta \Phi)(f)} \ket{\Phi}.
\end{equation}
The above two together require that $\delta \Phi \equiv 0$, so that $\Phi$ is a cocycle. 

The second set of `flux terms' enforce that each $W^{t.c.}_{v^\vee} \ket{\Psi} = \ket{\Psi}$, so that if $\braket{\Phi | \Psi} \neq 0$, then
\begin{equation}
  \bra{\Phi} W^{t.c.}_{v^\vee} \ket{\Psi} = \bra{\Phi}\prod_{e^\vee \subset v^\vee} X_{e^\vee} \ket{\Psi} = \braket{\Phi | \Psi}
\end{equation}
Also note that 
\begin{equation}
    \prod_{e^\vee \subset v^\vee} X_{e^\vee} \ket{\Phi} = \ket{\Phi + \delta {\bf v}}
\end{equation}
Note that the above two equations exactly correspond to the local constraints Eq.~\eqref{localConstraintToricCode} for each $v$. 

In general dimension $d$, we'll have that the toric code looks quite similar. We can write down essentially the same $(d+1)$ dimensional action
\begin{equation}
    S_{t.c.} = \frac{1}{2} \int_{M^{d+1}} (a)(\partial b)
\end{equation}
for the toric code, except where now the partition functions will sum over $a \in C^{d-1}(M,\Z_2)$ with a Lagrange multiplier $b \in C_{d}(M,\Z_2)$ that now sums over $d$-chains to enforce $\delta a = 0$. And in the same way, we can impose boundary conditions $\Phi$ that give the restrictions of $a$ onto the boundary of the manifold, and we'd get again
\begin{equation}
    Z_{t.c.}(\Phi) = |C_{d}(M,\Z_2)| \sum_{a \in C^1(M,\Z_2) | \delta a = 0} 1 =  |C_d(M,\Z_2)|  |Z^{d-1}(M,\Z_2)|
\end{equation}
so that amplitudes of a ground state wave-function would be a sum over closed loop configurations on the dual lattice all weighted with equal positive amplitudes.

And again, one can engineer an entirely analogous Hamiltonian
\begin{equation}
\begin{split}
    H_{t.c.} = - \sum_{(\msquare_{d})^\vee \in \substack{\text{dual} \\ \text{vertices}}} U^{t.c.}_{(\msquare_{d})^\vee}  - \sum_{(\msquare_{d-2})^\vee \in \substack{\text{dual} \\ \text{2-cells}}} W^{t.c.}_{(\msquare_{d-2})^\vee} \text{, where} \\
    U^{t.c.}_{(\msquare_{d})^\vee} := \prod_{(\msquare_{d-1})^\vee \subset (\msquare_{d})^\vee} Z_{(\msquare_{d-1})^\vee}, \quad\quad W^{t.c.}_{(\msquare_{d-1})^\vee} := \prod_{(\msquare_{d-1})^\vee \supset (\msquare_{d-2})^\vee} X_{(\msquare_{d-1})^\vee}.
\end{split}
\end{equation}
which with the charge terms $U$ enforce that the dual loop configurations are closed and the flux terms $W$ give local constraints so that all amplitudes so that in some ground sate, loop configurations in the same cohomology class are given equal amplitudes.

See Fig.~\ref{fig:toricCode_Terms} for a depiction of the terms in the Hamiltonian.

\begin{figure}[h!]
  \centering
  \includegraphics[width=\linewidth]{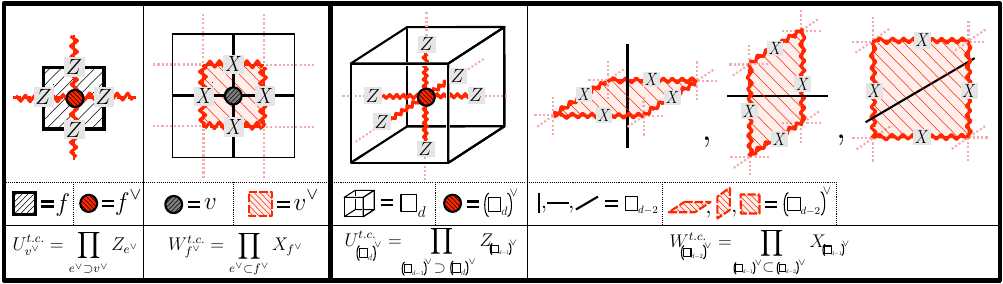}
  \caption{The different terms that occur in the $\Z_2$ toric code Hamiltonian. (Left) Two dimensions (Right) Three dimensions, with notation written for arbitarary dimensions.}
  \label{fig:toricCode_Terms}
\end{figure}

The $\Z_n$ toric code will also be entirely analogous except it's defined in terms of $\Z_n$ cochains rather than $\Z_2$ ones. We'll define the action in $(d+1)$-dimensions as
\begin{equation}
    S_{\Z_n t.c.} = \frac{1}{n} \int_{M^{d+1}} (a)(\partial b)
\end{equation}
in terms of the arguments $a \in C^1(M,\Z_n)$ and $b \in C_d(M,\Z_n)$ with boundary conditions $\restr{b}{\partial M} = 0$. And again we impose boundary conditions $\Phi \in C^1(\partial M, \Z_n)$ so that $\Phi = \restr{a}{\partial M}$. And again, we'd have that
\begin{equation}
    Z_{\Z_n t.c.}(\Phi) = |C_{d}(M,\Z_n)| \sum_{a \in C^1(M,\Z_n) | \delta a = 0} 1 =  |C_d(M,\Z_n)|  |Z^1(M,\Z_n)|
\end{equation}
which would enforce that the ground state is a sum over all closed $\Z_n$-valued cochains.

And again, using the $\cup'$ formalism over $\Z$ in Sec.~\ref{sec:cochainsOverZ}, it is straightforward to write down a Hamiltonian
\begin{equation}
     H_{\Z_n t.c.} = - \sum_{(\msquare_{d})^\vee \in \substack{\text{dual} \\ \text{vertices}}} U^{\Z_n t.c.}_{(\msquare_{d})^\vee}  - \sum_{(\msquare_{d-2})^\vee \in \substack{\text{dual} \\ \text{2-cells}}} W^{\Z_n t.c.}_{(\msquare_{d-2})^\vee}
\end{equation}
for which the $U^{\Z_n t.c.}$ enforce that ground states are dual to closed $\Z_n$-cochains and for which $W^{\Z_n t.c.}$ give that cohomologous cochains are given equal weights in the wavefunctions.

\subsection{ Three-fermion Walker-Wang model}


In this section, we reproduce the lattice Hamiltonian for the 3-fermion Walker-Wang model \cite{WW12, BCFV14}. This model is constructed from the 3-fermion braided fusion category, which is an abelian category with three nontrivial objects $\{e,m,\epsilon\}$ with fusion rules
\begin{equation}
    e \times e = m \times m = \epsilon \times \epsilon = 1 \quad,\quad e \times m = \epsilon.
\end{equation}
The $F$ symbols are trivial, and the $R$-symbols are given by
\begin{equation}
\begin{split}
    R_{\mu, \mu} &= -1, \text{ for all } \mu = e,m,\epsilon \\
    R_{e,m} &= R_{m,\epsilon} = R_{\epsilon,e} = -1 \\
    R_{m,e} &= R_{\epsilon,m} = R_{e,\epsilon} = +1.
\end{split}
\end{equation}
The $R_{\mu \mu} = -1$ tell us that all the nontrivial particles are self-fermions. 

A bulk action corresponding to this category can be found using the Crane-Yetter state-sum, whose Hamiltonian formulation is given by the corresponding Walker-Wang model. First, note that since the category is abelian, one can define the assignments of $e,m,\epsilon$ particles by two closed cochains $A,B \in Z^2(M^4,\Z_2)$. On a 2-cell, $A=B=0$ means placing an identity particle on the 2-cell, $A=1,B=0$ means place an $e$ particle, $A=0,B=1$ means place an $m$ particle, and $A=B=1$ means place an $e \times m = \epsilon$ particle. Then, the amplitudes of the bulk action are given by the $15j$ symbols of the assigments multiplied over all the 4-simplices. Writing out the $15j$ symbols shows that the amplitude to such an assignment of $A,B$ is exactly 
\begin{equation}
    S_{\text{3-fermion}}(A,B) = \frac{1}{2} \int_{M^4} A \cup A + B \cup B + A \cup B
\end{equation}
since each $15j$ symbol on $\braket{01234}$ evaluates to $(-1)^{(A \cup A + B \cup B + A \cup B)(\braket{01234})} = (-1)^{A(012)A(234) + B(012)B(234) + A(012)B(234)}$. See also \cite{KT14, HLS19, HJJ21} for more general actions of theories with 1-form or higher-group symmetries. On a closed manifold, one can evaluate a partition function on closed $M^4$ as 
\begin{equation}
\begin{split}
    Z(M^4) &\propto \sum_{A,B \in Z^2(M^4,\Z_2)} (-1)^{\int_M A \cup A + B \cup B + A \cup B} \\
    &= \sum_{A,B \in Z^2(M^4,\Z_2)} (-1)^{\int (w_2 + w_1^2) \cup A + B \cup B + A \cup B} \\
    &= \sum_{A,B \in Z^2(M^4)} (-1)^{\int_{M^4} (w_2 + w_1^2 + B) \cup A + B \cup B} \\
    &= (-1)^{\int_{M^4} (w_2 + w_1^2) \cup (w_2 + w_1^2)} = (-1)^{\int_{M^4} w_2^2 + w_1^4}.
\end{split}
\end{equation}
Above, the second line used the Wu relation $\int_{M^4} A \cup A = \int_{M^4} (w_2 + w_1^2) \cup A$, and the fourth line notes that the sum over $A$ acts as a Lagrange multiplier setting $B = (w_2 + w_1^2)$. See also~\cite{barkeshli2019_nonorientable} for an alternate perspective on this result in terms of `anomaly indicators'. 

\subsubsection{Bulk Hamiltonian}
Now, we construct the Hamiltonian of this model and show that indeed it is the same as the one presented in Ref.~\cite{BCFV14,HFH18}. Note that we now need to consider two sets of $\Z_2$ degrees of freedom on the spatial manifold, $\Phi_1,\Phi_2 \in Z^2(N^3,\Z_2)$ that represent the boundary conditions of $A,B$ respectively. Now in our procedure, we consider the situation where both the bulk $M^4$ and boundary $N^3$ are topologically trivial, so that $A = \delta a$ and $B = \delta b$ in the bulk manifold. Then, with respect to the boundary conditions $\Phi_1 = \restr{A}{N} = \restr{\delta a}{N},\Phi_2 = \restr{B}{N}=\restr{\delta b}{N},$ the ground state $\ket{\Psi}$ on $N$ will have amplitudes
\begin{equation}
\begin{split}
    \braket{\Phi_1, \Phi_2 | \Psi} &\propto (-1)^{\int_M (\delta a) \cup (\delta a) + (\delta b) \cup (\delta b) + (\delta a) \cup (\delta b)} \\
    &= (-1)^{\int_M \delta(a \cup \delta a + b \cup \delta b + a \cup \delta b)} \\
    &= (-1)^{\int_N a \cup \delta a + b \cup \delta b + a \cup \delta b}.
\end{split}
\end{equation}
From here, we compute the variations of these amplitudes with respect to $\Phi_1 \to \Phi_1 + \delta \lambda_1$ and $\Phi_2 \to \Phi_2 + \delta \lambda_2$ which correspond to $a \to a + \lambda_1$ and $b \to b + \lambda_2$ respectively. These have the form
\begin{equation} \label{eq:threeFermion_Flux1}
\begin{split}
    \int_N & \big( (a + \lambda_1) \cup \delta(a + \lambda_1) + b \cup b + (a + \lambda_1) \cup \delta b \big) -\big( a \cup a + b \cup b + a \cup \delta b \big) = \int_N a \cup \delta \lambda_1 + \lambda_1 \cup \delta a + \lambda_1 \cup \delta \lambda_1 + \lambda_1 \cup \delta b \\
    &= \int_N \delta a \cup \lambda_1 + \lambda_1 \cup \delta a + \lambda_1 \cup \delta \lambda_1 + \lambda_1 \cup \delta b + \int_{\partial N} a \cup \lambda_1,
\end{split}
\end{equation}
and
\begin{equation} \label{eq:threeFermion_Flux2}
\begin{split}
    \int_N & \big( a \cup \delta a + (b + \lambda_2) \cup (b + \lambda_2) + a \cup \delta (b + \lambda_2) \big) - \big( a \cup a + b \cup b + a \cup \delta b \big) = \int_N b \cup \delta \lambda_2 + \lambda_2 \cup \delta b + \lambda_2 \cup \delta \lambda_2 + a \cup \delta \lambda_2 \\
    &= \int_N \delta b \cup \lambda_2 + \lambda_2 \cup \delta b + \lambda_2 \cup \delta \lambda_2 + \delta a \cup \lambda_2 + \int_{\partial N} (a + b) \cup \lambda_2,
\end{split}
\end{equation}
where the second lines of each of the above use an integration-by-parts. For now, we'll ignore these boundary integrals over $\partial N$ and interpret them later in the discussion of the boundary Hamiltonian.

At least for the bulk terms, we've expressed all the changes in amplitudes strictly in terms of $\Phi_1 = \delta a$ and $\Phi_2 = \delta b$, which means we have expressions for how the wavefunction amplitudes change with respect to local changes of the loop states. Our space of states will be two $\Z_2$ degrees $\{\ket{0}_{A, \msquare_2^\vee},\ket{1}_{A,\msquare_2^\vee}\}$ and $\{\ket{0}_{B, \msquare_2^\vee},\ket{1}_{B,\msquare_2^\vee}\}$ of freedom on each dual edge, dual to the 2-dimensional plaquettes that describe the cochains.  And, there is the associated algebra of Pauli matrices $X_{A,\msquare_2^\vee},X_{B,\msquare_2^\vee}$ and $Z_{A,\msquare_2^\vee},Z_{B,\msquare_2^\vee}$.




First, the Hamiltonian will have charge (or star) terms $U^{\text{3-fermion}}_{\msquare^\vee_3,\{A,B\}}$, at each vertex $\msquare^\vee_3$ dual to a 3-cube, that enforce $\delta \Phi_1 = \delta \Phi_2 = 0$. Next, there will be flux (or plaquette) terms $W^{\text{3-fermion}}_{\msquare^\vee_1,\{A,B\}}$ for each possible change of wave-function cochain $\Phi_1 \to \Phi_1 + \delta \bmsquare_1$ or $\Phi_2 \to \Phi_2 + \delta \bmsquare_1$ for $\bmsquare_1 \in C^1(N,\Z_2)$ the indicator cochains representing changes in the different $\Z_2$ degrees of freedom. These together create a Hamiltonian
\begin{equation}
    H^{\text{3-fermion}} = -\sum_{\msquare_3} U^{\text{3-fermion}}_{\msquare^\vee_3,A} -\sum_{\msquare_3} U^{\text{3-fermion}}_{\msquare^\vee_3,B} -\sum_{\msquare_1} W^{\text{3-fermion}}_{\msquare^\vee_1,A} -\sum_{\msquare_1} W^{\text{3-fermion}}_{\msquare^\vee_1,B}
\label{eq: 3fWW original Hamiltonian}
\end{equation}
whose terms are depicted in Fig.~\ref{fig:threeFermion_Terms}. Note that these terms match the previous 3-fermion Walker-Wang model construction \cite{BCFV14} exactly.

\begin{figure}[h!]
  \centering
  \includegraphics[width=\linewidth]{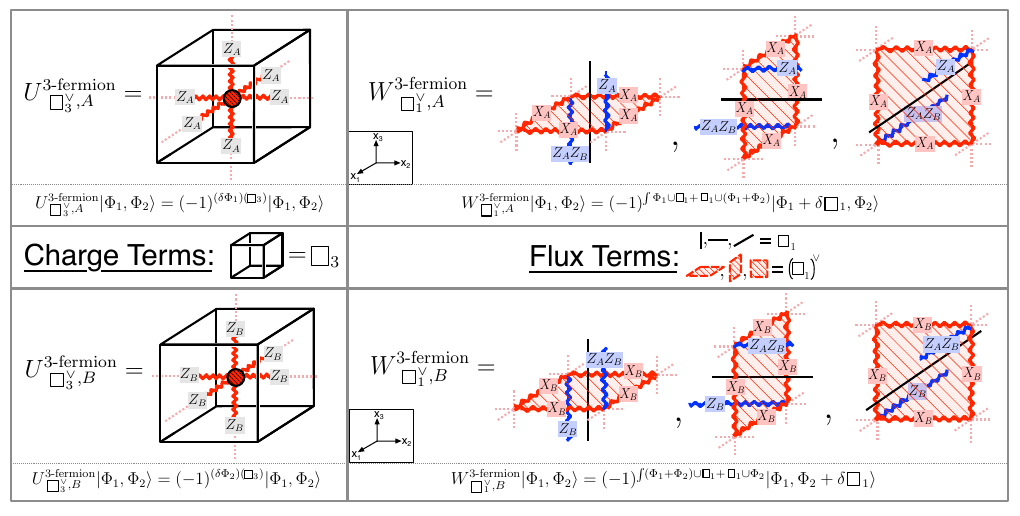}
  \caption{Terms in the Three-Fermion Walker-Wang model in terms of hypercubic cochain operations, following the flux terms Eqs.~(\ref{eq:threeFermion_Flux1},\ref{eq:threeFermion_Flux2}) and Fig.~\ref{fig:cupProductImages}}
  \label{fig:threeFermion_Terms}
\end{figure}

\subsubsection{Boundary Hamiltonian}
In the derivation of the bulk Hamiltonian above, a key step in getting the right wave-function amplitudes was an integration-by-parts on the second lines of Eqs.~(\ref{eq:threeFermion_Flux1},\ref{eq:threeFermion_Flux2}). However, the integration-by-parts left some boundary terms that would affect the required plaquette terms on the boundary. 

Before deriving boundary terms for the Hamiltonian, it's helpful to have a short digression on how to think about cochains in the presence of a boundary. Recall that in the discussion above, $a,b \in C^1(N,\Z_2)$ are 1-cochains that are dual to two-dimensional sheets in the bulk, whereas $\delta a, \delta b \in Z^2(N,\Z_2)$ are 2-cochains dual to one-dimensional lines. However the restrictions $\restr{\{a,b\}}{\partial N}$ and $\restr{\{\delta a, \delta b\}}{\partial N}$ to $\partial N$ will still be 1-cochains and 2-cochains respectively, which will instead be dual to 1-dimensional lines and 0-dimensional points on $\partial N$. This has the interpretation that the loop configurations associated to $\delta a$ in the bulk $N$ are continued onto some lines associated to $a$ on $\partial N$. See Fig.~\ref{fig:cochainsOnBoundary} for a depiction.

\begin{figure}[h!]
  \centering
  \includegraphics[width=0.4\linewidth]{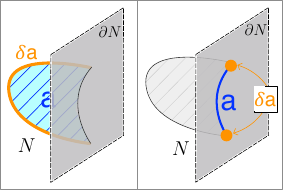}
  \caption{Here, $a \in C^1(N,\Z_2)$ and $\delta a \in C^2(N,\Z_2)$ are sheets and lines in the bulk manifold $N$. However restricted to $\partial N$, $a$  becomes a line that continues the bulk $\delta a$ and $\delta a$ on the boundary becomes the endpoints of the boundary line.}
  \label{fig:cochainsOnBoundary}
\end{figure}

As such, in our Hamiltonian formulation, the lines of the string-net restricted to $\partial N$ will actually be represented by the cochain $a$ that represented the bulk sheets. In this section we consider wave-functions formed by the basis of all of the dual one-dimsional edges on $N, \partial N$, which would be represented by the Hilbert space
\begin{equation}
    \mathcal{H}_{\text{3-fermion}} = \mathrm{span} \left( \Big\{\ket{\delta a, \delta b; \restr{a}{\partial N}, \restr{b}{\partial N}} \Big| a,b \in C^1(N,\Z_2) \Big\} \right).
\end{equation}

Another way to think of the boundary terms is to consider an `augmented' lattice. For example, if we take $N$ to be the half of the hypercubic lattice, which is the full lattice restricted to one $x_i \le 0$, then we could create an $N^{\text{aug}}$ by adding an extra `row' of $x_i = 1$ to the lattice and considering cochains on that augmented lattice. In addition, one needs to impose some boundary conditions that all cochains involving any coordinate with $x_i > 0$ should be set to zero. Then, the cochains on $(N,\partial N)$ will be in one-to-one correspondence with cochains on $N^{\text{aug}}$ subject to the cochains vanishing if they involve $x_i > 0$. This correspondence is illustrated in Fig.~\ref{fig:dualLatticeOnBoundary}.

\begin{figure}[h!]
  \centering
  \includegraphics[width=\linewidth]{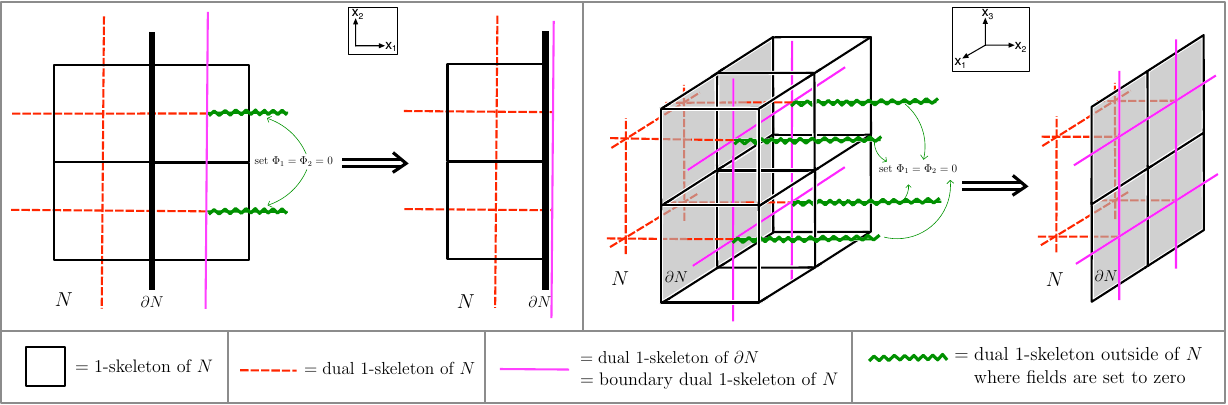}
  \caption{Cochains on $(N,\partial N)$ can be represented by augmenting $N$ by one extra row and considering cochains on the augmented lattice subject to boundary conditions of zero. In our case, we set the fields $\Phi_1 \equiv \delta a = 0$ and $\Phi_2 \equiv \delta b = 0$. (Left) The two-dimensional case. (Right) The three-dimensional case.}
  \label{fig:dualLatticeOnBoundary}
\end{figure}

Now we are in a position to list out the boundary terms implied by the Eqs.~(\ref{eq:threeFermion_Flux1},\ref{eq:threeFermion_Flux2}). They are depicted in Fig.~\ref{fig:threeFermion_BoundaryTerms}. They can be derived from similar considerations as the bulk terms we derived previously. The charge terms enforce that bulk and boundary lines are closed loops when added together. 

\begin{figure}[h!]
  \centering
  \includegraphics[width=\linewidth]{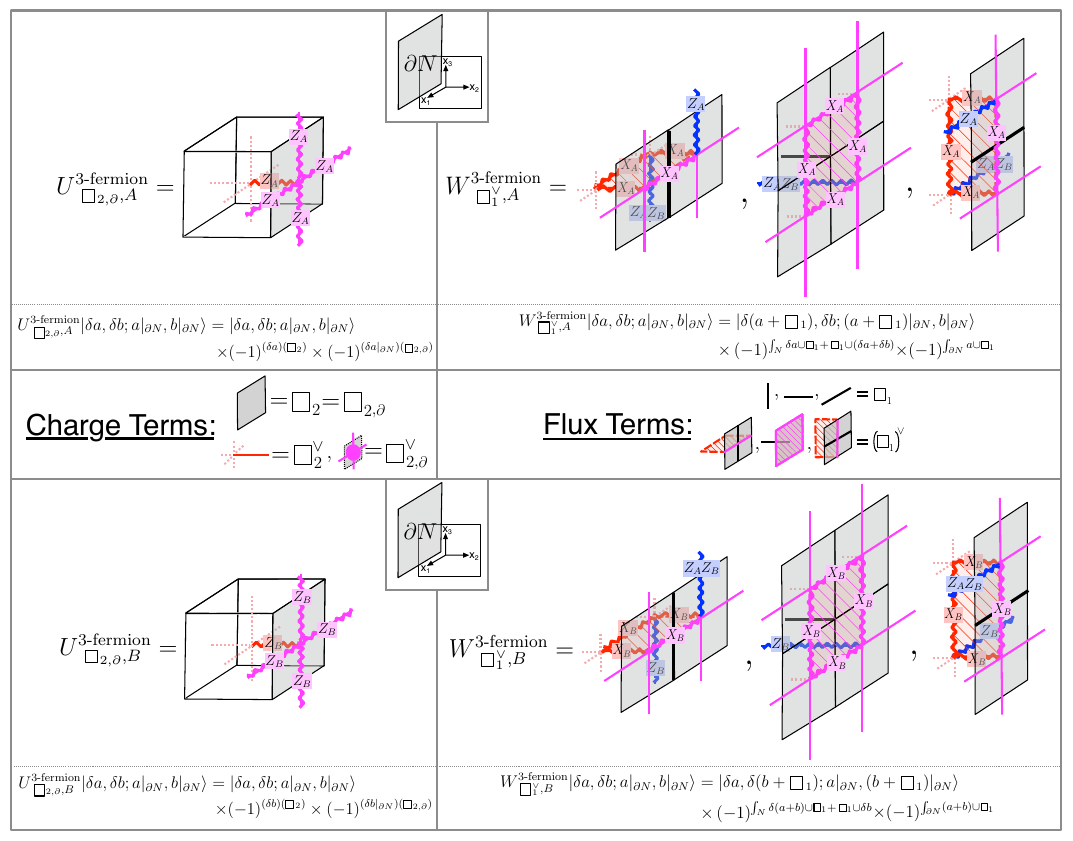}
  \caption[]{Boundary terms of the 3-fermion model subject to a termination in the $(x_1,x_3)$ plane. Here, we illustrate the `charge terms' in terms of the squares $\msquare_{2,\partial}$ on $\partial N$. }
  \label{fig:threeFermion_BoundaryTerms}
\end{figure}

We illustrated these terms in terms of the relevant cochains on $(N,\partial N)$. However, we note that we could just as well think of them as cochains on the augmented lattice $N^{\text{aug}}$; these terms are precisely the bulk terms of Fig.~\ref{fig:threeFermion_Terms} subject to the boundary conditions $a = b = 0$ outside of $N$.

In~\cite{HFH18}, we note that essentially the same set of terms were found to be essentially the unique possible boundary terms for the 3-fermion model \footnote{Actually, the terms they listed on the boundary correspond to our $W_{\msquare^\vee_1, A}^{\text{3-fermion}}$ and $W_{\msquare^\vee_1, A}^{\text{3-fermion}} \times W_{\msquare^\vee_1, B}^{\text{3-fermion}}$, which correspond to the same space of ground states.}.

\subsubsection{Relation to Quantum Cellular Automaton (QCA)}

In Ref. \cite{HFH18}, the 3-fermion Walker-Wang model is used to construct a quantum cellular automaton (QCA). To define the QCA, we need to modify the Hamiltonian such that all of the charge and terms are independent, that no product of such terms is the identity. In the current expression \eqref{eq: 3fWW original Hamiltonian}, the star terms and plaquette terms are not independent (one can check that the product of plaquette terms on faces of a cube is equal to a product of star terms). Ref. \cite{HFH18} introduce the polynomial method and an algorithm to eliminate this redundancy. The final Hamiltonian contains only the plaquette terms expressed by polynomials, which can't be visualized on the cubic lattice easily. In the following, we are going to use the $\cup_1$ product identity on \eqref{eq:threeFermion_Flux1} and \eqref{eq:threeFermion_Flux2} and get modified plaquette terms such that they are non-redundant and complete (able to generate original star terms and plaquette terms).

From \eqref{eq:threeFermion_Flux1} with $\lambda_1 = \bmsquare_1$ (the 1-cochain with value $1$ on the edge $\msquare_1$ and $0$ otherwise) on closed $N$ ($\partial N = 0$), we have
\begin{eqs}
    &\int_N \delta a \cup \bmsquare_1 + \bmsquare_1 \cup \delta a + \bmsquare_1 \cup \delta \bmsquare_1 + \bmsquare_1 \cup \delta b \\
    =& \int_N \delta \bmsquare_1 \cup_1 \delta a + \bmsquare_1 \cup \delta b,
\end{eqs}
where we have used $ \bmsquare_1 \cup \delta \bmsquare_1 = 0$ and the recursive relation $\delta (\alpha \cup_1 \beta) = \delta \alpha \cup_1 \beta + \alpha \cup_1 \delta \beta + \alpha \cup \beta + \beta \cup \alpha$. The new plaquette term is
\begin{eqs}
    W'_{\msquare_1^\vee, A} = \prod_{f|\msquare_1 \in \partial f} X_{A,f} \prod_{f'} Z_{A,f'}^{\int_N \delta \bmsquare_1 \cup_1 \bface'}
    \prod_{f''} Z_{B,f''}^{\int_N \bmsquare_1 \cup \bface''},
\label{eq: new plaquette term 1}
\end{eqs}
or equivalently
\begin{eqs}
    W'_{\msquare_1^\vee, A} \ket{\Phi_1, \Phi_2} = (-1)^{\int \delta \bmsquare_1 \cup_1 \Phi_1 + \bmsquare_1 \cup \Phi_2} \ket{\Phi_1 + \delta \bmsquare_1, \Phi_2}.
\end{eqs}
Similarly, we have
\begin{eqs}
    W'_{\msquare_1^\vee, B} = \prod_{f|\msquare_1 \in \partial f} X_{B,f} \prod_{f'} Z_{B,f'}^{\int_N \delta \bmsquare_1 \cup_1 \bface'}
    \prod_{f''} Z_{A,f''}^{\int_N \bface'' \cup \bmsquare_1},
\label{eq: new plaquette term 2}
\end{eqs}
or equivalently
\begin{eqs}
    W'_{\msquare_1^\vee, B} \ket{\Phi_1, \Phi_2} = (-1)^{\int \delta \bmsquare_1 \cup_1 \Phi_2 + \Phi_2 \cup \bmsquare_1 } \ket{\Phi_1 , \Phi_2 + \delta \bmsquare_1}.
\end{eqs}
These plaquette terms are drawn in Fig.~\ref{fig:threeFermion_AlternateTerms}.

\begin{figure}[h!]
  \centering
  \includegraphics[width=0.7\linewidth]{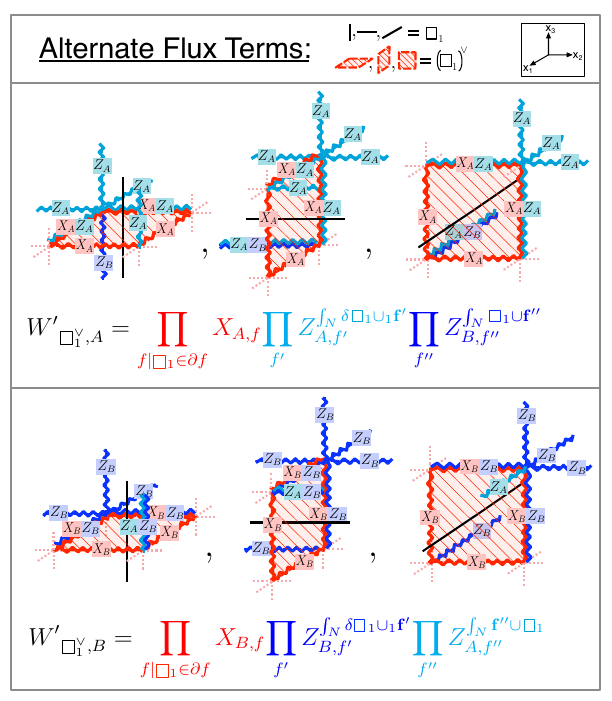}
  \caption[]{Alternate Flux terms $W'_{\msquare_1^\vee, A}, W'_{\msquare_1^\vee, B}$ of the 3-fermion model, derived using the $\cup_1$ product.}
  \label{fig:threeFermion_AlternateTerms}
\end{figure}

On the cubic lattice, these new plaquette can generate the original star terms and plaquette terms. We will derive this explicitly:
\begin{enumerate}
    \item We consider the product of $W'_{\msquare_1^\vee, A}$ on the six faces of a dual cube. The product is exactly a single $Z_B$ star term, which can be computed by directly multiplying the new plaquette terms in Fig.~\ref{fig:threeFermion_AlternateTerms}, which corresponds to replacing $\bmsquare_1 \ra \delta \bmsquare_0$ ($\msquare_0$ is the center point of the dual cube). 
    in \eqref{eq: new plaquette term 1}.
    Notice that this argument only holds for the cubic lattice, not for arbitrary triangulations. The special property of the cubic lattice is that there is 1-1 correspondence between vertices and cubes. For arbitrary triangulations, the product of $W'_{\msquare_1^\vee, A}$ may contain multiple $Z_B$ star terms.
    \item Similarly, the product of $W'_{\msquare_1^\vee, B}$ on the six faces of a dual cube gives a single $Z_A$ star term, so we are able to generate all star terms.
    \item Notice that the new plaquette terms in Fig.~\ref{fig:threeFermion_AlternateTerms} and original plaquette terms in Fig.~\ref{fig:threeFermion_Terms} are differed by a star term. By multiplying the proper star term on the new plaquette term, we are able to generate all original plaquette terms.
\end{enumerate}

We have shown that the 3-fermion Walker-Wang model can be expressed by the new plaquette terms (without star terms) using the $\cup_1$ product on the cubic lattice. In fact, we can show that these new plaquette terms are locally flippable separators. The flipper can be drawn explicitly and we can check that the QCA has order 2 without using computers. We leave the discussion of QCA construction to a future work~\cite{myfuturework}.

\subsection{Double-Semion model on a square lattice} \label{sec:doubleSemionSquareLattice}

In this section, we show how the cup product construction can help us phrase the double-semion model on a square lattice. Ref.~\cite{wangEtAl2019} recently gave a definition of the double-semion model on a square lattice that's equivalent, up to local unitary transformations, to the original double-semion~\cite{levinGu2012}. In particular, we construct a Hamiltonian equivalent to that of~\cite{wangEtAl2019} up to a local unitary transformation, and we construct the Hamiltonian associated to it.

\subsubsection{The Levin-Gu SPT}
We recall that the double-semion model is obtained by \textit{gauging} the $\Z_2$ symmetry associated to the $\Z_2$-SPT known as the Levin-Gu model~\cite{levinGu2012}. The Levin-Gu SPT is associated to a microscopic Hilbert space $\{\ket{0}_v,\ket{1}_v\}$ on each \textit{vertex} $v$, typically on some \textit{triangulated} manifold. The ground-state $\ket{\Psi}_\text{Levin-Gu}$ of this SPT is prepared from the trivial product state $\ket{\psi_0} = \bigotimes_v \frac{1}{\sqrt{2}}\big(\ket{0}_v + \ket{1}_v \big)$. Note $\ket{\psi_0}$ is an equal superposition of a product of $\ket{0},\ket{1}$ states on each vertex, so can be expressed as
\begin{equation}
    \ket{\psi_0} \propto \sum_{\tilde{a} \in C^0(N,\Z_2)} \ket{\tilde{a}}
\end{equation}
where $\ket{\tilde{a}}$ is the product-state $\ket{\tilde{a}} = \bigotimes_{v \in \text{vertices}} \ket{\tilde{a}(v)}_v$. Then, the Levin-Gu state is constructed as acting on the trivial state above with the ``$Z \cdot CZ \cdot CCZ$'' circuit~\cite{TV18}, defined as:
\begin{equation}
    ``Z \cdot CZ \cdot CCZ\text{''} := \left( \prod_{v \in \text{vertices}} Z_v \right)
\left( \prod_{\braket{v_0 v_1} \in \text{edges}} CZ_{\braket{v_0 v_1}} \right)
\left( \prod_{\braket{v_0 v_1 v_2} \in \text{triangles}} CCZ_{\braket{v_0 v_1 v_2}} \right)
\end{equation}
where $Z_v \ket{\tilde{a}} = (-1)^{\tilde{a}(v)} \ket{a}$, $CZ_{\braket{v_0 v_1}} \ket{a} = (-1)^{\tilde{a}(v_0) \tilde{a}(v_1)} \ket{\tilde{a}}$, and $CCZ_{\braket{v_0 v_1 v_2}} \ket{\tilde{a}} = (-1)^{\tilde{a}(v_0)\tilde{a}(v_1)\tilde{a}(v_2)} \ket{\tilde{a}}$. Then, the Levin-Gu SPT is given by
\begin{equation}
\ket{\Psi}_\text{Levin-Gu} =
 ``Z \cdot CZ \cdot CCZ\text{''} \ket{\psi_0}.
\end{equation}
The ground state is still an equal-probability superposition over all $\{\ket{\tilde{a}}\}_{\tilde{a} \in C^0(N,\Z_2)}$, except that the individual amplitudes are given a $\pm$ sign attached to each $\ket{\tilde{a}}$. In Appendix~\ref{app:cochainActionGenDoubleSemion}, we show that on an arbitrary 2-manifold triangulation $N$, the particular amplitudes have the form \cite{FH16, debray2019, FHHT20} \footnote{The generalized double semions model is proposed by Freedman and Hastings~\cite{FH16}. Its TQFT description and the relation with Stiefel-Whitney classes are studied in Ref.~\cite{debray2019}. The combinatorial formula under barycentric subdivision is proved in Ref.~\cite{FHHT20}.}:
\begin{equation}
\ket{\Psi}_\text{Levin-Gu} \propto \sum_{a \in C^0(N,\Z_2)} (-1)^{\int_N \tilde{a} \cup \delta \tilde{a} \cup \delta \tilde{a} + (w_1)(\tilde{a} \cup \delta \tilde{a}) + (w_2)(\tilde{a})} \ket{\tilde{a}}
\end{equation}
where $w_1 \in C_1(N,\Z_2)$, $w_2 \in C_0(N,\Z_2)$ are canonical chain-level representatives dual to the first and second Stiefel-Whitney classes~\cite{GT76}. Note that the topological action (acting on trivial 1-forms $\delta \tilde{a}$) associated to this wavefunction is
\begin{equation}
    S_{d.s.}(\delta \tilde{a}) = \frac{1}{2} \int_M \delta \tilde{a} \cup \delta \tilde{a} \cup \delta \tilde{a} + w_1(\delta \tilde{a} \cup \delta \tilde{a}) + w_2(\delta \tilde{a})
\end{equation}
which reproduces the ground-state amplitudes as above, since on a manifold with boundary, we have:
\begin{equation}
    S_{d.s.}(\delta \tilde{a}) = \frac{1}{2} \int_M \delta \Big( \tilde{a} \cup \delta \tilde{a} \cup \delta \tilde{a} + w_1(\tilde{a} \cup \delta \tilde{a}) + w_2(\delta \tilde{a}) \Big)
    = \frac{1}{2} \int_{N = \partial M} \tilde{a} \cup \delta \tilde{a} \cup \delta \tilde{a} + w_1(\tilde{a} \cup \delta \tilde{a}) + w_2(\tilde{a}).
\end{equation}
The action general cochains would be gotten by \textit{gauging} the above, thus replacing $\delta \tilde{a} \to a$:
\begin{equation}
    S_{d.s.}(a) = \frac{1}{2} \int_M a \cup a \cup a + w_1(a \cup a) + w_2(a).
\end{equation}

For our purposes, we want to generalize the above to fit with our cubical $\cup$ operations. And since the vector fields we use to define the thickening and shifting are constant throughout the cubical complex, we'll have that $w_1 \equiv w_2 \equiv 0$ \footnote{This follows from the obstruction-theoretic characterization of $w_k$ where it is Poincaré dual to the locus where $(d-k+1)$ `generic' vector fields on the manifold become non-independent, which is generically a codimension-$k$ submanifold. Since the frame of vector fields we use in the cubical complex is constant and thus always independent, we set all $w_k \equiv 0$.}. These terms have the interpretation of `gravitational terms', and can be readily converted to a Hamiltonian formulation on triangulations by the same methods presented here using their chain-level represtatives.

We also note that in Appendix~\ref{app:cochainActionGenDoubleSemion}, we derive a cochain-level action for the generalized double-semion model in all dimensions. The same methods presented here can be used to define Hamiltonians in essentially the same way, although we do not write these out explicitly.

\subsubsection{Double-semion model from gauging Levin-Gu SPT}
The double-semion model is related to the Levin-Gu SPT by \textit{gauging} a global $\Z_2$ symmetry. To describe this symmetry and how to gauge it, we first express $\ket{\Psi}_\text{Levin-Gu}$ in its more traditional form in terms of its `domain walls', following more closely the presentation of~\cite{FHHT20}.

Typically, double-semion models are defined on trivalent graphs on surfaces, which are dual to triangulations. In this perspective, each $c \in C^0(N,\Z_2)$ that goes into the wavefunction is usually referred to with respect to its dual regions, where there are regions of space labeled by `$0$' or `$1$' depending on whether $c(v) = 0,1$ in the vertex representing the dual region. Note that on this dual trivalent cellulation, the dual region can always be split up into \textit{separate} regions that don't meet at corners (whereas on a square lattice, for example, one may have that squares diagonal to each other meet at a corner). As such, the ``$Z \cdot CZ \cdot CCZ$'' circuit above acts on a state $\ket{a}$ as:
\begin{equation}
    ``Z \cdot CZ \cdot CCZ\text{''} \ket{a} = (-1)^{N_{v^\vee=1} - N_{e^\vee=1} + N_{f^\vee=1}} \ket{a} = (-1)^{\chi(`1\text{'-region})}.
\end{equation}
Above, we have that each $N_{v^\vee=1},N_{e^\vee=1},N_{f^\vee=1}$ are the number of faces (dual to vertices $v$), edges, vertices (dual to faces $f$) respectively that are part of the region labeled as `1'. And, $\chi(`1\text{'-region})$ is the Euler characteristic of the `1'-region. If the manifold $N$ is simply-connected, one can express the Euler characteristic of these disjoint regions of the plane as:
\begin{equation}
    (-1)^{\chi(`1\text{'-region})} = (-1)^{N_{d.w.}}
\end{equation}
where $N_{d.w.}$ is the number of domain walls between the $0$ and $1$ regions. 

From here, note that all wave-function amplitudes (on a simply connected space) can be expressed entirely in terms of the domain walls between the $0$ and $1$ regions, so they're invariant under switching $0 \leftrightarrow 1$ which is a $\Z_2$ global symmetry generated by the operator $\bigotimes_v X_v$. 

So, to `gauge' the Levin-Gu Hamiltonian means finding a Hamiltonian in terms of a state-space associated to link variables $e=\braket{ij}$, so $\mathcal{H} = \bigotimes_e \mathcal{H}_{e}$ where $\mathcal{H}_e = \mathrm{Span}(\ket{0}_e,\ket{1}_e)$, and such if $N$ is simply-connected or if the domain wall lines are topologically trivial, the ony ground state $\ket{\Psi}_{d.s.}$ is
\begin{equation}
    \ket{\Psi}_{d.s.} \propto \sum_{\Phi \in C^1(N,\Z_2) | \delta \Phi = 0} (-1)^{N_{loops}} \ket{\Phi}
\end{equation}
where $\ket{\Phi}$ is state associated to the $\Phi \in C^1(M,\Z_2)$ as before. Note that we express the amplitude in terms of $N_{loops}$, the number of closed loops in the loop configuration dual to $\Phi$, because gauging a symmetry entails viewing these loops (which used to be domain walls) as the underlying degrees of freedom.

\subsubsection{Local constraints and Hamiltonian on square lattice}
Now, we follow the strategy to find local constraints for the double-semion model to derive a Hamiltonian.  Since we're working on a simply-connected spatial manifold $N$, we can assume that any $\Phi \in Z^1(N,\Z_2) = \delta \tilde{a}$ for some $\tilde{a} \in C^0(N,\Z_2)$. Then following the above (recalling that $w_2 = w_1 = 0$), we'll have that we can express the double-semion ground state as
\begin{equation}
    \ket{\Psi}_{d.s.} \propto \sum_{\tilde{a} \in C^0(N,\Z_2)} (-1)^{\int_N \tilde{a} \cup \delta \tilde{a} \cup \delta a} \ket{\delta \tilde{a}}.
\end{equation}
Note that amplitudes of the above depend on the region $a$ that trivializes $\Phi$ so the amplitude itself is not local in terms of the loop configuration. However, what we really want is an expression such that the phase shift from \textit{local} changes $\ket{\Phi} \to \ket{\Phi + \delta {\bf v}}$ can be expressed entirely in terms of $\Phi$ and the indicator ${\bf v}$. We'll talk in terms of the change $2 (\Delta S)$, where $S$ is the action used to define the amplitude $e^{2 \pi i S}$. In fact, such a local expression does exist in terms of the $\cup_1$ product. Start with
\begin{equation}
\begin{split}
    2(\Delta S) &:= \int (\tilde{a}+{\bf v}) \cup \delta(\tilde{a}+{\bf v}) \cup \delta(\tilde{a}+{\bf v}) - \tilde{a} \cup \delta \tilde{a} \cup \delta \tilde{a} \\
    &= \int {\bf v} \cup \delta \tilde{a} \cup \delta \tilde{a} + \tilde{a} \cup \delta {\bf v} \cup \delta \tilde{a} + {\bf v} \cup \delta {\bf v} \cup \delta \tilde{a} + \tilde{a} \cup \delta \tilde{a} \cup \delta {\bf v} + {\bf v} \cup \delta \tilde{a} \cup \delta {\bf v} + \tilde{a} \cup \delta {\bf v} \cup \delta {\bf v} + {\bf v} \cup \delta {\bf v} \cup \delta {\bf v}
\end{split}
\end{equation}
Then, we can massage the above expression in terms of \textit{only} $\delta a$ as:
\begin{equation} \label{eqn:DoubleSemionAmplitudeChange}
\begin{split}
    2(\Delta S) &= \int {\bf v} \cup (\delta \tilde{a} \cup \delta \tilde{a}) + \delta \tilde{a} \cup {\bf v} \cup \delta \tilde{a} + {\bf v} \cup (\delta {\bf v} \cup \delta \tilde{a}) + (\delta \tilde{a} \cup \delta \tilde{a}) \cup (\delta {\bf v} \cup \delta \tilde{a}) \cup {\bf v} + \delta \tilde{a} \cup \delta {\bf v} \cup {\bf v} + {\bf v} \cup \delta {\bf v} \cup \delta {\bf v} \\
    &= \int \delta {\bf v} \cup_1 (\delta \tilde{a} \cup \delta \tilde{a}) + \delta {\bf v} \cup_1 (\delta {\bf v} \cup_1 \delta \tilde{a}) + \delta \tilde{a} \cup {\bf v} \cup \delta \tilde{a} + (\delta \tilde{a} \cup \delta {\bf v} \cup {\bf v}) + {\bf v} \cup \delta {\bf v} \cup \delta {\bf v} \\
    &= \int \underbrace{\delta \tilde{a} \cup {\bf v} \cup \delta \tilde{a}}_{\text{term i}} + \underbrace{\delta {\bf v} \cup_1 ((\delta \tilde{a} + \delta {\bf v}) \cup \delta \tilde{a})}_{\text{term ii}} + \underbrace{\delta \tilde{a} \cup \delta {\bf v} \cup {\bf v}}_{\text{term iii}} + \underbrace{{\bf v} \cup \delta {\bf v} \cup \delta {\bf v}}_{\text{term i{\bf v}}}.
\end{split}
\end{equation}
Above, the first line used some integrations-by-parts, the second line used the $\cup_1$ identity to combine some terms, and the third line simply combines various terms. At this point, all quantities are in terms of $\delta \tilde{a} = \Phi$.

Up until here, our discussions and expressions in terms of cochains have been completely general. But now, we restrict our attention to the square lattice. In particular, for ${\bf v}$ the indicator on a single vertex of the square lattice, it's a quick computation to show that $\int {\bf v} \cup \delta {\bf v} \cup \delta {\bf v} = 0$. As such, `term iv' becomes zero on the square lattice.

Now, we construct the Hamiltonian $H^{d.s.}$ for the model. As in the other examples, we'll have that the Hamiltonian splits into a set of commuting `star terms' and `flux terms'. The star terms will be the same as the two-dimensional toric-code,
\begin{equation}
    U^{d.s.}_{f^\vee} = U^{d.s.}_{f^\vee} = \prod_{e^\vee \supset (f^\vee)} Z_e
\end{equation}
and are again meant to enforce $\delta \Phi = 0$, so that any term in the ground state must correspond to a loop configuration. The `flux terms' will serve the purpose to constrain $\ket{\Psi}$'s ground state amplitudes as in Eq.~\eqref{eqn:DoubleSemionAmplitudeChange} (with $\delta a \to \Phi$):
\begin{equation}
    \braket{\Phi | \Psi} = \braket{\Phi + \delta {\bf v} | \Psi} (-1)^{\int \Phi \cup {\bf v} \cup \Phi + \delta {\bf v} \cup_1 ((\Phi + \delta {\bf v}) \cup \Phi) + \Phi \cup \delta {\bf v} \cup {\bf v}}
\end{equation}
To engineer such a term in the Hamiltonian, we'll consider a flux term implementing $\ket{\Phi} \to \pm \ket{\Phi + \delta {\bf v}}$ of the form:
\begin{equation}
   \left( \prod_{f^\vee \subset v^\vee} \left( \frac{1 + U^{d.s.}_{f^\vee}}{2} \right) \right) W^{d.s.}_{v^\vee}  \left( \prod_{f^\vee \subset v^\vee} \left( \frac{1 + U^{d.s.}_{f^\vee}}{2} \right) \right).
\end{equation}
Note that each $\prod_{f^\vee \subset v^\vee} \left( \frac{1 + U_{f^\vee}}{2} \right)$ is a projector onto the set of states with $\delta \Phi = 0$ on dual vertices $f^\vee$ that are a subset of $v^\vee$. Given this projector (and using that $\int {\bf v} \cup \delta {\bf v} \cup \delta {\bf v} = 0$), we'll need that the $W^{d.s.}_{v^\vee}$ satisfy
\begin{equation}
    W^{d.s.}_{v^\vee} \ket{\Phi} = \ket{\Phi + \delta {\bf v}} (-1)^{\int \Phi \cup {\bf v} \cup \Phi + \delta {\bf v} \cup_1 ((\Phi + \delta {\bf v}) \cup \Phi) + \Phi \cup \delta {\bf v} \cup {\bf v}}
\end{equation}
for all \textit{closed} $\Phi$. Such flux terms are computed in Fig.~\ref{fig:doubleSemionModel}. In that figure, we have in the $\{\ket{0},\ket{1}\}$ basis that $S = \begin{pmatrix} 1 & 0 \\ 0 & i \end{pmatrix}$ and $\overline{S} = \begin{pmatrix} 1 & 0 \\ 0 & -i \end{pmatrix}$ are two `quarter-phase-gates'. The choice of $W^{d.s.}_{v^\vee}$ we make depends on whether $v$ is on the even or odd sublattice of the square grid. This choice wasn't strictly necessary, but we made this choice to prove our Hamiltonian is equivalent to the one in~\cite{wangEtAl2019}. The equivalence of the flux terms of~\cite{wangEtAl2019} to ours is in Fig.~\ref{fig:doubleSemion_EqualsWangEtAl}.

\begin{figure}[h!]
    \centering
    \begin{minipage}{\textwidth}
        \centering
        \includegraphics[width=\linewidth]{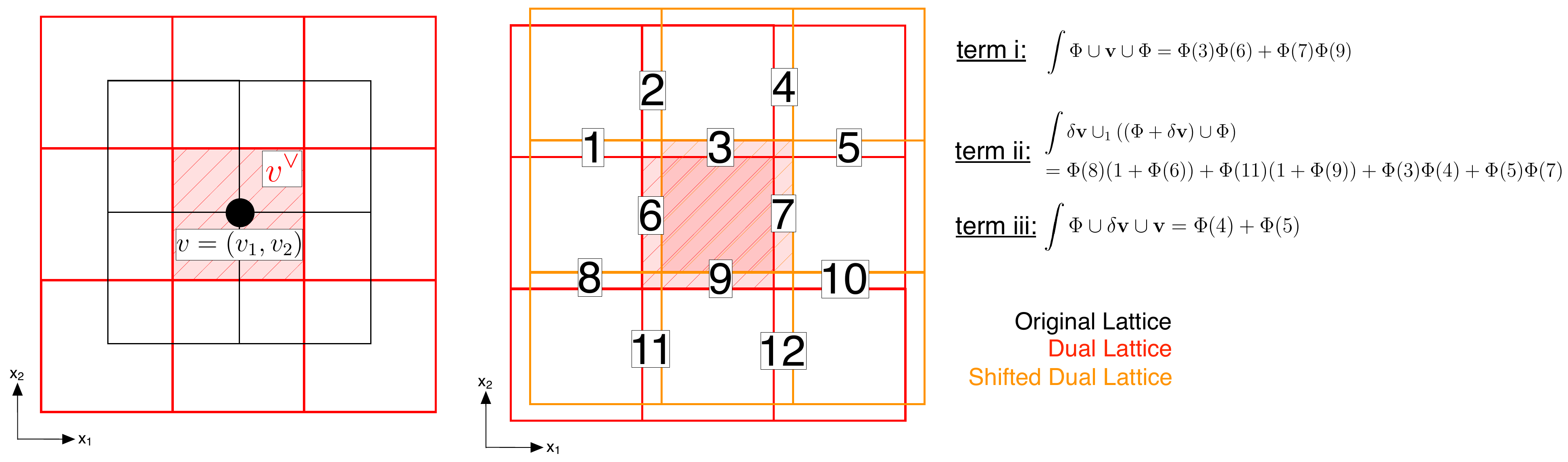}
    \end{minipage}
    \begin{minipage}{\textwidth}
         \centering
         \includegraphics[width=\linewidth]{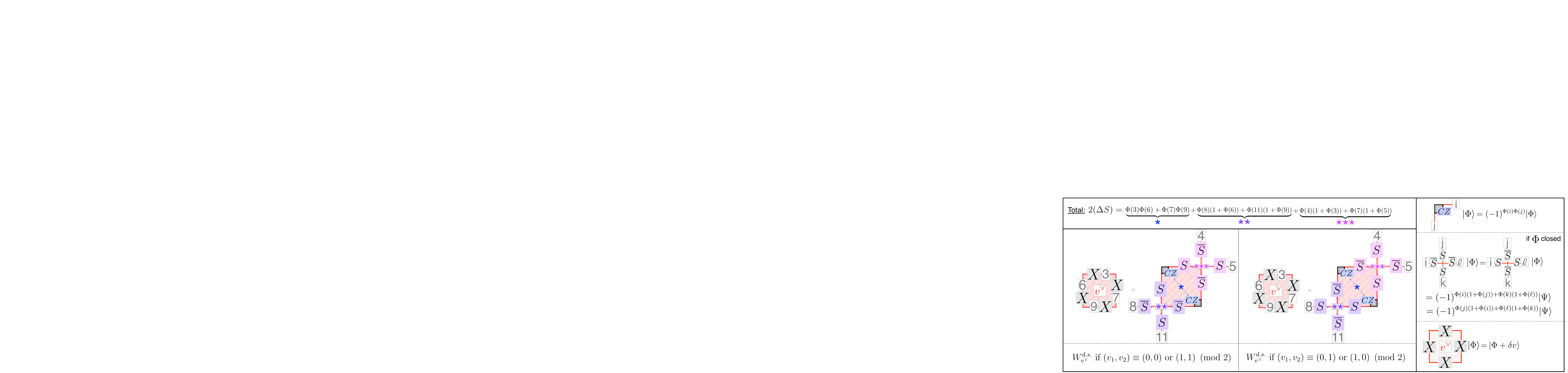}
    \end{minipage}
    \caption{Double-semion model on a square lattice. (Top) Explicitly computing the terms of $\int \Phi \cup {\bf v} \cup \Phi + \delta {\bf v} \cup_1 ((\Phi + \delta {\bf v}) \cup \Phi) + \Phi \cup \delta {\bf v} \cup {\bf v}$ with respect to the relevant dual edges, labeled as $1,\cdots,12$. The shifted lattice is shown for visual aid in considering the cup product pairing, as in Fig.~\ref{fig:cupProductImages}. (Bottom) The choice of $W^{d.s.}_{v^\vee}$ we make, which depend on whether $v$ is on the even or odd sublattices, together with the calculations that show they product the right local change in wave-function amplitude. The equalities involving the phase gates $S,\overline{S}$ require that $\Phi$ be closed and can be verified by case-work. There are many choices of $W^{d.s.}_{v^\vee}$ that reproduce the same expression, but we choose ones that let us match~\cite{wangEtAl2019} up to a local unitary circuit, as in Fig.~\ref{fig:doubleSemion_EqualsWangEtAl}.}
    \label{fig:doubleSemionModel}
\end{figure}

\begin{figure}[h!]
  \centering
  \includegraphics[width=\linewidth]{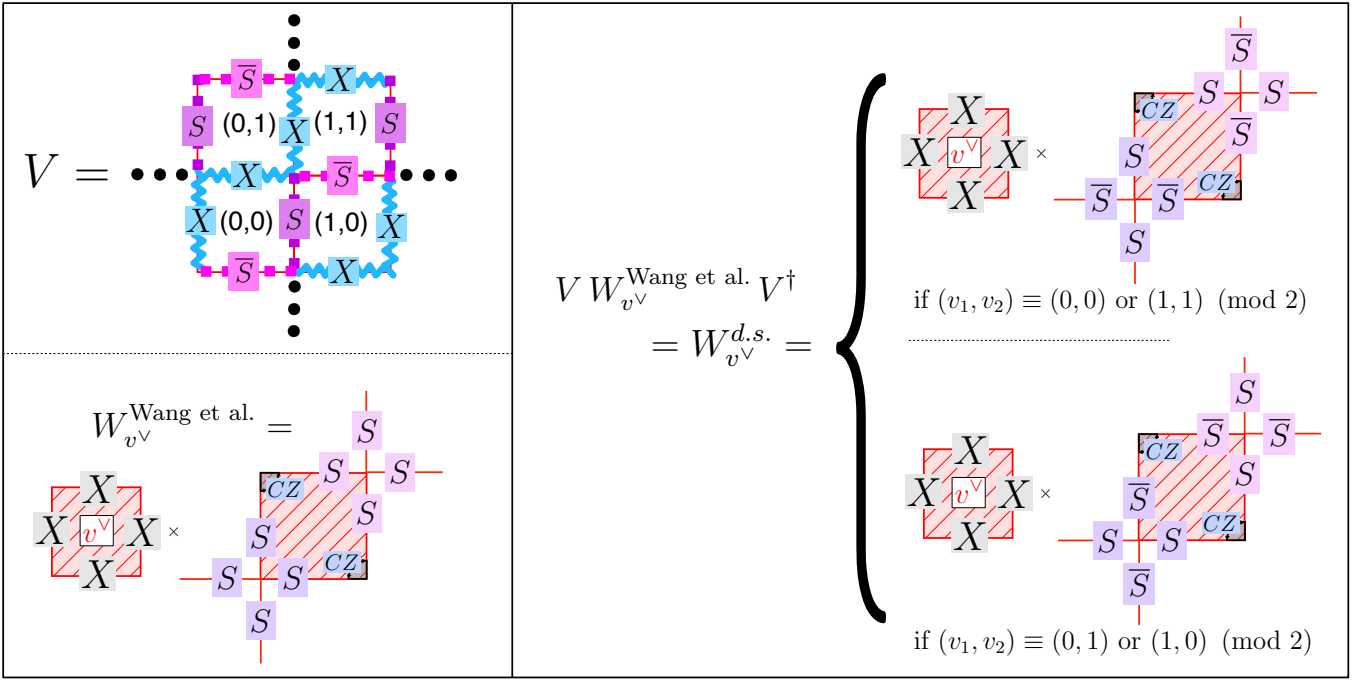}
  \caption{(Top-Left) A circuit $V$ consisting of a product of local unitaries. Here, $(0,0),(1,0),(0,1),(1,1)$ in the squares represent the $(x_1,x_2)$ coordinates of the vertices dual to the squares, and the circuit is a product of the shown one but periodically extended throughout the plane. There are two disjoint sets of `zig-zags'. On one of them, there are only $X$ operators. On the other, there are $S$ operators on vertical edges, and $\overline{S}$ operators on the horizontal ones. (Bottom-Left) The flux terms used in~\cite{wangEtAl2019} to define the double-semion model on a square lattice. (Right) Showing that $V W^{\text{Wang et al.}_{v^\vee}} V^\dagger = W^{d.s.}$. This equality uses that $S^\dagger = \overline{S}$, $X S X= i \overline{S}$, $(X_1 X_2) (CZ_{12}) (X_1 X_2) = Z_1 Z_2 CZ_{12}$, and $Z S = \overline{S}$, where $CZ_{12}$ is the controlled-Z operator acting on degrees of freedom $1,2$, and $X_{1,2}$ are the Pauli-$X$ operators acting on $1,2$.}
  \label{fig:doubleSemion_EqualsWangEtAl}
\end{figure}

\subsection{`Fermionic' toric code in arbitrary dimensions and duality with trivial fermionic insulator} \label{sec:fermionicToricCode}

Following~\cite{kapustinThorngren2017,chen2020BosonArbDimensions}, one can define a `twisted' version of the toric code called the `fermionic toric code' in $(d+1)$-dimensions. The action is written in terms of a $(d-2)$-form gauge field $a \in C^{d-2}(M,\Z_2)$ and a Lagrange multiplier $b$ enforcing $\delta a = 0$. The action given by
\begin{equation}
\begin{split}
    S_{f.t.c.} &= \frac{1}{2} \int_{M} a (\partial b) + a \cup_{d-3} a \\
    &= \int_{M^{d+1}} a (\partial b) + \mathrm{Sq}^2(a).
\end{split}
\end{equation}
where $\mathrm{Sq}^2(a) := a \cup_{d-3} a$ is the so-called `Steenrod square' term. This model is a kind of higher-form $\Z_2$ gauge theory; on a spatial slice, the field `$a$' is dual to a loop configuration living on the dual cell decomposition a the ground state on a manifold will be a sum over such closed-loop configurations weighted by signs. 

As a particular example, in $(3+1)$ dimensions, this model can be thought of as the Crane-Yetter/Walker-Wang model with an input premodular category $\{1,\psi\}$, where `$1$' is the identity anyon and $\psi$ is a $\Z_2$-object satisfying $\psi \times \psi = 1$ and has fermionic statistics, $\theta_\psi = -1$. On a triangulated bulk $M^{3+1}$, the bulk amplitudes on a 4-simplex $\braket{01234}$ can be written as $(a \cup_0 a)=a(012)a(234)$ where $a$ is a cocycle that tells us when to decorate a $15j$ symbol with the $\psi$ object.

Now let's consider the local constraints on wavefunction amplitudes that this bulk action gives for trivial loop configurations. Suppose $\Phi = \delta \lambda$ is a trivial wavefunction amplitude. Then we can use the cochain-level identity 
\begin{equation}
    \mathrm{Sq}^2(\delta \lambda) = (\delta \lambda) \cup_{d-3} (\delta \lambda) = \delta(\lambda \cup_{d-3} \delta\lambda + \lambda \cup_{d-4} \lambda)
\end{equation}
and integrate on the bulk $M^{d+1}$ with boundary $\partial M = N$ to give
\begin{equation}
    \int_M \mathrm{Sq}^2(\delta \lambda) = \int_N \lambda \cup_{d-3} \delta\lambda + \lambda \cup_{d-4} \lambda.
\end{equation}
Finally, comparing to Eq.~\eqref{eq:GrassmannIntegral_TrivialLoops} we note that this last equation gives that the wavefunction amplitudes that precisely equal the Gu-Wen Grassmann integral $\sigma(\delta \lambda)$ evaluated on $\delta \lambda$. The local constraint on wave-function amplitudes is fixed by the quadratic refinement property Eq.~\eqref{eq:quadRefineSigma}, that $\sigma(\Phi + \delta \bmsquare_{d-2}) = \sigma(\Phi) \sigma(\delta \bmsquare_{d-2}) (-1)^{\int \Phi \, \cup_{d-2} \, \delta\bmsquare_{d-2}} = \sigma(\Phi) (-1)^{\int \Phi \, \cup_{d-2} \, \delta\bmsquare_{d-2}}$ where the last equality uses Eq.~\eqref{eq:sigmaElLoops}. This is equivalent to the requirement that wavefunction amplitudes of a ground state $\ket{\Psi}$ on the spatial manifold must satisfy
\begin{equation}
    \braket{\Phi + \delta \bmsquare_{d-2} | \Psi} = \braket{\Phi | \Psi} (-1)^{\int \Phi \, \cup_{d-2} \, \delta\bmsquare_{d-2}}.
\end{equation}

Now, the job is to engineer a commuting projector Hamiltonian that realizes this. The Hamiltonian will be in terms of the matrices used in the exact bosonization identities Eqs.~(\ref{eq:exactbosonization1}-\ref{eq:exactbosonization3}). Specifically we'll have
\begin{equation}
    H_{f.t.c.} = - J_1\sum_{\msquare_d} W_{\msquare_d} - J_2\sum_{\msquare_{d-2}} G_{\msquare_{d-2}} 
\end{equation}
where $J_{1,2}$ are arbitrary coupling constants. The first set of terms $W_{\msquare_d}$ enforce $\delta \Phi = 0$ in the ground states. the second set $G_{\msquare_{d-2}}$ the quadratic refinement. 

Note that taking $J_2 \to \infty$ dynamically enforces the gauge constraint $G_{\msquare_{d-2}}\ket{\Psi} = \ket{\Psi}$ on the bosonic side that then allows us to `fermionize' the system, i.e. use the exact bosonization map Eqs.~(\ref{eq:exactbosonization1}-\ref{eq:exactbosonization3}) to write it purely in terms of a fermionic Hamiltonian (at least on a spatial manifold). The dual fermionic Hamiltonian $\tilde{H}_{f.t.c}$ is actually a trivial one:
\begin{equation}
    \tilde{H}_{f.t.c} = i J_1 \sum_{\msquare_d} \gamma_{\msquare_d} \gamma'_{\msquare_d}.
\end{equation}
If we want to add `hopping' terms to the Hamiltonian, we could just as well use Eq.~\eqref{eq:exactbosonization2} of the exact bosonization dictionary and add the corresponding terms on both sides.

\section{Discussion}
We've introduced higher cup product operations over both $\Z_2$ and $\Z$ on hypercubic lattices motivated by recent geometric interpretations on triangulations. 
Over $\Z$ our definitions of various sign factors that show up are well-motivated by considering a signed intersection products of the thickened dual cochains. While our derivations of the recursion relations are long combinatorial exercises, we expect there to be more straightforward geometric derivation of the results. One technical and conceptual obstacle to this is the presence of lower-dimensional cells in the thickened intersection equations that do not appear in the higher-cup equations. A more geometric derivation of the recursion relations would presumably need to deal with them. And in general, it would be interesting to see if the lower-dimensional cells have any topological interpretations.

In addition, we've shown how it serves as a useful formalism to think about and derive many lattice constructions of topological phases on a hypercubic lattice, in particular giving expositions on how certain well-known models can be analyzed through the lens of the cup product formalisms. Although we gave an explicit account of this for the 3-fermion Walker-Wang model~\cite{BCFV14} in the presence of a two-dimensonal spatial boundary, the same ideas should also apply for Crane-Yetter/Walker-Wang constructions of any abelian topological orders. It would be interesting to see if earlier works studying defects in two-dimensional systems and their gapped of boundaries (see e.g.~\cite{BJQ2013_abelianTopOrder}) can be reframed in terms of cup products.

Because of the amenability of hypercubic lattices to Fourier analysis, our constructions may be useful in direct comparisons of topological phases of free and interacting systems. In addition, studies of dynamical properties of topological lattice Hamiltonians may be more readily accessible. For example, the hypercubic construction of the fermionic toric code in Section~\ref{sec:fermionicToricCode} is dual to free fermionic systems, and it may be interesting to study textbook many-body physics of fermionic systems through the lens of their bosonic duals.

Also, our thickening and shifting prescriptions on dual cellular decompositions for constructing higher cup products should apply on more general cellular decompositions, for example crystalline lattice systems. It may be interesting to flesh these out and see if crystalline topological phases (see e.g.~\cite{thorngrenElse2018_crystalline,manjunath2020classification,debray2021invertible} and references therein) can be analyzed through this lens.

\section*{Acknowledgments}
Y.C. thanks Po-Shen Hsin and his PhD advisor Prof. Anton
Kapustin for many very helpful suggestions. Y.C. also thanks Nathanan Tantivasadakarn for his help on this project. The discussion with Tyler Ellison, Arpit Dua, Wilbur Shirley, and Dom Williamson are very useful. Y.C. thank to the Condensed Matter Theory Center and the Joint Quantum Institute at the University of Maryland for hosting and funding.

S.T. wants to thank Maissam Barkeshli for discussions about the 3-fermion model, and also Danny Bulmash and Ryohei Kobayashi for a related collaboration.

\appendix



\section{Cochain action for generalized double-semion model} \label{app:cochainActionGenDoubleSemion}
The generalized double-semion model~\cite{FH16} is a generalization to arbitrary dimensions of the double-semion model. In this section, we derive a cochain-level action for the generalized double-semion model, following the result of~\cite{debray2019} which derived a topological action at the cohomology level. The cochain-level action is discussed for the barycentric subdivision~\cite{FHHT20} and we give a proof for arbitrary triangulations in all dimensions.

This cochain-level action would allow one to derive Hamiltonians in a similar way as described in Sec.~\ref{sec:doubleSemionSquareLattice}.

The generalized double-semion model in $(d+1)$-dimensions is gotten by gauging a global $\Z_2$ symmetry of a certain $\Z_2$-SPT, which we call the `Generalized Levin-Gu SPT'.
We are given $(N^d,T)$ triangulation $T$ of the $d$-dimensional \textit{spatial} manifold $N$. Let $T^k$ denote the $k$-simplices of the triangulation. The $\Z_2$-SPT is formed in the setting of a two-state system $\{\ket{0}_v, \ket{1}_v\}$ at each vertex $v \in T^0$. 

For a cochain $\tilde{a} \in C^0(N,\Z_2) $, we denote $\ket{\tilde{a}} \in \bigotimes_{v \in T^0} \{\ket{0}_v, \ket{1}_v\}$ denote the product state $\ket{\tilde{a}} = \bigotimes_{v \in T^0} \ket{\tilde{a}(v)}$. And for $0 \le k \le d$ and some $k$-simplex $\braket{0 \cdots k}$, let $C^{(k)}Z_{\braket{0 \cdots k}}$ be the operator:
\begin{equation}
    C^{(k)}Z_{\braket{0 \cdots k}} \ket{\tilde{a}} := (-1)^{\tilde{a}(0)\tilde{a}(1) \cdots \tilde{a}(k)} \ket{\tilde{a}}.
\end{equation}
We'll interchangeably use the notation
\begin{equation}
    C^{(k)}Z_{\braket{0 \cdots k}} \equiv \underbrace{C \cdots C}_{k \text{ times}} Z_{\braket{0 \cdots k}}.
\end{equation}
For example, $C^{(0)}Z_{\braket{0}} \equiv Z_{\braket{0}}$, $C^{(1)}Z_{\braket{01}} \equiv CZ_{\braket{01}}$, $C^{(2)}Z_{\braket{012}} \equiv CCZ_{\braket{012}}$, etc. 

To define the ground state of the $\Z_2$-SPT, first define the circuit
\begin{equation}
    ``Z \times C^{(1)}Z \times \cdots \times C^{(d)} Z \text{''} := \left( \prod_{\braket{0} \in T^0} Z_{\braket{0}} \right) \left( \prod_{\braket{01} \in T^1} CZ_{\braket{01}}\right) \cdots \left( \prod_{\braket{0 \cdots d} \in T^d} C^{(d)}Z_{\braket{0 \cdots d}}\right)
\end{equation}
Then, the ground state of the generalized Levin-Gu SPT will be the above circuit acting on an equal superposition of $\{0,1\}$ spins at each vertex. I.e., define
\begin{equation}
\begin{split}
    \ket{\Psi}_{\text{Levin-Gu}(d)} &\propto ``Z \times C^{(1)}Z \times \cdots \times C^{(d)} Z \text{''}
    \left( \sum_{\tilde{a} \in C^0(N,\Z_2)} \ket{\tilde{a}} \right) \\
    &= \sum_{\tilde{a} \in C^0(N,\Z_2)} \ket{\tilde{a}} \cdot (-1)^{\sum_{\braket{0} \in T^0} \tilde{a}(0) + \cdots + \sum_{\braket{0 \cdots d} \in T^d} \tilde{a}(0)\cdots \tilde{a}(d)} 
\end{split}
\end{equation}
We can express the amplitudes of $\ket{\tilde{a}}$ in the ground-state wavefunction in terms of the Poincaré dual of $\tilde{a} \in C^0(M,\Z_2)$. In particular, the dual of $\tilde{a}$ will be the region of $M$ consisting of the duals of the `1-region'. By essentially the same argument as in Sec.~\ref{sec:doubleSemionSquareLattice}, we'll have 
\begin{equation}
     \ket{\Psi}_{\text{Levin-Gu}(d)} \propto  \sum_{\tilde{a} \in C^0(N,\Z_2)} \ket{\tilde{a}} \cdot (-1)^{\chi(\text{`1-region'} ; \tilde{a})}
\end{equation}
where $(-1)^{\chi(\text{`1-region'} ; \tilde{a})}$ encodes the (mod 2) Euler characteristic of the `1-region' dual to $\tilde{a}$.

Note that all of the discussion above in defining $\ket{\Psi}_{\text{Levin-Gu}, d}$ didn't require us to think about branching structures, since all of the $C^{(k)}Z$ gates don't depend on any orderings of vertices. Now, say that we had a triangulation with branching structure instead. Then, our claim is that we can express the amplitudes above as 
\begin{prop} \label{lem:generalizedLevinGuLemma}
Let $\tilde{a} \in C^0(N,\Z_2)$ in some branched triangulation of $N$. Then, the generalized Levin-Gu SPT wavefunction amplitudes
\begin{equation}
    (-1)^{\sum_{\braket{0} \in T^0} \tilde{a}(0) + \cdots + \sum_{\braket{0 \cdots d} \in T^d} \tilde{a}(0) \cdots \tilde{a}(d)} 
    = (-1)^{\chi(\text{`1-region'} ; \tilde{a})}
\end{equation}
can be expressed as
\begin{equation} \label{eq:generalizedLevinGuAmplitudes}
    \chi(\text{`1-region'} ; \tilde{a}) = \int_{N^d} w_d(\tilde{a}) + w_{d-1}(\tilde{a} \cup \delta \tilde{a}) + \cdots + w_1(\tilde{a} \cup \underbrace{\delta \tilde{a} \cup \cdots \cup \delta \tilde{a}}_{d-1 \text{ times}}) + \tilde{a} \cup \underbrace{\delta \tilde{a} \cup \cdots \cup \delta \tilde{a}}_{d \text{ times}} \,\text{ (mod 2)}.
\end{equation}
\end{prop}
In the above proposition, each $w_k \in C_{d-k}(N,\Z_2)$ are cycles dual to the $k^{th}$ Stiefel Whitney classes, which are canonically defined~\cite{GT76} given the branched-triangulated $N$, and each $w_k(\tilde{a} \cup \underbrace{\delta \tilde{a} \cup \cdots \cup \delta \tilde{a}}_{d-k \text{ times}})$ is the chain-cochain pairing.

Note that in general, $\chi(\text{`1-region'} ; \tilde{a}) = \chi(\text{`0-region'} ; \tilde{a}) + \chi(N) \text{ (mod 2)}$, where $\chi(N)$ is the Euler characteristic of $N$. As such, there's a global $\Z_2$ symmetry generated by flipping all spins $0 \leftrightarrow 1$ and multiplying the wavefunction by $(-1)^{\chi(N)}$.

Before moving on to the proof of the formula, we first give the action associated to \textit{gauging} the above symmetry, similar to the main text, in Sec.~\ref{sec:doubleSemionSquareLattice}. First, the topological action on manifold $M$ for trivial 1-forms would be
\begin{equation}
    S_{GDS(d)}(\delta \tilde{a}) = \frac{1}{2} \int_{M^{d+1}} w_d(\delta \tilde{a}) + w_{d-1}(\delta \tilde{a} \cup \delta \tilde{a}) + \cdots + w_1(\underbrace{\delta \tilde{a} \cup \cdots \cup \delta \tilde{a}}_{d \text{ times}}) + \underbrace{\delta \tilde{a} \cup \cdots \cup \delta \tilde{a}}_{d+1 \text{ times}}
\end{equation}
which would give the same amplitudes as Eq.~\eqref{eq:generalizedLevinGuAmplitudes} after setting $\partial M = N$. As such, the general action associated to gauging the symmetry would be gotten by replacing $\delta \tilde{a} \to a$, so
\begin{equation}
    S_{GDS(d)}(a) = \frac{1}{2} \int_{M^{d+1}} w_d(a) + w_{d-1}(a \cup a) + \cdots + w_1(\underbrace{a \cup \cdots \cup a}_{d \text{ times}}) + \underbrace{a \cup \cdots \cup a}_{d+1 \text{ times}}.
\end{equation}
This is exactly the topological action derived in~\cite{debray2019}.

Before proving the main proposition of this Appendix, we first review the chain-level definitions of $w_k$ and then give usable expressions for $\tilde{a} \cup \delta \tilde{a} \cup \cdots \cup \delta \tilde{a}$.

\subsection{Formulas for $w_k$ and $\tilde{a} \cup \delta \tilde{a} \cup \cdots \cup \delta \tilde{a}$}
We review the formula for the dual of $w_k$, following~\cite{GT76}. In this section, we will use notation so that a $k$-simplex $\braket{v_0 \cdots v_k}$ is written in the order of the branching structure, so that $v_0 \to \cdots \to v_k$ in the branching structure.

First, for some $q$-simplex $t = \braket{0 \cdots q}$, define the chains $\partial_p(t) \in C_p(N,\Z_2) , p \le q$ as follows:
\begin{equation} \label{eq:regularSimpFullExpr}
\begin{split}
    \partial_p(\braket{0 \cdots q}) &=  \sum_{1 \le i_1 < i_2 - 1 < \cdots < i_k-(k-1) \le q-k }  \braket{0 \, \underbrace{i_1 \, i_1 + 1} \, \underbrace{i_2 \, i_2 + 1} \cdots \underbrace{i_k \, i_k + 1}} \text{ if } p = 2k \text{ is even} \\
    \partial_p(\braket{0 \cdots q}) &=  \sum_{1 \le i_1 < i_2 -1 < \cdots < i_k - (k-1) \le q-k-1}  \braket{0 \, \underbrace{i_1 \, i_1 + 1} \, \underbrace{i_2 \, i_2 + 1} \cdots \underbrace{i_k \, i_k + 1} \, q} \text{ if } p = 2k + 1 \text{ is odd} \\
\end{split}
\end{equation}
where the underbraces are just for visual aid in showing which pairs of arguments of vertices must be `one apart' from each other. Ref.~\cite{GT76} called these chains sums over so-called \textit{regular} $p$-simplices $s \subset t$. As an explicit example, these chains with a 6-simplex as an argument would be:
\begin{equation}
\begin{split}
    \partial_0(\braket{0123456}) &= \braket{0} \\
    \partial_1(\braket{0123456}) &= \braket{06} \\
    \partial_2(\braket{0123456}) &= \braket{0\underbrace{12}} + \braket{0\underbrace{23}} + \braket{0\underbrace{34}} + \braket{0\underbrace{45}} + \braket{0\underbrace{56}} \\
    \partial_3(\braket{0123456}) &= \braket{0\underbrace{12}6} + \braket{0\underbrace{23}6} + \braket{0\underbrace{34}6} + \braket{0\underbrace{45}6} \\
    \partial_4(\braket{0123456}) &= \braket{0\underbrace{12}\underbrace{34}} + \braket{0\underbrace{12}\underbrace{45}} + \braket{0\underbrace{12}\underbrace{56}} + \braket{0\underbrace{23}\underbrace{45}} + \braket{0\underbrace{23}\underbrace{56}} + \braket{0\underbrace{34}\underbrace{56}} \\
    \partial_5(\braket{0123456}) &= \braket{0\underbrace{12}\underbrace{34}6} + \braket{0\underbrace{12}\underbrace{45}6} + \braket{0\underbrace{23}\underbrace{45}6} \\
    \partial_6(\braket{0123456}) &= \braket{0\underbrace{12}\underbrace{34}\underbrace{56}} \\
    \partial_k(\braket{0123456}) &= 0 \text{ if } k>6
\end{split}
\end{equation}

Then, the explicit formula for $w_{d-p} \in C_p(N,\Z_2)$ is
\begin{equation}
    w_{d-p} = \sum_{q > p}\sum_{\braket{0 \cdots q} \in T^q} \partial_p(t).
\end{equation}
As examples, in $d=2$, we have
\begin{equation*}
\begin{split}
    w_2 &= \sum_{\braket{0} \in T^0} \braket{0} + \sum_{\braket{01} \in T^1} \braket{0} + \sum_{\braket{012} \in T^2} \braket{0} \\
    w_1 &= \sum_{\braket{01} \in T^1} \braket{01} + \sum_{\braket{012} \in T^2} \braket{02} \\ 
    w_0 &= \sum_{\braket{012} \in T^2} \braket{012},
\end{split}
\end{equation*}
and in $d=3$, 
\begin{equation*}
\begin{split}
    w_3 &= \sum_{\braket{0} \in T^0} \braket{0} + \sum_{\braket{01} \in T^1} \braket{0} + \sum_{\braket{012} \in T^2} \braket{0} + \sum_{\braket{0123} \in T^3} \braket{0} \\
    w_2 &= \sum_{\braket{01} \in T^1} \braket{01} + \sum_{\braket{012} \in T^2} \braket{02} + \sum_{\braket{0123} \in T^3} \braket{03} \\ 
    w_1 &= \sum_{\braket{012} \in T^2} \braket{012} + \sum_{\braket{0123} \in T^3} \big( \braket{012} + \braket{023} \big) \\
    w_0 &= \sum_{\braket{0123} \in T^3} \braket{0123}.
\end{split}
\end{equation*}
Note in general that $w_0$ will consist of the entire manifold.

Now, we'll spell out a formula for $(\tilde{a} \cup \underbrace{\delta \tilde{a} \cup \cdots \cup \delta \tilde{a}}_{q \text{ times}})(0 \cdots q)$. The first few examples are
\begin{equation}
\begin{split}
   (\tilde{a} \cup \delta \tilde{a})(01)
   &= \tilde{a}(0) \big(1 + \tilde{a}(1) \big) \\
   (\tilde{a} \cup \delta \tilde{a} \cup \delta \tilde{a})(012) 
   &= \tilde{a}(0)\tilde{a}(2)\big(1 + \tilde{a}(1)\big) \\
   (\tilde{a} \cup \delta \tilde{a} \cup \delta \tilde{a} \cup \delta \tilde{a})(0123) 
   &= \tilde{a}(0)\tilde{a}(2) \big(1 + \tilde{a}(1)\big) \big(1 + \tilde{a}(3)\big) \\
   (\tilde{a} \cup \delta \tilde{a} \cup \delta \tilde{a} \cup \delta \tilde{a} \cup \delta \tilde{a})(01234) 
   &= \tilde{a}(0)\tilde{a}(2) \tilde{a}(4) \big(1 + \tilde{a}(1)\big) \big(1 + \tilde{a}(3)\big)
\end{split}
\end{equation}
which use $(\tilde{a} \cup \underbrace{\delta \tilde{a} \cup \cdots \cup \delta \tilde{a}}_{q \text{ times}})(0 \cdots q) = \tilde{a}(0)\big(\tilde{a}(0)+\tilde{a}(1)\big)\big(\tilde{a}(1)+\tilde{a}(2)\big) \cdots \big(\tilde{a}(q-1) + \tilde{a}(q)\big)$ and that $\tilde{a}(i)\tilde{a}(i) = \tilde{a}(i)$. It's a simple inductive argument to show
\begin{equation} \label{eq:aCupDeltaA_repeat}
    (\tilde{a} \cup \underbrace{\delta \tilde{a} \cup \cdots \cup \delta \tilde{a}}_{q \text{ times}})(0 \cdots q) =
    \begin{cases}
       \left( \prod_{i=0}^k \tilde{a}(2 i) \right) \left( \prod_{i=0}^{k-1} (1+ \tilde{a}(2 i + 1 )) \right) \text{ if } q = 2k \text{ is even,} \\
       \left( \prod_{i=0}^k \tilde{a}(2 i) \right) \left( \prod_{i=0}^{k} (1+ \tilde{a}(2 i + 1 )) \right) \text{ if } q = 2k+1 \text{ is odd}
    \end{cases}
\end{equation}

\subsection{Proof of Proposition \ref{lem:generalizedLevinGuLemma}}
Now we prove the main proposition of this Appendix. It is instructive to first verify more explicitly the cases $d=1$ and $d=2$ before proceeding to the general case. Throughout this subsection, we will use the following shorthand
\begin{equation}
    \overline{[i_1, i_2, \cdots, i_k]} := \tilde{a}(i_1)\tilde{a}(i_2) \cdots \tilde{a}(i_k)
\end{equation}
to refer to products of $\tilde{a}$ matrix elements.

In $d=1$, we'll have 
\begin{equation*}
    w_1 = \sum_{\braket{0} \in T^0} \braket{0} + \sum_{\braket{01} \in T^1} \braket{0}
\end{equation*}
which tells us 
\begin{equation}
\begin{split}
    \int_{w_0} \tilde{a} \cup \delta \tilde{a} + \int_{w_1} \tilde{a} &= \left( \sum_{\braket{01} \in T^1} \big(  \overline{[0,1]} + \overline{[0]} \big) \right) + 
    \left(\sum_{\braket{0} \in T^0}\overline{[0]} + \sum_{\braket{01} \in T^1} \overline{[0]} \right)\\ 
    &= \sum_{\braket{0} \in T^0} \overline{[0]} + \sum_{\braket{01} \in T^1} \overline{[0,1]}.
\end{split}
\end{equation}
Here, each big set of parentheses on the first line is an integral over $w_1,w_0$ respectively, together with the evaluation of $\tilde{a} \cup \delta \tilde{a} \cup \cdots$ on them. The last equality is what we wanted to show. 

Now, let's check $d=2$. we'll have
\begin{equation*}
    w_2 = \sum_{\braket{0} \in T^0} \braket{0} + \sum_{\braket{01} \in T^1} \braket{0} + \sum_{\braket{012} \in T^2} \braket{0} \quad,\quad w_1 = \sum_{\braket{01} \in T^1} \braket{01} + \sum_{\braket{012} \in T^2} \braket{02}.
\end{equation*}
and thus
\begin{equation}
\begin{split}
    \int_{w_0} & \tilde{a} \cup \delta\tilde{a} \cup \delta\tilde{a} + \int_{w_1} \tilde{a} \cup \delta\tilde{a} + \int_{w_2} \tilde{a} \\
    &= \left( \sum_{\braket{012}} \big(\overline{[0,1,2]} + \overline{[0,2]}\big) \right) 
     + \left( \sum_{\braket{012}} \big(\overline{[0,2]} + \overline{[0]}\big) + \sum_{\braket{01}} \big( \overline{[0,1]} + \overline{[0]} \big)\right) 
     + \left( \sum_{\braket{012}} \overline{[0]} + \sum_{\braket{01}} \overline{[0]} + \sum_{\braket{0}} \overline{[0]} \right) 
\end{split}
\end{equation}
where each set of large parentheses comes from summing the contribution of $\tilde{a} \cup \delta\tilde{a} \cup \cdots \cup \delta\tilde{a}$ over the sum of simplices that occur in each $w_{0,1,2}$. From here, note that the sums over $\braket{0},\braket{01}$ are algebraically the same as what we saw in the $d=1$ case. As such, these sums will give the ``$Z \cdot CZ$'' part and we only need to check that the sums over $\braket{012}$ conspire to give the ``$CCZ$'' term like we expect. Indeed things do cancel out to give the final result
\begin{equation*}
    \int_{w_0} \tilde{a} \cup \delta\tilde{a} \cup \delta\tilde{a} + \int_{w_1} \tilde{a} \cup \delta\tilde{a} + \int_{w_2} \tilde{a} = \sum_{\braket{0} \in T^0} \overline{[0]} + \sum_{\braket{01} \in T^1} \overline{[0,1]} + \sum_{\braket{012} \in T^2} \overline{[0,1,2]}.
\end{equation*}

This lesson that many of the terms in the $d$-dimensional sum also occur in the $(d-1)$-dimensional one will indeed hold for all dimensions. Specifically, the algebraic expresssions for the sum $\sum_{\braket{0 \cdots k}}$ over $k$-simplices will always look the same in any dimension. As such, if we verify the proposition in dimension $(d-1)$ that the integral gives the ``$Z \times CZ \times \cdots \times C^{(d-1)}Z$'' circuit, then in $d$-dimensions we only have to verify that the sums over $d$-simplices in the relevant integrals give the ``$C^{(d)}Z$'' terms. The rest of this subsection will go towards evaluating this sum over $d$-simplices.

Now we proceed to give the argument in general dimensions. In terms of the $\partial_p(t)$ defined earlier, the precise equation we want to show is
\begin{equation}
    \int_{\partial_{d}(\braket{0 \cdots d})} \tilde{a} \cup \underbrace{\delta\tilde{a} \cup \cdots \cup \delta\tilde{a}}_{d \text{ times}} + \cdots + \int_{\partial_1(\braket{0 \cdots d})} \tilde{a} \cup \delta \tilde{a} + \int_{\partial_0(\braket{0 \cdots d})} \tilde{a}
    = \overline{[0,\cdots,d]}
\end{equation}
and summing over $\braket{0 \cdots d}$ would give the ``$C^{(d)}Z = \sum_{\braket{0 \cdots d}} \overline{[0,\cdots,d]}$'' term. The rest of the $C^{(k)}Z$ terms for $k < d$ come from the analogous terms with respect to $\{\partial_p(\braket{0 \cdots k})\}_{p=1}^k$.

We can do this by an inductive argument by noting some nice recursive properties of the formulas presented previously for the $\partial_p(\braket{0 \cdots d})$ and $\tilde{a} \cup \underbrace{\delta \tilde{a} \cup \cdots \cup \delta \tilde{a}}_{p \text{ times}}$. 

First, one can note that for $p \ge 2$
\begin{equation} \label{eq:recursionRegularSimplices}
    \partial_p(0 \cdots d) = \sum_{\braket{2 \tilde{i}_1 \cdots \tilde{i}_{p-2}} \in \partial_{p-2}(\braket{2 \cdots d})} \braket{0 1  2 \tilde{i}_1 \cdots \tilde{i}_{p-2}} + \sum_{\braket{0 j_1 \cdots j_{p}} \in \partial_{p}(\braket{0 2 \cdots d})} \braket{0 j_1 \cdots j_{p}}
\end{equation}
where $\braket{0 \tilde{i}_1 \cdots \tilde{i}_{p-2}} \in \partial_{p-2}(2 \cdots d)$ refers to a simplex that's nonzero in the chain $\partial_{p-2}(2 \cdots d)$, etc. This can be derived from Eq.~\eqref{eq:regularSimpFullExpr} from the cases of setting $i_1 = 1$ vs. $i_1 = 0$. In other words, there is a Fibonacci-like recursion relation for $\partial_p(0 \cdots d)$. The terms not containing the vertex $1$ can be mapped one-to-one with terms of $\partial_{p-2}(\braket{2 \cdots d})$. The terms containing $1$ can be mapped one-to-one with $\partial_{p}(\braket{0 2 \cdots d})$. For $p < 2$, we'll have
\begin{equation} \label{eq:recursionRegSimp_LastTwo}
\begin{split}
    \partial_0(\braket{0 \cdots d}) &= \sum_{\braket{0} \in \partial_0(\braket{0 2 \cdots d)}} \braket{0} = \braket{0}  \\
    \partial_1(\braket{0 \cdots d}) &= \sum_{\braket{0 i_1} \in \partial_1(\braket{0 2 \cdots d)}} \braket{0 \, i_1} = \braket{0 d}.
\end{split}
\end{equation}
which correspond to only the $\partial_p(\braket{2 \cdots d})$, not the $\partial_{p-2}(\braket{0 2 \cdots d})$ terms.

Next, we can note that for $p \ge 2$,
\begin{equation} \label{eq:recursionAcupDeltaArepeated}
    (\tilde{a} \cup \underbrace{\delta \tilde{a} \cup \cdots \cup \delta \tilde{a}}_{p \text{ times}})(0 \cdots d) = \big(\tilde{a}(0) + \tilde{a}(0)\tilde{a}(1)\big) (\tilde{a} \cup \underbrace{\delta \tilde{a} \cup \cdots \cup \delta \tilde{a}}_{p-2 \text{ times}})(2 \cdots d)
\end{equation}
which follows immediately from Eq.~\eqref{eq:aCupDeltaA_repeat}.

Now, the inductive proof for $d > 2$ starts as
\begin{equation}
\begin{split}
    &\int_{\partial_{d} (\braket{0 \cdots d})} \tilde{a} \cup \underbrace{\delta\tilde{a} \cup \cdots \cup \delta\tilde{a}}_{d \text{ times}} + \cdots + \int_{\partial_1(\braket{0 \cdots d})} \tilde{a} \cup \delta \tilde{a} + \int_{\partial_0(\braket{0 \cdots d})} \tilde{a} \\
    &= \left( \sum_{\braket{0 j_1 \cdots j_{d}} \in \partial_{d}(\braket{0 2 \cdots d})} (\tilde{a} \cup \underbrace{\delta \tilde{a} \cup \cdots \cup \delta \tilde{a}}_{d \text{ times}})(\braket{0 j_1 \cdots j_{d}}) + \sum_{\braket{2 \tilde{i}_1 \cdots \tilde{i}_{d-2}} \in \partial_{d-2}(\braket{2 \cdots d})} (\tilde{a} \cup \underbrace{\delta \tilde{a} \cup \cdots \cup \delta \tilde{a}}_{d \text{ times}})(\braket{0 1  2 \tilde{i}_1 \cdots \tilde{i}_{d-2}}) \right) \\
    &+ \cdots \\
    &+ \left( \sum_{\braket{0 j_1 \cdots j_{p}} \in \partial_{p}(\braket{0 2 \cdots d})} (\tilde{a} \cup \underbrace{\delta \tilde{a} \cup \cdots \cup \delta \tilde{a}}_{p \text{ times}})(\braket{0 j_1 \cdots j_{p}}) + \sum_{\braket{2 \tilde{i}_1 \cdots \tilde{i}_{p-2}} \in \partial_{p-2}(\braket{2 \cdots d})} (\tilde{a} \cup \underbrace{\delta \tilde{a} \cup \cdots \cup \delta \tilde{a}}_{p \text{ times}})(\braket{0 1  2 \tilde{i}_1 \cdots \tilde{i}_{p-2}}) \right) \\
    &+ \cdots \\
    &+ \left( \sum_{\braket{0 j_1 j_2} \in \partial_{2}(\braket{0 2 \cdots d})} (\tilde{a} \cup \delta \tilde{a} \cup \delta \tilde{a})(\braket{0 j_1 j_2}) + (\tilde{a} \cup \delta \tilde{a} \cup \delta \tilde{a})(\braket{0 1 2}) \right) \\
    &+ \sum_{\braket{0 i_1} \in \partial_1(\braket{0 2 \cdots d)}} (\tilde{a} \cup \delta \tilde{a}) (0 i_1) \\
    &+ \sum_{\braket{0} \in \partial_0(\braket{0 2 \cdots d)}} \tilde{a}(0)
\end{split}
\end{equation}
where each line is the sum for a value of $p$: the first $(d-2)$ lines use Eq.~\eqref{eq:recursionRegularSimplices} and the last two lines use Eq.~\eqref{eq:recursionRegSimp_LastTwo}. Note that $\partial_d(\braket{0 2 \cdots d})$ is empty since $d > d-1$, and we'll drop that term in the next line.

Now we can rearrange terms and use Eq.~\eqref{eq:recursionAcupDeltaArepeated} to give
\begin{equation*}
\begin{split}
    = \Bigg( &\sum_{\braket{0 j_1 \cdots j_{d-1}} \in \partial_{d-1}(\braket{0 2 \cdots d})} (\tilde{a} \cup \underbrace{\delta \tilde{a} \cup \cdots \cup \delta \tilde{a}}_{d-1 \text{ times}})(\braket{0 j_1 \cdots j_{d-1}}) + \cdots \\
    + &\sum_{\braket{0 j_1 \cdots j_{p}} \in \partial_{p}(\braket{0 2 \cdots d})} (\tilde{a} \cup \underbrace{\delta \tilde{a} \cup \cdots \cup \delta \tilde{a}}_{p \text{ times}})(\braket{0 j_1 \cdots j_{p}}) + \cdots + \sum_{\braket{0 i_1} \in \partial_2(\braket{0 2 \cdots d)}} (\tilde{a} \cup \delta \tilde{a}) (0 \, i_1) + \sum_{\braket{0} \in \partial_0(\braket{0 2 \cdots d)}} \tilde{a}(0) \Bigg) \\
    + \big(\tilde{a}(0) &+ \tilde{a}(0)\tilde{a}(1)\big) \times \\
    \Bigg( &\sum_{\braket{2 \tilde{i}_1 \cdots \tilde{i}_{d-2}} \in \partial_{d-2}(\braket{2 \cdots d})} (\tilde{a} \cup \underbrace{\delta \tilde{a} \cup \cdots \cup \delta \tilde{a}}_{d-2 \text{ times}})(\braket{2 \tilde{i}_1 \cdots \tilde{i}_{d-2}}) + \cdots \\
    + &\sum_{\braket{2 \tilde{i}_1 \cdots \tilde{i}_{p-2}} \in \partial_{p-2}(\braket{2 \cdots d})} (\tilde{a} \cup \underbrace{\delta \tilde{a} \cup \cdots \cup \delta \tilde{a}}_{p-2 \text{ times}})(\braket{2 \tilde{i}_1 \cdots \tilde{i}_{p-2}}) + \cdots + \sum_{\braket{2 i_1} \in \partial_2(\braket{2 \cdots d)}} (\tilde{a} \cup \delta \tilde{a}) (2 \, i_1) + \sum_{\braket{2} \in \partial_0(\braket{2 \cdots d)}} \tilde{a}(2)
    \Bigg).
\end{split}
\end{equation*}
Now, we note that the big sums in parentheses are related to the same sum in $(d-1)$ and $(d-2)$, that
\begin{equation*}
\begin{split}
    &= \left( \int_{\partial_{d-1}(\braket{0 2 \cdots d})} \tilde{a} \cup \underbrace{\delta\tilde{a} \cup \cdots \cup \delta\tilde{a}}_{d-1 \text{ times}} + \cdots + \int_{\partial_1(\braket{0 2 \cdots d})} \tilde{a} \cup \delta \tilde{a} + \int_{\partial_0(\braket{0 2 \cdots d})} \tilde{a} \right)\\
    &+ \big(\tilde{a}(0) + \tilde{a}(0)\tilde{a}(1)\big) 
    \left( \int_{\partial_{d-2}(\braket{2 \cdots d})} \tilde{a} \cup \underbrace{\delta\tilde{a} \cup \cdots \cup \delta\tilde{a}}_{d-2 \text{ times}} + \cdots + \int_{\partial_1(\braket{2 \cdots d})} \tilde{a} \cup \delta \tilde{a} + \int_{\partial_0(\braket{2 \cdots d})} \tilde{a} \right).
\end{split}
\end{equation*}
So, this sum that we want to show equals ``$C^{(d)}Z$'' term on $\braket{0 \cdots d}$ is related inductively to ``$C^{(d-1)}Z$'',``$C^{(d-2)}Z$'' terms acting on $\braket{0 2 \cdots d}$ and $\braket{2 \cdots d}$. Thus by induction, we'd have that the above equals:
\begin{equation*}
    = \big(\tilde{a}(0)\tilde{a}(2) \cdots \tilde{a}(d)\big) + \big(\tilde{a}(0) + \tilde{a}(0)\tilde{a}(1)\big)\big(\tilde{a}(2) \cdots \tilde{a}(d)\big) = \tilde{a}(0)\tilde{a}(1)\tilde{a}(2) \cdots \tilde{a}(d) = \overline{[0, \cdots, d]}.
\end{equation*}

\section{Defining $\cup_m$ products via thickening and shifting} \label{app:defininingCupMProducts}

\subsection{Vector field frame for the higher cup products}
Now, we apply the thickening and shifting procedure of~\cite{T20} to define formulas for the cup product and higher cup product. First, we'll define a vector field frame in $\msquare_d$ as follows. We'll have $d$ vectors $\Vec{b}_i$, $i = 1,\dots,d$ as:

\begin{equation}
b_{ij} = (\Vec{b}_i)_j = 1/j^{i-1}
\end{equation}

So, vectors $\Vec{b}_i$ taken together will form a Vandermonde matrix

\begin{equation}
b_{ij} = 
\begin{pmatrix}
1      & 1                 & \cdots & 1                 \\
1      & \frac{1}{2}       & \cdots & \frac{1}{d}       \\
\vdots & \vdots            & \ddots & \vdots            \\
1      & \frac{1}{2^{d-2}} & \cdots & \frac{1}{d^{d-2}} \\
1      & \frac{1}{2^{d-1}} & \cdots & \frac{1}{d^{d-1}} \\
\end{pmatrix}
\end{equation}

The specific form of these vectors will be useful for us a bit later when we explicitly thicken and measure the intersection products. 

First, it'll be helpful to go through everything for the $\cup_0$ product then later see how everything generalizes. 

\subsection{Intersections for the $\cup := \cup_0$ product}
For a dual $(d-p)$-cell $P^\vee_{(z_1 \cdots z_d)}$, we can define its shifted counterpart $\tilde{P}^\vee_{(z_1 \cdots z_d)}(\epsilon)$ similarly to the parameterization in Eq.~\eqref{dualCellParam}:

\begin{equation}
\tilde{P}^\vee_{(z_1 \cdots z_d)}(\epsilon) = \msquare_d \cap \{\epsilon \Vec{b}_1 + z_{\ihat_1} t_{\ihat_1} \mathbf{e}_{\ihat_1}  + \cdots +  z_{\ihat_{d-p}} t_{\ihat_{d-p}} \mathbf{e}_{\ihat_{d-p}} | t_{\ihat} \ge 0 \}
\end{equation}

Now, we can analyze the intersections between these shifted dual cells and the original dual cells. In particular, given dual a $(d-p)$-cell $P^\vee_{(z'_1 \cdots z'_d)}$ and a shifted dual $(d-q)$-cell $\tilde{P}^\vee_{(z''_1 \cdots z''_d)}(\epsilon)$, we'll have that as $\epsilon \to 0$, the limit of the intersections will be empty, or will be some other cell $P^\vee_{(z_1 \cdots z_d)}$, that

\begin{equation}
\lim_{\epsilon \to 0^+} P^\vee_{(z'_1 \cdots z'_d)} \cap \tilde{P}^\vee_{(z''_1 \cdots z''_d)}(\epsilon) = P^\vee_{(z_1 \cdots z_d)}, \text{ or will be empty}
\end{equation}

We want to find the pairs $\{(z'_1 \cdots z_d),(z''_1 \cdots z''_d)\}$ so that $\lim_{\epsilon \to 0^+} P^\vee_{(z'_1 \cdots z'_d)} \cap \tilde{P}^\vee_{(z''_1 \cdots z''_d)}(\epsilon)$ is \textit{nonempty} and \textit{stays at the highest possible dimension}, i.e. that $P^\vee_{(z_1 \cdots z_d)}$ is a $(d-p-q)$-dimensional cell. In general, there will be examples where the limit of the intersections becomes lower dimensional. This is summarized in the following proposition:

\begin{prop} \label{cup0Prop}
Let $P^\vee_{(z'_1 \cdots z'_d)}$ be a dual $(d-p)$-cell and $\tilde{P}^\vee_{(z''_1 \cdots z''_d)}(\epsilon)$ be a shifted dual $(d-q)$-cell. 
Let $\{i_1 < \cdots < i_p\} \subset \{1 \cdots d\}$ be the indices with ${z'_{i_{\dots}}} = \bullet$ and $\{\ihat_1 < \cdots < \ihat_{d-p}\} = \{1 \cdots d\} \backslash \{i_1 \cdots i_p\}$. 
Similarly, let $\{j_1 < \cdots < j_q\} \subset \{1 \cdots d\}$ be the indices with ${z''_{j_{\dots}}} = \bullet$ and let $\{\jhat_1 < \cdots < \jhat_{d-q}\} = \{1 \cdots d\} \backslash \{j_1 \cdots j_q\}$. 

\begin{enumerate}
\item If $\{i_1 \cdots i_p\} \cap \{j_1 \cdots j_q\}$ is nonempty, then $P^\vee_{(z'_1 \cdots z_d)} \cap \tilde{P}^\vee_{(z''_1 \cdots z''_d)}(\epsilon)$ is empty for $\epsilon > 0$.
\item Suppose $\{i_1 \cdots i_p\} \cap \{j_1 \cdots j_q\}$ is empty. Then if for some $k$, $z'_k = +$ and $z''_k = -$ or vice-versa, $\lim_{\epsilon \to 0^+} P^\vee_{(z'_1 \cdots z'_d)} \cap \tilde{P}^\vee_{(z''_1 \cdots z''_d)}(\epsilon)$ has lower dimension than $(d-p-q)$.
\item If for some $k$, $z'_k = \bullet$ and $z''_k = +$, or $z'_k = -$ and $z''_k = \bullet$, then $P^\vee_{(z'_1 \cdots z'_d)} \cap \tilde{P}^\vee_{(z''_1 \cdots z''_d)}(\epsilon)$ is empty for $\epsilon > 0$.
\end{enumerate}

These tells us that if $(z'_1 \cdots z'_d)$,$(z''_1 \cdots z''_d)$ are to contribute to the intersection, then first they can't share any indices both valued as $\bullet$. Second, any indices that are both not $\bullet$ should be the same, $+$ or $-$. The final statement will suppose that they satisfy these two properties. 

\begin{enumerate}
    \setcounter{enumi}{3}
    \item If for each $k$ with $z'_k = \bullet$ we have $z''_k = -$ and for each $k$ with $z''_k = \bullet$ that $z'_k = +$, then $\lim_{\epsilon \to 0} P^\vee_{(z'_1 \cdots z'_d)} \cap \tilde{P}^\vee_{(z''_1 \cdots z''_d)}(\epsilon) = P^\vee_{(z_1 \cdots z_d)}$, where $(z_1 \cdots z_d)$ represents a dual $(d-p-q)$-dimensional cell. In particular, if $z'_k = z''_k \in \{+,-\}$, then $z_k = z'_k = z''_k$. And if $z'_k = \bullet$ or $z''_k = \bullet$, then $z_k = \bullet$.
\end{enumerate}
where `limit' here means the Cauchy limit of the sets. 
\end{prop}

\begin{proof}
Let's start with the first statement. Note that if $\{i_1 \cdots i_p\}$ and $\{j_1 \cdots j_q\}$ share at least one element, then $\{\ihat_1 \cdots \ihat_{d-p}\}$ and $\{\jhat_1 \cdots \jhat_{d-q}\}$ must share at least $(d-p-q+1)$ elements. Note that the hyperplanes that $P^\vee_{(z'_1 \cdots z'_d)}$, $P^\vee_{(z''_1 \cdots z''_d)}$ are spanned by the vectors $\{\mathbf{e}_{\ihat_1} \cdots \mathbf{e}_{\ihat_{d-p}}\}$ and $\{\mathbf{e}_{\jhat_1} \cdots \mathbf{e}_{\jhat_{d-q}}\}$. So, by an argument of genericity of intersections similar to Part 1 of Proposition 2 of~\cite{T20}, the fact that the hyperplanes share at least $(d-p-q+1)$ vectors means that for any $\epsilon > 0$, $P^\vee_{(z'_1 \cdots z'_d)} \cap \tilde{P}^\vee_{(z''_1 \cdots z''_d)}(\epsilon)$ is empty.

Now, let's go to the second statement. First note that $\lim_{\epsilon \to 0^+} P^\vee_{(z'_1 \cdots z'_d)} \cap \tilde{P}^\vee_{(z''_1 \cdots z''_d)}(\epsilon)$ is automatically a subset of $P^\vee_{(z'_1 \cdots z'_d)} \cap P^\vee_{(z''_1 \cdots z''_d)}$. So, if $P^\vee_{(z'_1 \cdots z'_d)} \cap P^\vee_{(z''_1 \cdots z''_d)}$ has lower dimension than $(d-p-q)$, then we've proved Part 2. First, note that $\{i_1 \cdots i_p\} \cap \{j_1 \cdots j_q\}$ being empty means that $\{\ihat_1 \cdots \ihat_{d-p}\}$ $\{\jhat_1 \cdots \jhat_{d-q}\}$ share exactly $(d-p-q)$ coordinates. This means that the hyperplanes spanned by $P^\vee_{(z'_1 \cdots z'_d)}$ and $P^\vee_{(z''_1 \cdots z''_d)}$ will intersect in a $(d-p-q)$-dimensional hyperplane spanned by $\{\mathbf{e}_\ell | \ell \in \{\ihat\} \cap \{\jhat\}\}$. Note that if $z'_k = +$ and $z''_k = -$ (or vice versa), then $k \in \{\ihat_1 \cdots \ihat_{d-p}\} \cap \{\jhat_1 \cdots \jhat_{d-q}\}$. But, the definitions of the cells means that any point in $P^\vee_{(z'_1 \cdots z'_d)} \cap P^\vee_{(z''_1 \cdots z''_d)}$ will have its $k$th coordinate equalling zero, since one cell requires the $k$th coordinate to be at least zero while the other one requires that it's at most zero. So, this means that $P^\vee_{(z'_1 \cdots z'_d)} \cap P^\vee_{(z''_1 \cdots z''_d)}$ can only be at most a $(d-p-q-1)$-dimensional hyperplane, spanned by the $\mathbf{e}_\ell$ with $\ell \neq k$ in $\{\ihat\} \cap \{\jhat\}$

To see the third statement, note that the cell is being shifted in the direction $\Vec{b}_1 = (1,\cdots,1)$. Say now that $\epsilon > 0$. If $z'_k = \bullet$ and $z''_k = +$, then any point $(x_1, \cdots, x_d) \in P^\vee_{(z'_1 \cdots z'_d)} \cap \tilde{P}^\vee_{(z''_1 \cdots z''_d)}(\epsilon)$ will have to satisfy $x_k = 0$ since $z'_k = \bullet$, and also satisfy $x_k \ge \epsilon$ since $z''_k = +$; this is not possible. Similarly if $z'_k = -$ and $z''_k = \bullet$, then we'd need that $x_k \le 0$ since $z'_k = -$, and we'd need $x_k = \epsilon$ since $z''_k = \bullet$, which is impossible.

The final statement can be seen by directly computing the intersection for any $\epsilon > 0$ and taking the limit.
\end{proof}

\subsection{Intersections for the $\cup_k$ products}
Following~\cite{T20}, we will define our shifted and thickened version of $P^\vee_{(z_1 \cdots z_d)}$ as:

\begin{equation}
\begin{split}
\tilde{P}^{\vee, \text{thick}}_{(z_1 \cdots z_d)}(\epsilon_{m+1}) = \msquare_d \cap \{ & \epsilon_1 \Vec{b}_1 + \dots + \epsilon_{m+1} \Vec{b}_{m+1} + z_{\ihat_1} t_{\ihat_1} \mathbf{e}_{\ihat_1}  + \cdots +  z_{\ihat_{d-p}} t_{\ihat_{d-p}} \mathbf{e}_{\ihat_{d-p}} | \\
& t_{\ihat} \ge 0, -\varepsilon \le \epsilon_i \le \varepsilon \text{ for } i \in \{1 \cdots m\}\}
\end{split}
\end{equation}

Here, we thicken along the first $m$ vectors $\Vec{b}_1,\dots,\Vec{b}_{m}$ in \textit{both} positive and negative direction up to parameter $\varepsilon$, and shift along the final $\Vec{b}_{m+1}$. Our goal will be to analyze what happens to the intersections in the following order of limits:

\begin{equation*}
\lim_{\varepsilon \to 0} \lim_{\epsilon_{m+1} \to 0^+} P^\vee_{(z'_1 \cdots z'_d)} \cap \tilde{P}^{\vee, \text{thick}}_{(z''_1 \cdots z''_d)}(\epsilon_{m+1})
\end{equation*}

Since we're thickening along $m$ directions, this thickened intersection between a $(d-p)$-cell and a $(d-q)$-cell will have its full dimension if this limiting set is a $(d-p-q+m)$-dimensional one, i.e. if it's a $(p+q-m)$-dimensional form. To think about this intersection question, we have the following proposition.

\begin{prop} \label{cup_mProp}
Let $P^\vee_{(z'_1 \cdots z'_d)}$ be a dual $(d-p)$-cell and $\tilde{P}^{\vee, \text{thick}}_{(z''_1 \cdots z''_d)}(\epsilon_{m+1})$ be a thickening and shifting of a dual $(d-q)$-cell.
Let $\{i_1 < \cdots < i_p\} \subset \{1 \cdots d\}$ be the indices with ${z'_{i_{\dots}}} = \bullet$. And let $\{j_1 < \cdots < j_q\} \subset \{1 \cdots d\}$ be the indices with ${z''_{j_{\dots}}} = \bullet$.
And, let $\{\lambda_1 < \cdots < \lambda_r\} := \{i_1 \cdots i_p\} \cap \{j_1 \cdots j_q\}$, and denote by $\{a_1 < \cdots < a_{p+q-r}\} := \{i_1 \cdots i_p\} \cup \{j_1 \cdots j_q\}$.
And denote by $\{\hat{\lambda}_1 < \cdots < \hat{\lambda}_{p+q-2r}\} := \{a_1 \cdots a_{p+q-r}\} \backslash \{\lambda_1 \cdots \lambda_r\}$.

\begin{enumerate}
\item If $r \neq m$, then $\lim_{\varepsilon \to 0} \lim_{\epsilon_{m+1 \to 0}} P^\vee_{(z'_1 \cdots z'_d)} \cap \tilde{P}^{\vee, \text{thick}}_{(z''_1 \cdots z''_d)}(\epsilon_{m+1})$ is empty or limits to lower-dimension than $(d-p-q+m)$.
\item Suppose $r=m$. Then if for some $k$, $z'_k = +$ and $z''_k = -$ or vice-versa, then $\lim_{\varepsilon \to 0} \lim_{\epsilon_{m+1} \to 0^+} P^\vee_{(z'_1 \cdots z'_d)} \cap \tilde{P}^{\vee, \text{thick}}_{(z''_1 \cdots z''_d)}(\epsilon)$ has lower dimension than $(d-p-q+m)$.
\item Now, suppose both $r=m$ and that $\{z'_k,z''_k\} \neq \{+,-\}$ for all $k$. Note that for each ${\hat{\lambda}_{\dots}}$, we'll have that either ${z'_{\hat{\lambda}_{\dots}}} = \bullet$ and ${z''_{\hat{\lambda}_{\dots}}} \in \{+,-\}$ or that ${z''_{\hat{\lambda}_{\dots}}} = \bullet$ and ${z'_{\hat{\lambda}_{\dots}}} \in \{+,-\}$. And for each ${\lambda_{\dots}}$, we'll have that ${z'_{\lambda_{\dots}}} = {z''_{\lambda_{\dots}}} = \bullet$. For every ${\hat{\lambda}_\ell}$, define $k(\ell)$ so that $\lambda_{k(\ell) - 1} < {\hat{\lambda}_\ell} < \lambda_{k(\ell)}$, taking $\lambda_0 = 0$ and $\lambda_{r+1} = \infty$.
    \begin{enumerate}
        \item Suppose for each ${\hat{\lambda}_\ell}$, that for each $\ell$ with $z'_{\hat{\lambda}_\ell} = \bullet$ we have $z''_k = - (-)^{k(\ell)}$ and for each $k$ with $z''_k = \bullet$ that $z'_k = + (-)^{k(\ell)}$. Then $\lim_{\epsilon \to 0} P^\vee_{(z'_1 \cdots z'_d)} \cap \tilde{P}^{\vee, \text{thick}}_{(z''_1 \cdots z''_d)}(\epsilon) = P^\vee_{(z_1 \cdots z_d)}$, where $(z_1 \cdots z_d)$ represents a dual $(d-p-q+m)$-dimensional cell. In particular, if $z'_k = z''_k \in \{+,-\}$, then $z_k = z'_k = z''_k$. And if $z_k = \bullet$ or $z'_k = \bullet$, then $z''_k = \bullet$.
        \item Otherwise, $\lim_{\varepsilon \to 0} \lim_{\epsilon_{m+1} \to 0^+} P^\vee_{(z_1 \cdots z_d)} \cap \tilde{P}^{\vee, \text{thick}}_{(z'_1 \cdots z'_d)}(\epsilon_{m+1})$ will be empty.
    \end{enumerate}
\end{enumerate}
\end{prop}

Now, let's prove our proposition on $\cup_m$. Part 1 follows from analogous reasoning to Proposition 2 of~\cite{T20}, so we won't describe it in detail. And Part 2 is essentially the same reason as Part 2 of Prop~\ref{cup0Prop}. So, let's now focus on Part 3. 

We'll only spell out the case of $(d-p)$- and $(d-q)$-dimensional dual cells with $p+q-m=d$. The more general case actually follows from this case. Recall our definition $(z_1 \cdots z_d)$: if $z'_k = z''_k \in \{+,-\}$, then $z_k = z'_k = z''_k$ and if $z_k = \bullet$ or $z'_k = \bullet$, then $z''_k = \bullet$. If we consider the intersection question restricted to the cell $P_{(z_1 \cdots z_d)}$, then we're seeing whether the limit of the intersection is empty or consists of the point at the center of $P_{(z_1 \cdots z_d)}$. If it's empty restricted to this cell, then it's empty everywhere in $\msquare_d$. And if it restricts to the point at the center of $P_{(z_1 \cdots z_d)}$, then the intersection will be the entire dual cell $P^\vee_{(z_1 \cdots z_d)}$. 

Defining $\{\ihat_1 \cdots \ihat_{d-p}\} = \{1 \cdots n\} \backslash \{i_1 \cdots i_p\}$ and $\{\jhat_1 \cdots \jhat_{d-q}\} = \{1 \cdots n\} \backslash \{j_1 \cdots j_q\}$, we have that
\begin{equation*}
P^\vee_{(z'_1 \cdots z'_d)} = \msquare_d \cap \{z'_{\ihat_1} s_{\ihat_1} \mathbf{e}_{\ihat_1}  + \cdots +  z'_{\ihat_{d-p}} s_{\ihat_{d-p}} \mathbf{e}_{\ihat_{d-p}} | s_{\ihat} \ge 0 \}
\end{equation*}
and
\begin{equation*}
\begin{split}
\tilde{P}^{\vee, \text{thick}}_{(z''_1 \cdots z''_d)}(\epsilon_{m+1}) = \msquare_d \cap \{ & \epsilon_1 \Vec{b}_1 + \dots + \epsilon_{m+1} \Vec{b}_{m+1} + z''_{\jhat_1} t_{\jhat_1} \mathbf{e}_{\jhat_1}  + \cdots +  z''_{\jhat_{d-q}} t_{\jhat_{d-q}} \mathbf{e}_{\jhat_{d-q}} | \\
& t_{\jhat} \le 0, -\varepsilon \le \epsilon_\ell \le \varepsilon \text{ for } \ell \in \{1 \cdots m\}\}
\end{split}
\end{equation*}

So, we want to find when we can solve for ${s_{\dots}},{t_{\dots}} > 0$ for $\epsilon_{m+1} > 0$

\begin{equation}
z'_{\ihat_1} s_{\ihat_1} \mathbf{e}_{\ihat_1}  + \cdots +  z'_{\ihat_{d-p}} s_{\ihat_{d-p}} \mathbf{e}_{\ihat_{d-p}} = \epsilon_1 \Vec{b}_1 + \dots + \epsilon_{m+1} \Vec{b}_{m+1} + z''_{\jhat_1} t_{\jhat_1} \mathbf{e}_{\jhat_1}  + \cdots +  z''_{\jhat_{d-q}} t_{\jhat_{d-q}} \mathbf{e}_{\jhat_{d-q}}
\end{equation}

Writing this out in components gives:

\begin{equation}
z'_{\ihat_1} s_{\ihat_1} \delta_{k,\ihat_1}  + \cdots +  z'_{\ihat_{d-p}} s_{\ihat_{d-p}} \delta_{k,\ihat_{d-p}} = \epsilon_1 b_{1,k} + \dots + \epsilon_{m} b_{m,k} + \epsilon_{m+1} b_{m+1,k} + z''_{\jhat_1} t_{\jhat_1} \delta_{k,\jhat_1}  + \cdots +  z''_{\jhat_{d-q}} t_{\jhat_{d-q}} \delta_{k,\jhat_{d-q}}
\end{equation}

Now, we'll define $S_{\ihat} = s_{\ihat}/\epsilon_{m+1}$ for all $\{\ihat\}$, $T_{\jhat} = t_{\jhat}/\epsilon_{m+1}$ for all $\{\jhat\}$, $A_{\ell} = \epsilon_{\ell}/\epsilon_{m+1}$ for $\ell = 1,\dots,m$. Then, we can divide the previous equation by $\epsilon_{m+1}$ to get:

\begin{equation}
z'_{\ihat_1} S_{\ihat_1} \delta_{k,\ihat_1}  + \cdots +  z'_{\ihat_{d-p}} S_{\ihat_{d-p}} \delta_{k,\ihat_{d-p}} = A_1 b_{1,k} + \dots + A_{m} b_{m,k} + b_{m+1,k} + z''_{\jhat_1} T_{\jhat_1} \delta_{k,\jhat_1}  + \cdots +  z''_{\jhat_{d-q}} T_{\jhat_{d-q}} \delta_{k,\jhat_{d-q}}
\end{equation}

Now, we will define

\begin{equation}
Y_k = 
\begin{cases}
z'_{k} S_{k}  &\text{ if } k \in \{\ihat\} \\
-z''_{k} T_{k} &\text{ if } k \in \{\jhat\} 
\end{cases}
\end{equation}

Noting that $\{\hat{\lambda}\} = \{\hat{\lambda}_1 \cdots \hat{\lambda}_{d-m}\} = \{\ihat\} \sqcup \{\jhat\}$, we can rewrite the equation as:

\begin{equation}
Y_k \delta_{k \in \{\hat{\lambda}\}} - (A_1 b_{1,k} + \dots + A_{m} b_{m,k} + b_{m+1,k}) = 0
\end{equation}

Note that we are abusing notation, since $Y_k$ is only defined if $k \in \{\hat{\lambda}\}$. These equations can be solved as 

\begin{equation} \label{mainSolZ}
\begin{split}
Y_{\hat{\lambda}} &= \frac{\det
\begin{pmatrix}
b_{1 \lambda_1}     & \cdots & b_{m+1,\lambda_1}     \\
\vdots             &        & \vdots                 \\
b_{1 \lambda_m}     & \cdots & b_{m+1,\lambda_m}     \\
b_{1 \hat{\lambda}} & \cdots & b_{m+1,\hat{\lambda}} \\
\end{pmatrix}
}{\det
\begin{pmatrix}
b_{1 \lambda_1} & \cdots & b_{m \lambda_1} \\
\vdots          &        & \vdots          \\
b_{1 \lambda_m} & \cdots & b_{m \lambda_m} \\
\end{pmatrix}
} \\ 
&\text{ where } \hat{\lambda} \in \{\hat{\lambda}_1, \dots, \hat{\lambda}_{n-m} \}
\end{split}
\end{equation}

\begin{equation}\label{mainSolA}
\begin{split}
A_{\ell} &= (-1)^{m-\ell}\frac{\det
\begin{pmatrix}
b_{\tilde{\imath}_1 \lambda_1} & \cdots & b_{\tilde{\imath}_m \lambda_1} \\
            \vdots             &        &        \vdots                  \\
b_{\tilde{\imath}_1 \lambda_m} & \cdots & b_{\tilde{\imath}_m \lambda_m} \\
\end{pmatrix}
}{\det
\begin{pmatrix}
b_{1 \lambda_1} & \cdots & b_{m \lambda_1} \\
    \vdots      &        &   \vdots          \\
b_{1 \lambda_m} & \cdots & b_{m \lambda_m} \\
\end{pmatrix}
}\\ 
&\text{ where } \ell \in \{1,\dots,m\} \text{ and } \{\tilde{\imath}_1,\dots,\tilde{\imath}_m\} = \{1,\dots,m+1\} \backslash \{ \ell \}
\end{split}
\end{equation}

This can be shown by plugging in these into the equations, multiplying both sides of the equation by the denominator matrix 
$\begin{pmatrix}
b_{1 \lambda_1} & \cdots & b_{m \lambda_1} \\
\vdots          &        & \vdots          \\
b_{1 \lambda_m} & \cdots & b_{m \lambda_m} \\
\end{pmatrix}$
and showing that the $(m+1) \times (m+1)$ determinant is the sum of the $m \times m$ determinants via the cofactor expansion of a matrix. This is illustrated in Appendix A of~\cite{T20} for a slightly more complicated equation than this.

Now, we analyze the consequences of these equations, using the fact that $b$ is a Vandermonde matrix. We'll have that $Y_{\hat{\lambda}}$ are all ratios of Vandermonde determinants. So,

\begin{equation}
\begin{split}
Y_{\hat{\lambda}} &= \frac{\det
\begin{pmatrix}
1 & \frac{1}{1+\lambda_1}     & \cdots & (\frac{1}{1+\lambda_1})^{m}     \\
1 & \frac{1}{1+\lambda_2}     & \cdots & (\frac{1}{1+\lambda_2})^{m}     \\
  &             \vdots        &        &  \vdots                         \\
1 & \frac{1}{1+\lambda_m}     & \cdots & (\frac{1}{1+\lambda_m})^{m}     \\
1 & \frac{1}{1+\hat{\lambda}} & \cdots & (\frac{1}{1+\hat{\lambda}})^{m} \\
\end{pmatrix}
}{\det
\begin{pmatrix}
1 & \frac{1}{1+\lambda_1} & \cdots & (\frac{1}{1+\lambda_1})^{m-1}  \\
1 & \frac{1}{1+\lambda_2} & \cdots & (\frac{1}{1+\lambda_2})^{m-1}  \\
  &         \vdots        &        &  \vdots                        \\
1 & \frac{1}{1+\lambda_m} & \cdots & (\frac{1}{1+\lambda_m})^{m-1}  \\
\end{pmatrix}
}  \\
&=  (\frac{1}{1+\hat{\lambda}}-\frac{1}{1+\lambda_1}) \cdots (\frac{1}{1+\hat{\lambda}}-\frac{1}{1+\lambda_m})
\end{split}
\end{equation}

This tells us that if $\lambda_{k(\ell)-1} < \hat{\lambda}_\ell < \lambda_{k(\ell)}$ (with $\lambda_0 = 0, \lambda_{m+1} = \infty$), then $Y_{\hat{\lambda}_\ell} > 0$ if $k(\ell)$ is odd and $Y_{\hat{\lambda}_\ell} < 0$ if $k(\ell)$ is even. Translating back to the original variables $S,T$, we'll have that $S_{\hat{\lambda}_\ell} = z'_{\hat{\lambda}_\ell} Y_{\hat{\lambda}_\ell}$ and $T_{\hat{\lambda}_\ell} = -z''_{\hat{\lambda}_\ell} Y_{\hat{\lambda}_\ell}$. This means that for $k(\ell)$ odd, $S_{\hat{\lambda}_\ell} > 0$ needs $z'_{\hat{\lambda}_\ell} = +$ and $T_{\hat{\lambda}_\ell} > 0$ needs $z''_{\hat{\lambda}_\ell} = -$. And for $k(\ell)$ odd, $S_{\hat{\lambda}_\ell} > 0$ needs $z'_{\hat{\lambda}_\ell} < 0$ and $T_{\hat{\lambda}_\ell} > 0$ needs $z''_{\hat{\lambda}_\ell} > 0$. 

These are exactly the conditions we listed in the proposition and that correspond to the diagrammatics described earlier. 

\section{Proofs of $\delta(\alpha \cup_m \beta) = \delta \alpha \cup_m \beta + \alpha \cup_m \delta \beta + \alpha \cup_{m-1} \beta + \beta \cup_{m-1} \alpha$} \label{app:proofOfHigherCupId}

Here, we give proofs of the recursive $\cup_m$ identities, following an analogous presentation in~\cite{morganHigherCupLecture} in the simplicial case.

\subsection{$\delta(\alpha \cup_0 \beta) = (\delta \alpha) \cup_0 \beta + \alpha \cup_0 (\delta \beta)$}
First, we can recall that $(\delta \alpha)(z_1 \cdots z_d)$ is gotten by summing over $\alpha(z'_1 \cdots z'_d)$  with $(z'_1 \cdots z'_d)$ obtained by replacing exactly one $\bullet$ from $(z_1 \cdots z_d)$ with either $+$ or $-$. As such, we can demonstrate the equation of Prop~\ref{prop:cup0_coboundary}:

\begin{equation*}
    \delta(\alpha \cup_0 \beta) = (\delta \alpha) \cup_0 \beta + \alpha \cup_0 (\delta \beta) \text{ on }C^*(\msquare_d)
\end{equation*}

\begin{proof}[Proof of Prop~\ref{prop:cup0_coboundary}]
See Figure~\ref{cup0BoundaryProof}. The notation in that figure is that if a square on either of the top or bottom lines of each term is replaced with a $+$ or $-$, then we sum over all the undetermined $\pm$ signs while keeping the $+$ or $-$ sign fixed for $\alpha$ or $\beta$. The first line representing $\delta(\alpha \cup_0 \beta)$ follows from the definition of the coboundary $\delta$. The second and third lines can similarly be argued for as follows. For example, $(\delta\alpha \cup_0 \beta)(\bullet, \cdots, \bullet)$ is of that form because every term  $\alpha(z'_1 \cdots z'_d)\beta(z''_1 \cdots z''_d)$ will have exactly one index $k$ for which both $z'_k \neq \bullet$ and $z''_k \neq \bullet$. In addition, all such terms would require $z''_k = -$, while $z'_k$ can be either choice of $+$ or $-$. And similar reasoning will allow us to argue for the form of $(\alpha \cup_0 \delta\beta)(\bullet, \cdots, \bullet)$.
\end{proof}

\begin{figure}[h!]
  \centering
  \includegraphics[width=\linewidth]{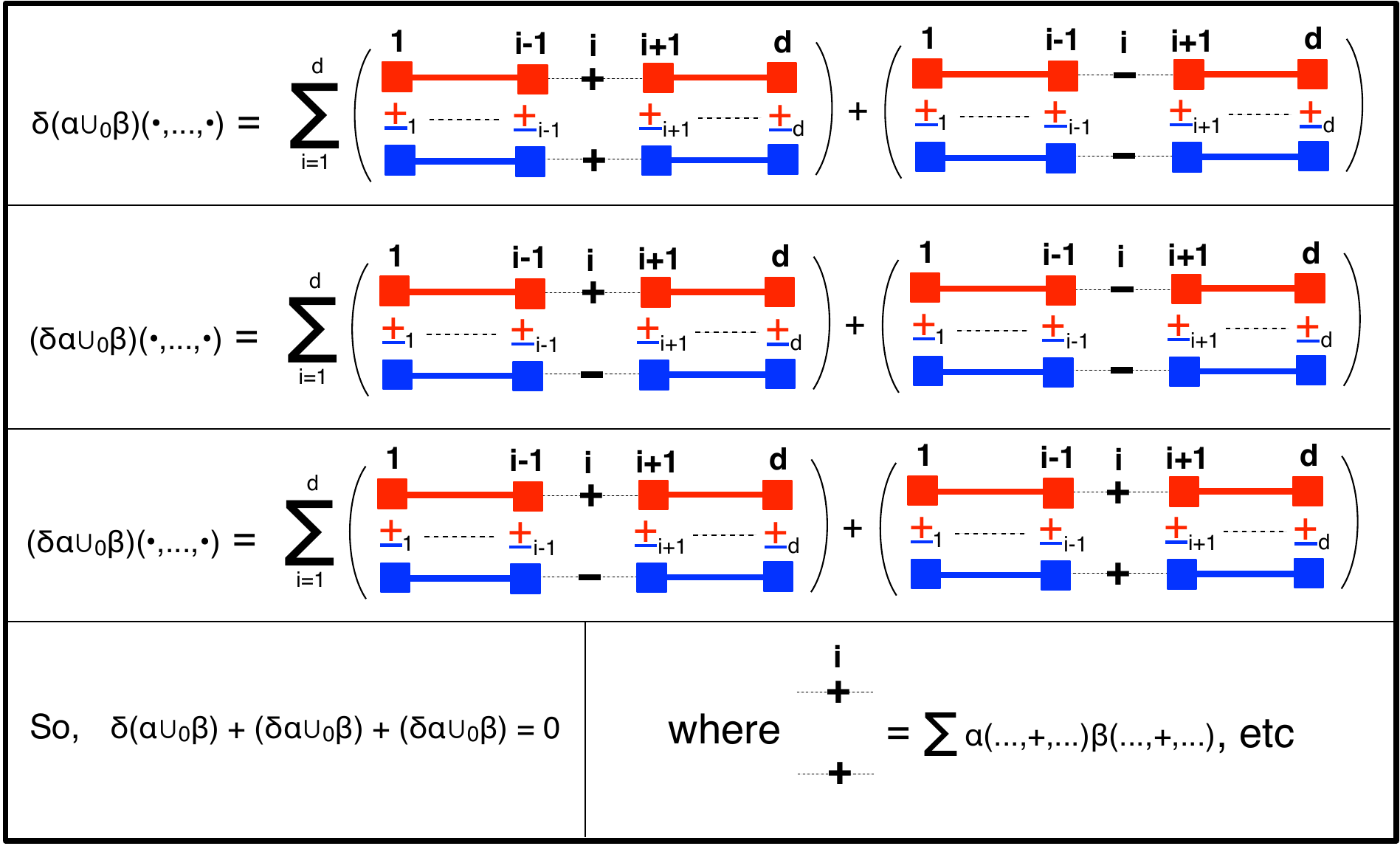}
  \caption{Proof of $\delta(\alpha \cup_0 \beta) = (\delta \alpha) \cup_0 \beta + \alpha \cup_0 (\delta \beta)$.}
  \label{cup0BoundaryProof}
\end{figure}

\begin{figure}[h!]
  \captionsetup{singlelinecheck=off}
  \centering
  \includegraphics[width=\linewidth]{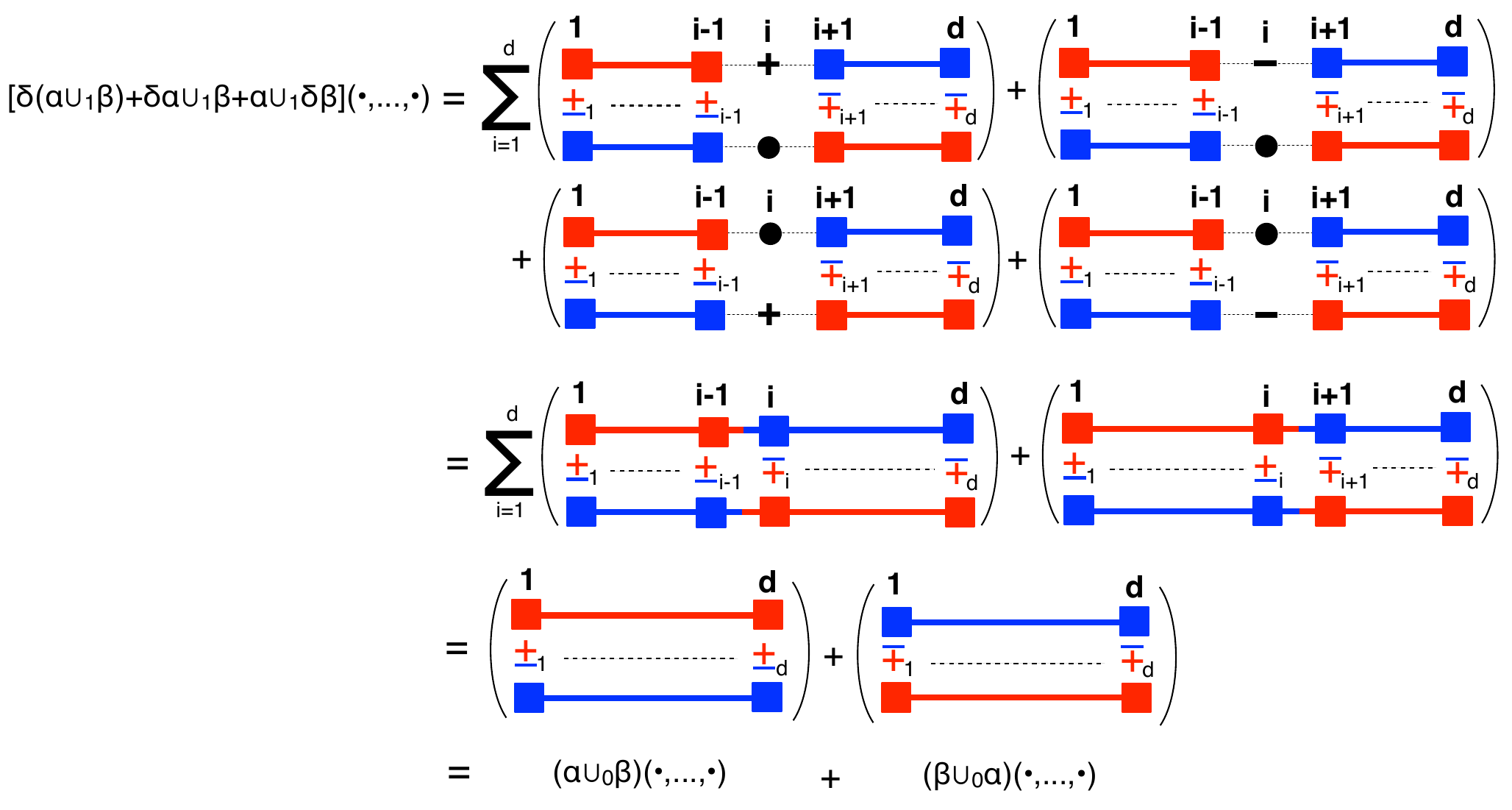}
  \caption[]{Proof of $\delta(\alpha \cup_1 \beta) + (\delta \alpha) \cup_1 \beta + \alpha \cup_1 (\delta \beta) = \alpha \cup_0 \beta + \beta \cup_0 \alpha$ (for $m$ even).
  \begin{itemize}
      \item The first equality is derived in a similar way to the proof of $\delta(\alpha \cup_0 \beta) = (\delta \alpha) \cup_0 \beta + \alpha \cup_0 (\delta \beta)$, and follows from the comments in the main text.
      \item The second equality is a rewriting of the first line, by grouping the terms diagonal to each other.
      \item The third equality follows from telescoping the sum. 
      \item The final equality is from the definition of $\cup_0$.
  \end{itemize}
  }
  \label{cup1BoundaryProof}
\end{figure}

\begin{figure}[h!]
  \captionsetup{singlelinecheck=off}
  \centering
  \includegraphics[width=\linewidth]{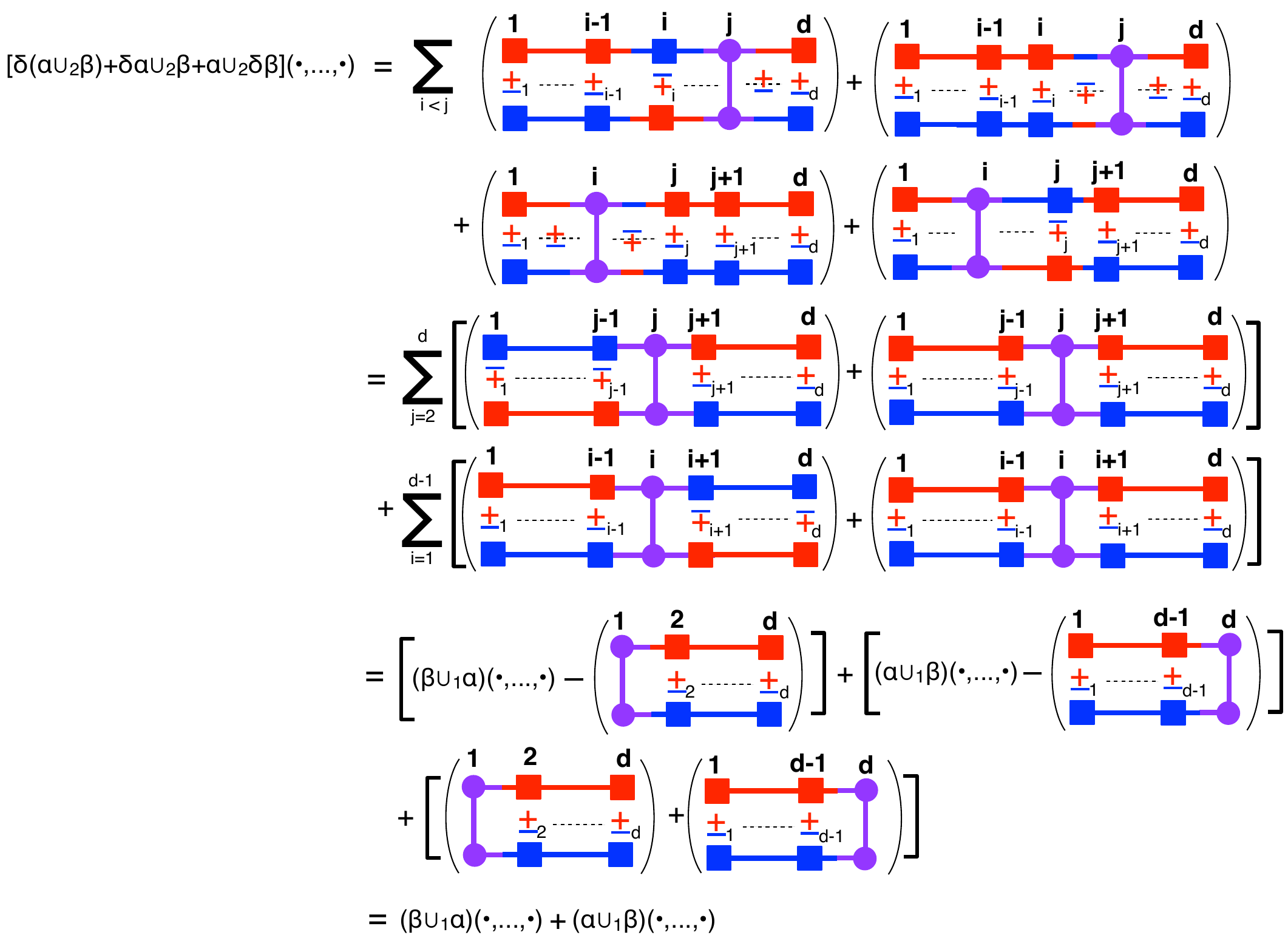}
  \caption[]{Proof of $\delta(\alpha \cup_2 \beta) + (\delta \alpha) \cup_2 \beta + \alpha \cup_2 (\delta \beta) = \alpha \cup_1 \beta + \beta \cup_1 \alpha$
  \begin{itemize}
      \item  The first equality is similar to the second equality of Figure~\ref{cup1BoundaryProof}. 
      \item   The second equality is gotten from telescoping terms in the first line. In particular, telescoping the top two lines (summing over $i<j$) of the second expression gives the first sum in the third expression. And telescoping the bottom two lines (summing over $j>i$) gives the second sum. We get $i: 1 \to d-1$ and $j:2 \to d$ because of the condition $i<j$. 
      \item  The third equality is as follows. The top terms of the fourth expression are rewriting the top-left plus bottom-right of the third expression. These differ from $(\alpha \cup_1 \beta)$ and $(\beta \cup_1 \alpha)$ because of the range of the summation. The bottom terms are from the cancellations between the top-right and bottom-left of the third expression. We are only left with the terms for $i=1$ and $j=d$.
      \item The fourth equality follows immediately.
  \end{itemize}
  }
  \label{cup2BoundaryProof}
\end{figure}

\begin{figure}[h!]
  \captionsetup{singlelinecheck=off}
  \centering
  \includegraphics[width=\linewidth]{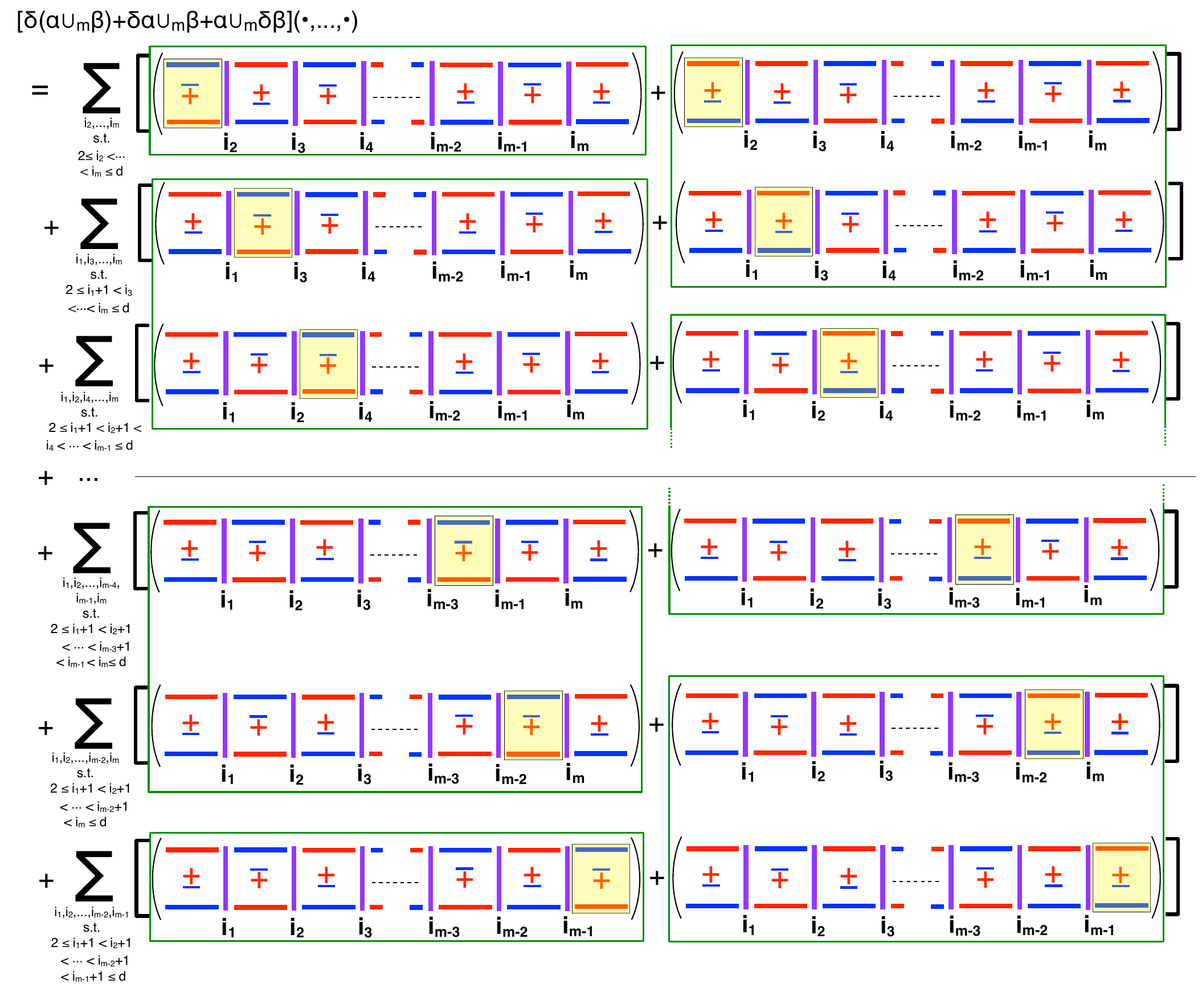}
  \caption[]{Part 1 of Proof of $\delta(\alpha \cup_m \beta) + (\delta \alpha) \cup_m \beta + \alpha \cup_m (\delta \beta) = \alpha \cup_{m-1} \beta + \beta \cup_{m-1} \alpha$ (for m even).
  \begin{itemize}
      \item This equality is similar to the second equality of Figure~\ref{cup2BoundaryProof}. Every line is obtained from the partial sum over one of of the original summation variables $\{i_1,\dots,i_m\}$. Each of these partial sums telescopes in the same way as the second equality of Figure~\ref{cup2BoundaryProof}. The boxes shaded in yellow are the final result of each telescoped part. Note the summation indices restrict which values of of the remaining $\{i_1,\dots,i_m\}$ are allowed.
      \item The green boxes are meant to pair up terms for the next equality in the proof, where a partial sum will telescope and cancel most terms except for some boundary terms. Note that the first and last green boxes containing only one term look similar to the $\cup_{m-1}$ definitions, except for slight differences in the index ranges.
  \end{itemize}
  }
  \label{cupMBoundaryProof_Part1}
\end{figure}

\begin{figure}[h!]
  \captionsetup{singlelinecheck=off}
  \centering
  \includegraphics[width=\linewidth]{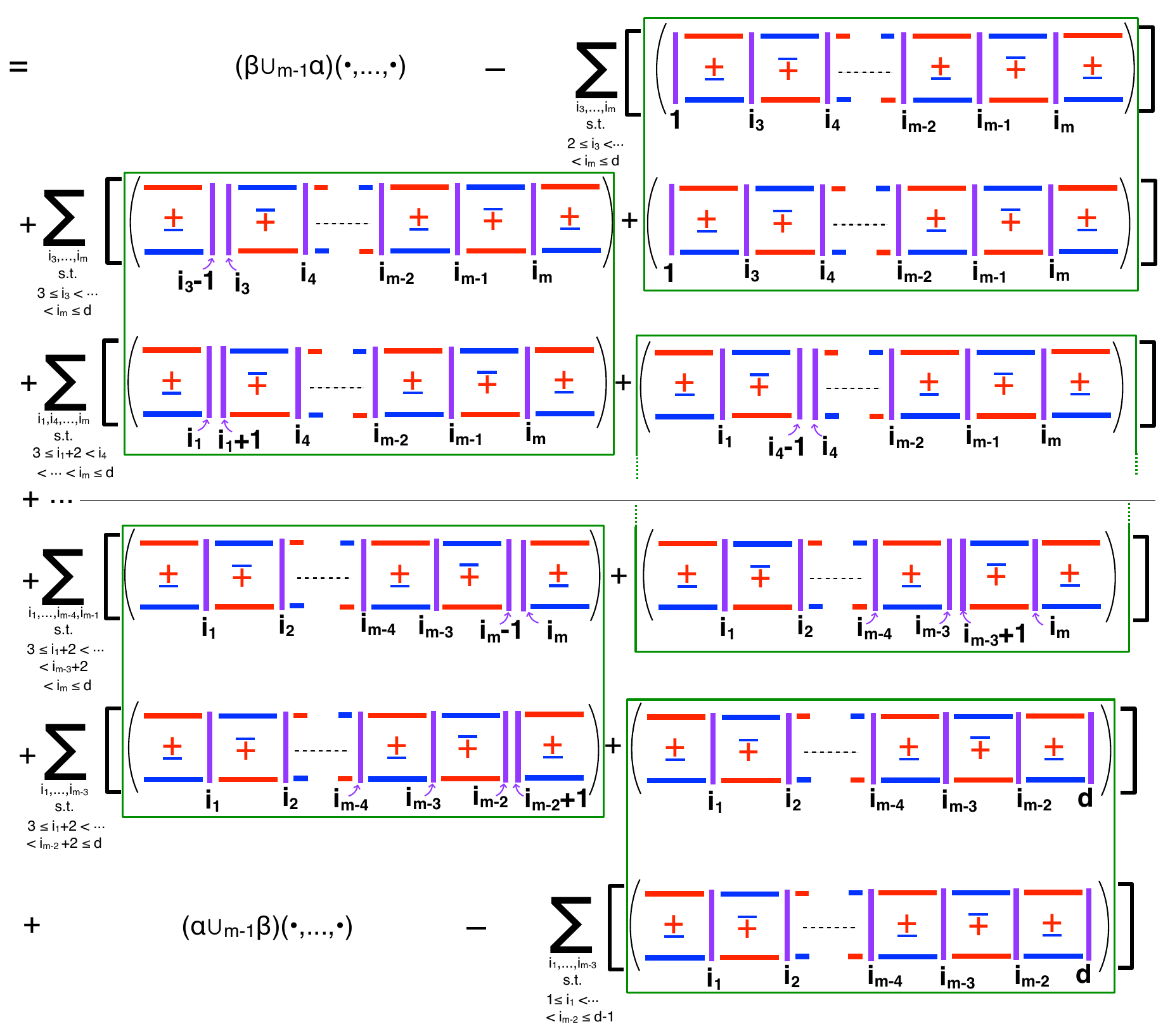}
  \caption[]{Part 2 of Proof of $\delta(\alpha \cup_m \beta) + (\delta \alpha) \cup_m \beta + \alpha \cup_m (\delta \beta) = \alpha \cup_{m-1} \beta + \beta \cup_{m-1} \alpha$ (for m even).
  \begin{itemize}
      \item The first and last lines correspond to the first and last green boxes of the previous part. Note that those green boxes are \textit{almost} the sums $\alpha \cup_{m-1} \beta$ and $\beta \cup_{m-1} \alpha$, except the indices being summed over are different than the $\cup_{m-1}$ definition. So, the first and last lines hold because they are precisely the $\cup_{m-1}$ products minus the terms listed.
      \item Again, every other line is obtained by a partial sum over one of the $\{i_1,\dots,i_m\}$, telescoping each of the green boxes of Part 1.
      \item Note again that we can pair up terms, as in the green boxes, where we will telescope with respect to some partial sums in the next part.
  \end{itemize}
  }
  \label{cupMBoundaryProof_Part2}
\end{figure}

\begin{figure}[h!]
  \captionsetup{singlelinecheck=off}
  \centering
  \includegraphics[width=\linewidth]{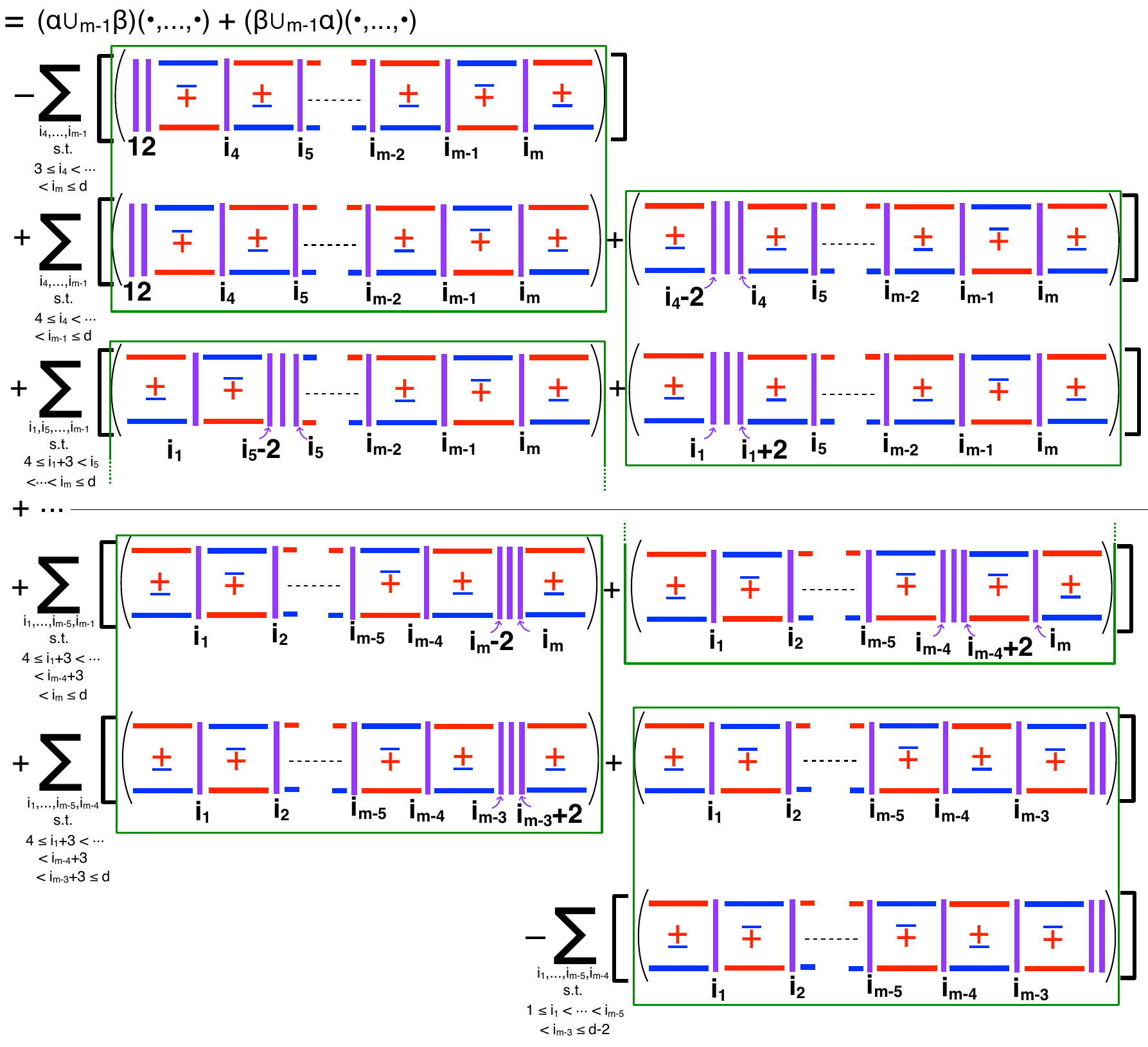}
  \caption[]{Part 3 of Proof of $\delta(\alpha \cup_m \beta) + (\delta \alpha) \cup_m \beta + \alpha \cup_m (\delta \beta) = \alpha \cup_{m-1} \beta + \beta \cup_{m-1} \alpha$ (for m even). 
  \begin{itemize}
      \item First, separate out the $\alpha \cup_{m-1} \beta + \beta \cup_{m-1} \alpha$, as on the top line. 
      \item Again, each other line is obtained by a telescoping sum of each of the green boxes of Part 2, leaving some boundary terms.
      \item And again, we can pair up these terms via the green boxes which we use for telescoping the sum in the next step. 
  \end{itemize}
  }
  \label{cupMBoundaryProof_Part3}
\end{figure}

\begin{figure}[h!]
  \captionsetup{singlelinecheck=off}
  \centering
  \includegraphics[width=\linewidth]{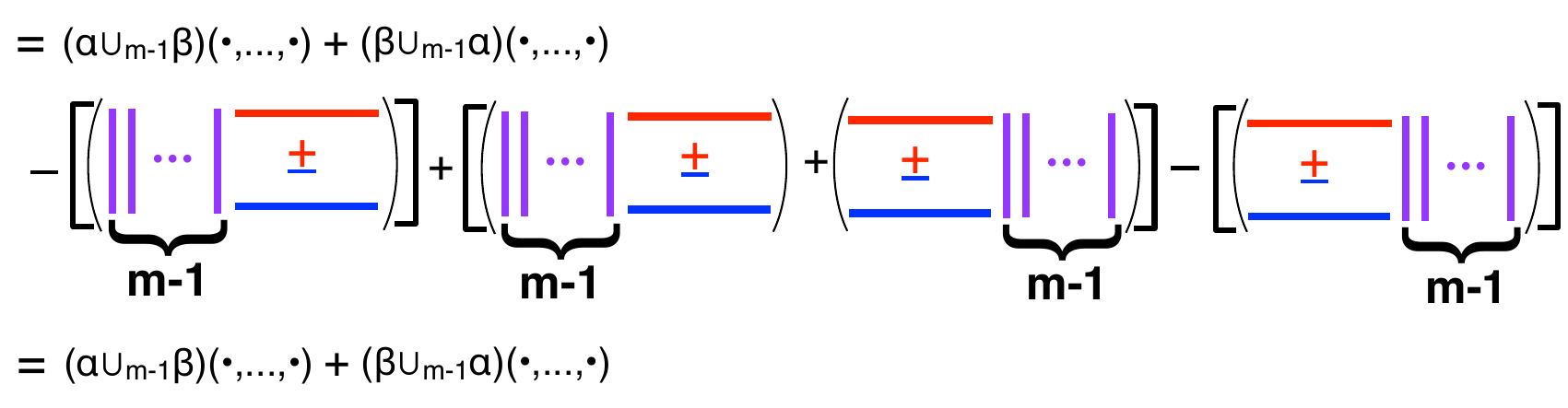}
  \caption[]{Part 4 of Proof of $\delta(\alpha \cup_m \beta) + (\delta \alpha) \cup_m \beta + \alpha \cup_m (\delta \beta) = \alpha \cup_{m-1} \beta + \beta \cup_{m-1} \alpha$ (for m even).
  \begin{itemize}
      \item We can iterate this process of telescoping the terms and pairing them up with each other. Each time we do this, the number of lines in the equation decreases by one.
      \item If we carry this process out the end, we are left with terms where all of the `purple walls' are grouped together like this. Each term in the square brackets is obtained by telescoping terms from the previous equality.
      \item We can cancel all these terms and we're left with just the $\alpha \cup_{m-1} \beta + \beta \cup_{m-1} \alpha$ part, which is what we expected.
  \end{itemize}
  }
  \label{cupMBoundaryProof_Part4}
\end{figure}

\subsection{$\delta(\alpha \cup_m \beta) = \delta \alpha \cup_m \beta + \alpha \cup_m \delta \beta + \alpha \cup_{m-1} \beta + \beta \cup_{m-1} \alpha$}
Now let's verify that the fundamental higher cup product equation 
\begin{equation*}
    \delta(\alpha \cup_m \beta) = \delta \alpha \cup_m \beta + \alpha \cup_m \delta \beta + \alpha \cup_{m-1} \beta + \beta \cup_{m-1} \alpha
\end{equation*}
of Prop~\ref{prop:cupM_coboundary} holds for these definitions. We'll first treat the special cases of $\cup_1$ and $\cup_2$ in more detail and use some lessons there to build up to the general proof. 

We saw before that $\delta(\alpha \cup_0 \beta) = \delta \alpha \cup_0 \beta + \alpha \cup_0 \delta \beta$ is because the coboundary $\delta$ replaces the $\bullet$ indices with either choice of $\{+,-\}$, and these replacements cancel out among each of the terms $\{\delta(\alpha \cup_0 \beta), \delta \alpha \cup_0 \beta, \alpha \cup_0 \delta \beta\}$. Similarly, we'll have that there will be many cancellations in $\delta(\alpha \cup_m \beta) + \delta \alpha \cup_m \beta + \alpha \cup_m \delta \beta$. However, for the $\cup_m$ case not all of the terms will cancel. This is because the term $\delta \alpha \cup_m \beta$ (similarly $\alpha \cup_m \delta \beta$) will have terms like 
\begin{equation*}
\delta\alpha(\dots,\bullet,\dots)\beta(\dots,\bullet,\dots) = \sum_{\dots} \alpha(\dots,+,\dots)\beta(\dots,\bullet,\dots) + \alpha(\dots,-,\dots)\beta(\dots,\bullet,\dots)
\end{equation*}
which are not cancelled out by $\delta(\alpha \cup_m \beta)$. In particular because the terms above will have $\alpha$ and $\beta$ sharing $(m-1)$ $\bullet$-indices, whereas every term in $\delta(\alpha \cup_m \beta)$ will have $\alpha$ and $\beta$ sharing $m$ $\bullet$-indices. But, all of the $m$ $\bullet$-indices terms will indeed cancel, and $\delta(\alpha \cup_m \beta) + \delta \alpha \cup_m \beta + \alpha \cup_m \delta \beta$ will only be left with terms like the one above where $\alpha$ and $\beta$ share exactly $(m-1)$ $\bullet$ indices. And after some cancellations between those terms, the remaining sum will end up exactly equalling $\alpha \cup_m \beta + \beta \cup_m \alpha$.

We'll illustrate this first for $\cup_1$.
\begin{prop}
$\delta(\alpha \cup_1 \beta) = \delta \alpha \cup_1 \beta + \alpha \cup_1 \delta \beta + \alpha \cup_{0} \beta + \beta \cup_{0} \alpha$
\end{prop}

\begin{proof}
See Figure~\ref{cup1BoundaryProof}.
\end{proof}

The proof for $\cup_2$ will be a few more steps and will borrow from what we learned in the $\cup_1$ proof. We illustrate it next. 

\begin{prop}
$\delta(\alpha \cup_2 \beta) = \delta \alpha \cup_2 \beta + \alpha \cup_2 \delta \beta + \alpha \cup_{1} \beta + \beta \cup_{1} \alpha$
\end{prop}

\begin{proof}
See Figure~\ref{cup2BoundaryProof}.
\end{proof}

The general proof for $\cup_m$ will again borrow from what we learned about telescoping terms in the $\cup_1$ and $\cup_2$ proofs. It will be a recursive process of telescoping and we illustrate it in multiple figures.

\begin{proof}[Proof of Prop~\ref{prop:cupM_coboundary}]
See Figs.~\ref{cupMBoundaryProof_Part1}-\ref{cupMBoundaryProof_Part4}. We give the proof for $m$ even, but the case of $m$ odd is basically identical. 
\end{proof}

\afterpage{\clearpage}

\section{Proof of $\cup$ and $\cup_m$ identities over $\Z$} \label{app:proofOfCupM_OverZ}
Now, we prove Propositions~\ref{prop:cup0_coboundary_overZ},\ref{prop:cupM_coboundary_overZ}. In the following, to slightly ease the minus signs in the computations, we will consider a deformed version, the $\cup'_k$ product defined by the same pattern
\begin{equation} \label{cupPrimeK_OverZ_Formula}
(\alpha^{n+k} \cup'_k \beta^{d-n})(\bullet_1, \cdots, \bullet_d) := \sum_{1 =: i_0 \le i_1 < \cdots < i_k \le i_{k+1} := d} \sum_{\substack{(z_1 \cdots z_d), (z'_1 \cdots z'_d) \\ \text{each } z_{i_1} = z'_{i_1} = \cdots = z_{i_k} = z'_{i_k} = \bullet \\ z_j = (-1)^{\ell(j)}+ \text{ and } z'_j = \bullet \text{, OR}\\ z_j = \bullet \text{ and } z'_j = (-1)^{\ell(j)}- \\ \text{where } i_{\ell(j)} < j < i_{\ell(j)+1}}} (-1)^s \alpha(z_1 \cdots z_d) \beta(z'_1 \cdots z'_d)
\end{equation}
where instead the sign factors are
\begin{equation}
(-1)^s = (-1)^{\text{\# of total + signs}} \text{sgn}
\begin{pmatrix}
1   & \cdots & \cdots & \cdots & \cdots & \cdots & \cdots & \cdots & d \\
i_1 & \to    & i_k    & b_1    & \to    & b_{n}  & a_1    & \to    & a_{d-n-k}
\end{pmatrix}.
\end{equation}

We will show that this deformed version satisfies, for any $\alpha$ an $(n-1+k)$-cochain and $\beta$ a $(d-n)$-cochain
\begin{equation}
    \big(\delta\alpha \cup'_k \beta - (-1)^{n-1} \alpha \cup'_k \delta\beta \big) - (-1)^k \delta(\alpha \cup'_k \beta) = (-1)^{d-n} \alpha \cup'_{k-1} \beta + (-1)^{n(d-n-k) + k}\beta \cup'_{k-1} \alpha 
\end{equation}
which is equivalent to the Propositions~\ref{prop:cup0_coboundary_overZ},\ref{prop:cupM_coboundary_overZ} for $\cup$.

\begin{proof}
Analogously to the $\Z_2$ case, we will start by analyzing the terms coming from $\delta(\alpha \cup'_k \beta)$ and show that they all get cancelled out by corresponding terms in $\delta\alpha \cup'_k \beta$ and $\alpha \cup'_k \delta\beta$. The strategy going forward will look like the $\Z_2$ case, where we showed all the `leftover' terms from $\delta\alpha \cup'_k \beta - (-1)^{n-1}\alpha \cup'_k \delta\beta$ give the $\alpha \cup'_{k-1} \beta$ and $\beta \cup'_{k-1} \alpha$ terms plus a series of other terms that all inductively cancelled by a series of telescoping sums. We will see how the sign factors introduced here conspire to cancel out in the telescoping sum. 

Let's consider a term in $\delta(\alpha \cup'_k \beta)$. 
In general, any term $\alpha(z_1 \cdots z_d)\beta(z'_1 \cdots z'_d)$ will have some index $1 \le i \le d$ for which $z_i = z'_i \in \{+,-\}$ which comes from the coboundary $\delta$. 
And there will be some `shared' indices $\{i_1 < \cdots < i_k\} \subset \{1 \cdots d\} \backslash \{i\}$ for which each ${z_{{i_{\dots}}}} = {z'_{{i_{\dots}}}} = \bullet$. 
And defining $i_0 := 0, i_{k+1} :=d$, we'll have that there is some $J \in \{0,\dots,k\}$ for which $i_J < i < i_{J+1}$.
Given these choices $i,i_1,\cdots,i_k$, and for some fixed $z_j$ $j \not\in \{i,i_1,\dots,i_k\}$ constrained as in Eq.~\eqref{cupK_OverZ_Formula}.
, we'll have that a term will look like
\begin{equation*}
\text{from } \delta(\alpha \cup'_k \beta) \,\,\,
\begin{cases}
(-1)^{s_1} &\alpha(\cdots \bullet_{i_1} \cdots \bullet_{i_J} \cdots +_i \cdots \bullet_{i_{J+1}} \cdots \bullet_{i_k} \cdots) \\
&\cdot \beta(\cdots \bullet_{i_1} \cdots \bullet_{i_J} \cdots +_i \cdots \bullet_{i_{J+1}} \cdots \bullet_{i_k} \cdots) \\
+ (-1)^{s_2} &\alpha(\cdots \bullet_{i_1} \cdots \bullet_{i_J} \cdots -_i \cdots \bullet_{i_{J+1}} \cdots \bullet_{i_k} \cdots) \\
&\cdot \beta(\cdots \bullet_{i_1} \cdots \bullet_{i_J} \cdots -_i \cdots \bullet_{i_{J+1}} \cdots \bullet_{i_k} \cdots)
\end{cases}
\end{equation*}
where the $\cdots$ represent the $z_j$,$j \not\in \{i,i_1,\dots,i_k\}$ and the $\pm_i$ represent $+$ or $-$ labels on the $i^{th}$ coordinate. We'll write out the signs $(-1)^{s_{1,2}}$ later on. 

Now we'll consider corresponding terms in $\delta\alpha \cup'_k \beta$. As before, for any term $\delta\alpha(\tilde{z}_1 \cdots \tilde{z}_d)\beta(z'_1 \cdots z'_d)$ from the $\cup'_k$ formula will select out some $1 \le i_1 \cdots i_k \le d$ for which ${\tilde{z}_{{i_{\dots}}}} = {z'_{i_{\dots}}} = \bullet$. Then when we expand out $\delta\alpha(\tilde{z}_1 \cdots \tilde{z}_d)$ into a sum over $\alpha(z_1 \cdots z_d)$, we'll have  that the terms correspond to those where a $\bullet$ in the arguments $\tilde{z}$ of $\delta\alpha$ gets turned into a $+$ or $-$ one at a time. The terms that cancel out the $\delta(\alpha \cup'_k \beta)$ will be those for which the $\bullet$ that gets replaced is \textit{different} from one of the $\bullet_{i_1},\dots,\bullet_{i_k}$. And the terms that contribute to the $\alpha \cup'_{k-1} \beta$ and $\beta \cup'_{k-1} \alpha$ and cancel out in the telescoping sum will be those for which one of the $\bullet_{i_1},\dots,\bullet_{i_k}$ gets replaced by $+$ or $-$. A term from $\delta\alpha \cup'_k \beta$ cancelling the $\delta(\alpha \cup'_k \beta)$ will come from $\bullet_i$ turning into $\pm$ and will look like
\begin{equation*}
\text{from } \delta\alpha \cup'_k \beta \,\,\,
\begin{cases}
(-1)^{s_3} &\alpha(\cdots \bullet_{i_1} \cdots \bullet_{i_J} \cdots +_i \cdots \bullet_{i_{J+1}} \cdots \bullet_{i_k} \cdots) \\
&\cdot \beta(\cdots \bullet_{i_1} \cdots \bullet_{i_J} \cdots (-1)^J -_i \cdots \bullet_{i_{J+1}} \cdots \bullet_{i_k} \cdots) \\
+ (-1)^{s_4} &\alpha(\cdots \bullet_{i_1} \cdots \bullet_{i_J} \cdots -_i \cdots \bullet_{i_{J+1}} \cdots \bullet_{i_k} \cdots) \\
&\cdot \beta(\cdots \bullet_{i_1} \cdots \bullet_{i_J} \cdots (-1)^J -_i \cdots \bullet_{i_{J+1}} \cdots \bullet_{i_k} \cdots)
\end{cases}
\end{equation*}
where the $(-1)^J -_i$ argument goes into $\beta$ as to be consistent with conditions in Eq.~\eqref{cupK_OverZ_Formula}. 

Similarly, we'll have that the terms from $\alpha \cup_k \delta\beta$ that go to cancel the $\delta(\alpha \cup'_k \beta)$ will look like
\begin{equation*}
\text{from } \alpha \cup'_k \delta\beta \,\,\,
\begin{cases}
(-1)^{s_5} &\alpha(\cdots \bullet_{i_1} \cdots \bullet_{i_J} \cdots (-1)^J +_i \cdots \bullet_{i_{J+1}} \cdots \bullet_{i_k} \cdots) \\
&\cdot \beta(\cdots \bullet_{i_1} \cdots \bullet_{i_J} \cdots +_i \cdots \bullet_{i_{J+1}} \cdots \bullet_{i_k} \cdots) \\
+ (-1)^{s_6} &\alpha(\cdots \bullet_{i_1} \cdots \bullet_{i_J} \cdots (-1)^J +_i \cdots \bullet_{i_{J+1}} \cdots \bullet_{i_k} \cdots) \\
&\cdot \beta(\cdots \bullet_{i_1} \cdots \bullet_{i_J} \cdots -_i \cdots \bullet_{i_{J+1}} \cdots \bullet_{i_k} \cdots)
\end{cases}
\end{equation*}
coming from replacing $\bullet_i$, $i \not\in \{i_1 \cdots i_k\}$ $\pm$, and where the arguments $(-1)^J +_i$ in $\alpha$ come from the constraints from the $\cup'_k$ formula. 

Now, we write down formulas for the signs $(-1)^{s_{1 \cdots 6}}$. In the following, we'll use ``$\text{total \# of + signs}$'' to refer to the total number of $+$ signs among all $z_1,z'_1,\dots,z_d,z'_d$. And, we'll refer to $\tilde{b}_1 < \cdots < \tilde{b}_{n-1}$ as the subset of arguments ${z'_{\dots}}$ of $\beta$ that are either of $\pm$, \textit{except for} $i$. And we refer to $b_1 < \cdots < b_n = \{\tilde{b}_1, \dots, \tilde{b}_{n-1}\} \cup \{i\}$. Similarly, $\tilde{a}_1 < \cdots < \tilde{a}_{d-n-k}$ refer to the subset of arguments ${z_{\dots}}$ of $\alpha$ that are either of $\pm$ except $i$, and $a_1 < \cdots < a_{d-n+k+1} = \{\tilde{a}_1, \dots, \tilde{a}_{d-n+k}\} \cup \{i\}$. In addition, we'll define $\ell_b(i) \in \{1,\dots,n\}$ and $\ell_a(i) \in \{1,\dots,d-n+k+1\}$ to be the positions of $i$ within the sets $b_1 < \cdots < b_n$ and $a_1 < \cdots < a_{d-n+k+1}$ respectively. We'll have
\begin{equation*}
\begin{split}
(-1)^{s_1} &= \text{sgn}
\begin{pmatrix}
1   & \cdots & \cdots & i-1         & ;   & i+1             & \cdots      & \cdots & d \\
i_1 & \to    & i_k    & \tilde{b}_1 & \to & \tilde{b}_{n-1} & \tilde{a}_1 & \to    & \tilde{a}_{d-n-k} 
\end{pmatrix}
(-1)^{(\text{total \# of + signs})-1} (-1)^i \\
(-1)^{s_2} &= \text{sgn}
\begin{pmatrix}
1   & \cdots & \cdots & i-1         & ;   & i+1             & \cdots      & \cdots & d \\
i_1 & \to    & i_k    & \tilde{b}_1 & \to & \tilde{b}_{n-1} & \tilde{a}_1 & \to    & \tilde{a}_{d-n-k} 
\end{pmatrix}
(-1)^{(\text{total \# of + signs})} (-1)^i \\
(-1)^{s_{3,4}} &= \text{sgn}
\begin{pmatrix}
1   & \cdots & \cdots & \cdots & \cdots & \cdots & \cdots      & \cdots & d \\
i_1 & \to    & i_k    & b_1    & \to    & b_n    & \tilde{a}_1 & \to    & \tilde{a}_{d-n-k} 
\end{pmatrix}
(-1)^{(\text{total \# of + signs})} (-1)^{J + \ell_b(i)} \\
(-1)^{s_{5,6}} &= \text{sgn}
\begin{pmatrix}
1   & \cdots & \cdots & \cdots      & \cdots & \cdots          & \cdots & \cdots & d \\
i_1 & \to    & i_k    & \tilde{b}_1 & \to    & \tilde{b}_{n-1} & a_1    & \to    & a_{d-n-k+1} 
\end{pmatrix}
(-1)^{(\text{total \# of + signs})} (-1)^{J + \ell_a(i)} 
\end{split}
\end{equation*}
First, note that the factors $(-1)^\text{\# of + signs}$ appears because of the $\cup'_k$ defintion and that the $\mp$ sign in the definition of $\delta$ in Eq.~\eqref{coboundaryDefinitionOverZ} gives a $(-1)$ whenever $\bullet$ is replaced with a $+$. Note also that the $(-1)^{s_1}$ term has a factor of $(-1)^{\text{\# of + signs}-1}$ from these factors because the two $+_i,+_i$ in $\alpha,\beta$ only gained one $(-1)$ factor in total from $\delta$ in $\delta(\alpha \cup'_k \beta)$. The permutations given are the ones that appear in the $\cup'_k$ definitions, noting that for $(-1)^{s_{1,2}}$ the permutations only act on the indices $\{1 \cdots d\} \backslash \{i\}$. The factors of $(-1)^{i}$ in $(-1)^{s_{1,2}}$ come from the corresponding factor $(-1)^\ell$ in Eq~\eqref{coboundaryDefinitionOverZ}. And the $(-1)^{J+\ell_{b,a}(i)}$ factors in $(-1)^{s_{\{3,4\},\{5,6\}}}$ come from the fact that the $\bullet$ that gets replaced by $i$ in $\alpha$ (resp $\beta$) corresponds to a $\pm$ coordinate in $\beta$ (resp $\alpha$) which appears as the $(J+\ell_b(i))^{th}$ (resp $(J+\ell_a(i))^{th}$) $\bullet$ coordinate in $\alpha$ (resp $\beta$). 

In total, we note that we can rewrite the above signs as
\begin{equation*}
\begin{split}
(-1)^s = & (-1)^\text{total \# of + signs} \text{sgn}
\begin{pmatrix}
1   & \cdots & \cdots & \cdots      & \cdots & \cdots          &  \cdots & \cdots      & \cdots & d \\
i_1 & \to    & i_k    & \tilde{b}_1 & \to    & \tilde{b}_{n-1} & i       & \tilde{a}_1 & \to    & \tilde{a}_{d-n-k} 
\end{pmatrix} \\
& \times 
\begin{cases}
-(-1)^{k+n} \text{ if } s = s_1 \\
(-1)^{k+n} \text{ if } s = s_2 \\
(-1)^{J+n} \text{ if } s = s_{3,4} \\
-(-1)^{J} \text{ if } s = s_{5,6} \\
\end{cases}
\end{split}
\end{equation*}
by counting the signs needed to bring the original permutations to the one shown above. To reiterate, these signs came from looking at terms in $\delta(\alpha \cup'_k \beta)$ that we want to show are all eliminated by some terms in $\delta \alpha \cup'_k \beta$ and $\alpha \cup'_k \delta\beta$. From here we can see the combination of signs as in this proposition ensure that the subset of terms we've looked at in $(\delta\alpha^{n-1+k} \cup'_k \beta^{d-n} - (-1)^{n-1} \alpha^{n-1+k} \cup'_k \delta\beta^{d-n})$ cancel out all terms in $(-1)^k\delta(\alpha^{n-1+k} \cup'_k \beta^{d-n})$. In particular, if $J$ is even, we'd need
\begin{equation*}
(-1)^k (-1)^{s_1} = -(-1)^{n-1} (-1)^{s_5},\quad (-1)^k (-1)^{s_2} = (-1)^{s_4},\quad (-1)^{s_3} - (-1)^{n-1}(-1)^{s_6} = 0
\end{equation*}
corresponding to the cancellations of the $(z_i,z'_i) = (+_i,+_i),(-_i,-_i),(+_i,-_i)$ terms respectively. While for $J$ odd we need
\begin{equation*}
(-1)^k (-1)^{s_1} = (-1)^{s_3},\quad (-1)^k (-1)^{s_2} = -(-1)^{n-1}(-1)^{s_6}.\quad (-1)^{s_4} - (-1)^{n-1}(-1)^{s_5} = 0
\end{equation*}
corresponding to the cancellations of the $(z_i,z'_i) = (+_i,+_i),(-_i,-_i),(-_i,+_i)$ terms. One can check that these are true which indeed implies that all terms from $\delta(\alpha \cup'_k \beta)$ are accounted for. 

Note that at this point, the $k=0$ case for the $\cup'_0$ formula is now proved. 

Now we want to analyze the remaining terms in $(\delta\alpha \cup'_k \beta - (-1)^{n-1} \alpha \cup'_k \delta \beta) - (-1)^k \delta(\alpha \cup'_k \beta)$ and show that they give terms some terms proportional to $\alpha \cup'_{k-1} \beta$ and $\beta \cup'_{k-1} \alpha$ while all the remaining terms telescope away as in the proofs of the $\Z_2$ case in Figs.~\ref{cup1BoundaryProof},~\ref{cup2BoundaryProof},~\ref{cupMBoundaryProof_Part1}-\ref{cupMBoundaryProof_Part4}.

We'll first explicitly give the proof of these for the $\cup'_1$ case following Fig.~\ref{cup1BoundaryProof}. After this, we'll briefly outline the general case and leave the explicit computations to the reader.

Note that the first line of that figure follows from our analysis of the cancellations of $\delta(\alpha \cup'_k \beta)$. Now, we note that the terms that go into being proportional to $\beta \cup'_0 \alpha$ come from the terms $\delta\alpha \cup'_1 \beta$ with $i_1=1$ and $z_1 = \bullet_{i_1}$ being replaced by $-$ and from $-(-1)^{n-1} \alpha \cup'_1 \delta\beta$ with $i_1=1$ and $z'_1 = \bullet_{i_1}$ replaced with $+$. The sign factor for both of these terms from $\delta\alpha \cup'_1 \beta$ or $\alpha \cup'_1 \delta\beta$ can be written:
\begin{equation*}
\begin{split}
\text{sgn} &
\begin{pmatrix}
1   & \cdots      & \cdots & \cdots          &  \cdots     & \cdots & d \\
1   & \tilde{b}_1 & \to    & \tilde{b}_{n} & \tilde{a}_1 & \to    & \tilde{a}_{d-n-1} 
\end{pmatrix} (-1)^\text{\# of total + signs} \underbrace{(-1)}_{\text{from } \delta}  \text{  from  } \delta\alpha \cup'_1 \beta
\\
-(-1)^{n-1}
\text{sgn} &
\begin{pmatrix}
1   & \cdots      & \cdots & \cdots          &  \cdots     & \cdots & d \\
1   & \tilde{b}_1 & \to    & \tilde{b}_{n-1} & \tilde{a}_1 & \to    & \tilde{a}_{d-n}
\end{pmatrix} (-1)^\text{\# of total + signs} \underbrace{(-1)}_{\text{from } \delta}  \text{  from  } \alpha \cup'_1 \delta\beta
\end{split}
\end{equation*}
where $\tilde{b}$ are the coordinates $z'_j=\pm,j>1$ inputted into $\beta$ and the $\tilde{a}$ are the coordinates $z_j=\pm,j>1$ that get inputted into $\alpha$. And we note that the $\tilde{b}$ are labeled by $1 \to n$ and $\tilde{a}$ are enumerated by from $1 \to d-n-1$ for the $\delta\alpha \cup'_1 \beta$ terms, and they're labeled by $1 \to n-1$ and $1 \to d-n$ resp for the $\alpha \cup'_1 \delta\beta$. 
And, the sign factors that would come from the terms in $\beta \cup'_0 \alpha$ would be 
\begin{equation*}
\begin{split}
\text{sgn}
\begin{pmatrix}
1 & \cdots      & \cdots & \cdots            & \cdots      & \cdots & d \\
1 & \tilde{a}_1 & \to    & \tilde{a}_{d-n-1} & \tilde{b}_1 & \to    & \tilde{b}_{n} 
\end{pmatrix} (-1)^\text{\# of total + signs} \text{  for  } \begin{pmatrix}z_1 \\ z'_1\end{pmatrix} = \begin{pmatrix} - \\ \bullet \end{pmatrix}
\\
\text{sgn}
\begin{pmatrix}
1           & \cdots & \cdots          & \cdots & \cdots      & \cdots & d \\
\tilde{a}_1 & \to    & \tilde{a}_{d-n} & 1      & \tilde{b}_1 & \to    & \tilde{b}_{n-1} 
\end{pmatrix} (-1)^\text{\# of total + signs} \text{  for  } \begin{pmatrix}z_1 \\ z'_1\end{pmatrix} = \begin{pmatrix} \bullet \\ + \end{pmatrix}
\end{split}
\end{equation*}
Note that the top term above is the term we want to match to $\delta\alpha \cup'_1 \beta$ and the bottom term above is meant for $\alpha \cup'_1 \delta\beta$. And comparing the signs from the permutations gives that these contributions total up to $(-1)^{1+n+n(d-n)} \beta \cup'_0 \alpha$. 

We can carry out a similar analysis for the terms that contribute to being proportional to $\alpha \cup'_0 \beta$, which would be terms from $\delta\alpha \cup'_1 \beta$ with $i_1=d$ and $z_d = \bullet_{i_1}$ replaced with $+$, and from $\alpha \cup'_1 \delta\beta$ with $i_1=d$ and $z'_d = \bullet_{i_1}$ replaced with $-$. The sign factors for these terms would give:
\begin{equation*}
\begin{split}
\text{sgn} &
\begin{pmatrix}
1   & \cdots      & \cdots & \cdots          &  \cdots     & \cdots & d \\
d   & \tilde{b}_1 & \to    & \tilde{b}_{n} & \tilde{a}_1 & \to    & \tilde{a}_{d-n-1} 
\end{pmatrix} (-1)^\text{\# of total + signs} \underbrace{(-1)^{n+1}}_{\text{from } \delta}  \text{  from  } \delta\alpha \cup'_1 \beta
\\
-(-1)^{n-1}
\text{sgn} &
\begin{pmatrix}
1   & \cdots      & \cdots & \cdots          &  \cdots     & \cdots & d \\
d   & \tilde{b}_1 & \to    & \tilde{b}_{n-1} & \tilde{a}_1 & \to    & \tilde{a}_{d-n}
\end{pmatrix} (-1)^\text{\# of total + signs} \underbrace{(-1)^{d-n+1}}_{\text{from } \delta}  \text{  from  } \alpha \cup'_1 \delta\beta
\end{split}
\end{equation*}
And, the sign factors that come from $\alpha \cup'_0 \beta$ would look like 
\begin{equation*}
\begin{split}
\text{sgn}
\begin{pmatrix}
1           & \cdots & \cdots        & \cdots      & \cdots & \cdots            & d \\
\tilde{b}_1 & \to    & \tilde{b}_{n} & \tilde{a}_1 & \to    & \tilde{a}_{d-n-1} & d      
\end{pmatrix} (-1)^\text{\# of total + signs} \text{  for  } \begin{pmatrix}z_d \\ z'_d\end{pmatrix} = \begin{pmatrix} + \\ \bullet \end{pmatrix}
\\
\text{sgn}
\begin{pmatrix}
1           & \cdots & \cdots          & \cdots & \cdots      & \cdots & d \\
\tilde{b}_1 & \to    & \tilde{b}_{n-1} & d      & \tilde{a}_1 & \to    & \tilde{a}_{d-n}
\end{pmatrix} (-1)^\text{\# of total + signs} \text{  for  } \begin{pmatrix}z_d \\ z'_d\end{pmatrix} = \begin{pmatrix} \bullet \\ - \end{pmatrix}
\end{split}
\end{equation*}
This means that the signs from these permutations give contributions that total to $(-1)^{d+n}(\alpha \cup'_0 \beta)$.

Now, we should verify that the remaining terms left over indeed cancel out exactly, as in the telescoping sum in the second to third line of Fig.~\ref{cup1BoundaryProof}. Note that the telescoping part of the sum can be expressed:
\begin{equation*}
\begin{split}
\sum_{i=1}^{d-1} \sum_{\substack{(z_1 \cdots \hat{z}_i \hat{z}_{i+1} \cdots z_d) \\ (z'_1 \cdots \hat{z}'_i \hat{z}'_{i+1} \cdots z'_d)}} 
\bigg(
     & \alpha(\cdots +_i       -_{i+1}       \cdots) \beta (\cdots \bullet_i \bullet_{i+1} \cdots) ((-1)^{t_1}+(-1)^{t_2}) \\
   + & \alpha(\cdots +_i       \bullet_{i+1} \cdots) \beta (\cdots \bullet_i +_{i+1}       \cdots) ((-1)^{t_3}+(-1)^{t_4}) \\
   + & \alpha(\cdots \bullet_i -_{i+1}       \cdots) \beta (\cdots -_i       \bullet_{i+1} \cdots) ((-1)^{t_5}+(-1)^{t_6}) \\
   + & \alpha(\cdots \bullet_i \bullet_{i+1} \cdots) \beta (\cdots -_i       +_{i+1}       \cdots) ((-1)^{t_7}+(-1)^{t_8}) 
\bigg)
\end{split}
\end{equation*}
where $\hat{z}_i,\hat{z}_{i+1}$ mean don't sum over them (since they're summed over in the terms of the summand above). Here each of the $(-1)^{t_1} + (-1)^{t_2},\dots,(-1)^{t_7} + (-1)^{t_8}$ represent a sum of signs from two terms that we want to show are zero. We'll show explicitly how $(-1)^{t_1} + (-1)^{t_2} = (-1)^{t_3} + (-1)^{t_4} = 0$ and leave the other two cases to the readers. In the following we'll refer to $\tilde{b}$ as the $j$ for which $z'_j = \pm$ and $\tilde{a}$ as the $j$ for which $z_j = \pm$.

First, note that both the $t_1$ and $t_2$ terms come from terms in $\delta \alpha \cup'_1 \beta$. In particular, the $t_1$ term comes from $i'_1 = i$ in the $\cup'_1$ formula and replacing $z_i = \bullet_{i'_1} \to +_i$ while the $t_2$  term comes from $i''_1 = i+1$ in $\cup'_1$ and replacing $z_{i+1} = \bullet_{i''_1} \to -_{i+1}$. This means that the signs coming from these are
\begin{equation*}
\begin{split}
(-1)^{t_1} &= 
\text{sgn}
\begin{pmatrix}
1 & \cdots      & \cdots & \cdots        & \cdots & \cdots & d \\
i & \tilde{b}_1 & \to    & \tilde{b}_{n} & a'_1   & \to    & a'_{d-n-1} 
\end{pmatrix} (-1)^\text{\# of total + signs} (-1)^{\#\{\tilde{b} < i\}}
\\
(-1)^{t_2} &= 
\text{sgn}
\begin{pmatrix}
1   & \cdots      & \cdots & \cdots        & \cdots  & \cdots & d \\
i+1 & \tilde{b}_1 & \to    & \tilde{b}_{n} & a''_1   & \to    & a''_{d-n-1}     
\end{pmatrix} (-1)^\text{\# of total + signs} (-1)^{\#\{\tilde{b} < i+1\}}
\end{split}
\end{equation*}
where here we refer to $\{a'_1 < \cdots < a'_{d-n-1}\} = \{\tilde{a}_1 \cdots \tilde{a}_{d-n-2}\} \cup \{i\}$ and $\{a''_1 < \cdots < a''_{d-n-1}\} = \{\tilde{a}_1 \cdots \tilde{a}_{d-n-2}\} \cup \{i+1\}$, and $\#\{\tilde{b} < i\}$ is the number of $\tilde{b}$ coordinates less than $i$. Note that in this case $\#\{\tilde{b} < i\} = \#\{\tilde{b} < i+1\}$. From here we can see that $(-1)^{t_1} = -(-1)^{t_2}$ because the two permutation differ by a transposition of $i \leftrightarrow i+1$ which gives one $(-1)$ factor. 

Now, note that the $t_3$ term comes from $\delta \alpha \cup'_1 \beta$ with $z_i = i'_1 = i$ and replacing $\bullet_{i'_1} \to +_i$, and the $t_4$ term comes from $\alpha \cup'_1 \delta\beta$ with $i''_1 = i+1$ and replacing $z'_{i''_1} = \bullet_{i''_1} \to +_{i+1}$. Note that here, $\tilde{b}$ are indexed from $1 \to n-1$ and $\tilde{a}$ are indexed from $1 \to (d-n-1)$ to have indices compatible with cochain degrees. We have:
\begin{equation*}
\begin{split}
(-1)^{t_3} = 
& \text{sgn}
\begin{pmatrix}
1 & \cdots      & \cdots & \cdots          & \cdots  & \cdots & d \\
i & \tilde{b}_1 & \to    & \tilde{b}_{n-1} & a^{*}_1 & \to    & a^{*}_{d-n}
\end{pmatrix} (-1)^\text{\# of total + signs} (-1)^{\#\{\tilde{b} < i\}}
\\
(-1)^{t_4} = -(-1)^{n-1}
& \text{sgn}
\begin{pmatrix}
1   & \cdots  & \cdots & \cdots  & \cdots      & \cdots & d \\
i+1 & b^{*}_1 & \to    & b^{*}_n & \tilde{a}_1 & \to    & \tilde{a}_{d-n-1}
\end{pmatrix} (-1)^\text{\# of total + signs} (-1)^{\#\{\tilde{a} < i\} + 1}
\end{split}
\end{equation*}
where $\{a^{*}_1 < \cdots < a^{*}_{d-n}\} = \{\tilde{a}_1,\dots,\tilde{a}_{d-n-1}\} \cup \{i+1\}$ and $\{b^{*}_1 < \cdots < b^{*}_{n}\} = \{\tilde{b}_1,\dots,\tilde{b}_{n-1}\} \cup \{i\}$, and $\#\{\tilde{b} < i\}$ is as before and $\#\{\tilde{a} < i\}$ is the number of $\tilde{a}$ less than $a$. These all imply that 
\begin{equation*}
-(-1)^{t_3} = (-1)^{t_4} = 
\text{sgn}
\begin{pmatrix}
1       & \cdots & \cdots  & \cdots  & \cdots & d \\
b^{*}_1 & \to    & b^{*}_n & a^{*}_1 & \to    & a^{*}_{d-n}
\end{pmatrix} (-1)^\text{\# of total + signs}
\end{equation*}
which implies that the terms cancel.

Similar arguments give that $(-1)^{t_5}+(-1)^{t_6}=(-1)^{t_7}+(-1)^{t_8}=0$. And at this point, the $\cup'_1$ formula is proved over $\Z$. We note that essentially the exact same procedure of analyzing the signs of permutations works for general $\cup_k$ products to reproduce the result of Part 1 of the proof over $\Z_2$ in Fig.~\ref{cupMBoundaryProof_Part1}. 

From this point, the proof of the rest of the iterated telescoping argument for the $\cup'_k$ proof follows the same pattern as in Figs.~\ref{cupMBoundaryProof_Part1}-\ref{cupMBoundaryProof_Part4}. For example, Part 2 of the proof follows from two things. First, a computation showing that the first and last green boxes of Fig.~\ref{cupMBoundaryProof_Part1} give $(-1)^{n+k}\big(\alpha \cup'_{k-1} \beta + (-1)^{n(d-n-k+1)}\beta \cup'_{k-1} \alpha\big)$ minus the leftover terms on the first and last lines of Fig.~\ref{cupMBoundaryProof_Part2}. Second, by considering the collection of minus signs from the $\cup'_k$ definition, we'll have that all other pairs of terms in green boxes will have opposite signs that cancel out. It will always be the case from the $\cup'_k$ and $\delta$ definitions that there is a factor of $(-1)^\text{\# total + signs}$, a sign of some permutation, and the $(-1)^\ell$ term from $\delta$, and carefully writing out the sign factors cancel out in a similar way from the explicit account above of the first cancellations. Then we can iterate this procedure as in Figs.~\ref{cupMBoundaryProof_Part3}-\ref{cupMBoundaryProof_Part4} and the same procedure and sign cancellations hold eventually only leaving behind $(-1)^{n+k}\big(\alpha \cup'_{k-1} \beta + (-1)^{n(d-n-k+1)}\beta \cup'_{k-1} \alpha\big)$.
\end{proof}

\section{Framings of Curves, winding numbers, and computing induced spin strurctures} \label{app:framingsOfCurvesAndWindings}

\subsection{Framings}
First, we quickly review the framings that will be necessary. Since we're on a hypercubic lattice, it'll be a bit easier to do this than in the case of triangulations

First, we note that the background framing will be nothing more than the constant frame of vector fields that we used to define the thickening and shifting. So, we define the \textit{background framing} as
\begin{equation}
F^\text{bkgd} =
  \left(\begin{array}{@{}ccccccc}
    1      & 1                 & \cdots & 1                 \\
    1      & \frac{1}{2}       & \cdots & \frac{1}{d}       \\
    \vdots &                   &        & \vdots            \\
    1      & \frac{1}{2^{d-3}} & \cdots & \frac{1}{d^{d-3}} \\  \hline
    1      & \frac{1}{2^{d-2}} & \cdots & \frac{1}{d^{d-2}} \\ 
    1      & \frac{1}{2^{d-1}} & \cdots & \frac{1}{d^{d-1}} \\ 
  \end{array}\right)
\end{equation}
whose rows will give the local coordinates of the $d$ vector fields we call $v^{\text{bckd}}_1 \cdots v^{\text{bckd}}_d$ in local coordinates. And the \textit{shared framing} will be the first $(d-2)$ of these vectors
\begin{equation}
F^\text{shared} =
  \left(\begin{array}{@{}ccccccc}
    1      & 1                 & \cdots & 1                 \\
    1      & \frac{1}{2}       & \cdots & \frac{1}{d}       \\
    \vdots &                   &        & \vdots            \\
    1      & \frac{1}{2^{d-3}} & \cdots & \frac{1}{d^{d-3}} \\
  \end{array}\right)
\end{equation}
whose rows we'll call $v^{\text{shared}}_1 \cdots v^{\text{shared}}_d$, and will define a framings of any curves living on the dual 1-skeleton.

Now, let's define some curve that passes through some hypercube $\msquare_d$ along the dual 1-skeleton. First, we introduce short-hand for the line segments of the dual 1-skeleton. In the notation of this paper, the dual 1-skeleton is given by the cells $P^\vee_{(z_1 \cdots z_d)}$ such that exactly one $z_i \in {+,-}$ and all of the other $z_{j \neq i} = \bullet$. 


Our goal will be to compute the total winding along the curve in terms of partial windings obtained from traversing between $\ihat_{\pm_1} \to \jhat_{\pm_2}$. It's helpful to decompose the curve into several parts with reference points (A,B,C,D) along the curve, as in Fig.~\ref{pathForWinding}. From here forward, we'll assume WLOG that $i \neq j$; there would be no winding traversing from $\ihat_{\pm} \to \ihat_{\mp}$ because in that case, the curve moves in a `straight line' with respect to the background framing so there's no winding.

\begin{figure}[h!]
  \centering
  \includegraphics[width=0.25\linewidth]{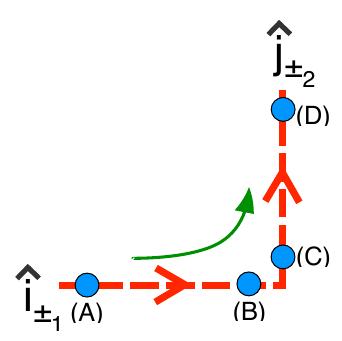}
  \caption{Path of a curve traversing along $\ihat_{\pm_1} \to \jhat_{\pm_2}$. We mark four convenient reference points (A,B,C,D) along the the curve.}
  \label{pathForWinding}
\end{figure}

Now, we define the tangent framing along this curve's segments \{(A) $\to$ (B),(B) $\to$ (C),(C) $\to$ (D)\}. Below, we use the notation that $(-1)^{p(1)}=\pm_1 1,(-1)^{p(2)}=\pm_2 1$ and $\delta_{i < j}$ is $1$ if $i<j$ and $0$ otherwise.

The tangent framing will consist of $d$ vector fields $v^{\text{tang}}_1,\cdots,v^{\text{tang}}_d$, the first $(d-2)$ of which are the `shared framing' with the background, so $v^{\text{tang}}_k = v^{\text{shared}}_k = v^{\text{bckd}}_k$ for $k \in \{1,\cdots,d-2\}$. The $d^{th}$ vector $v^{\text{tang}}_d$ will be proportional to the curve's tangent vector and the $(d-1)^{th}$ vector will be chosen so that $\det(F^\text{tang})$ has the same sign as $\det(F^\text{bckd})$, which is $\propto (-1)^{d \choose 2}$. Below, we write out the framing, where for the `(B) $\to$ (C)' segment, $t$ represents a parameter of the segment and ranges from $0 \to 1$.

\begin{equation}
F^\text{tang} = 
  \left(\begin{array}{@{}ccccccccccc}
    1      & \cdots & 1                     & 1                  & 1                     & \cdots & 1                     & 1                              & 1                     & \cdots & 1                  \\
    
    1      & \cdots & \frac{1}{i-1}         & \frac{1}{i}        & \frac{1}{i+1}         & \cdots & \frac{1}{j-1}         & \frac{1}{j}                    & \frac{1}{j+2}         & \cdots & \frac{1}{d}        \\
    
    \vdots &        & \vdots                &                    & \vdots                &        & \vdots                &                                & \vdots                &        & \vdots             \\
    
    1      & \cdots & \frac{1}{(i-1)^{d-232}}  & \frac{1}{(i+1)^{d-3}} & \cdots & \frac{1}{(j-1)^{d-3}} & \frac{1}{j^{d-3}}              & \frac{1}{(j+2)^{d-3}} & \cdots & \frac{1}{d^{d-3}}  \\ \hline
    
    0      & \cdots & 0                     & 0                  & 0                     & \cdots & 0                     & (-1)^{i+j+\delta_{i < j}+p(1)} & 0                     & \cdots & 0                  \\
    
    0      & \cdots & 0                     & (-1)^{p(1)+1}      & 0                     & \cdots & 0                     & 0                              & 0                     & \cdots & 0                  \\
    
  \end{array}\right) \text{ for  (A) } \to \text{ (B)}
\end{equation} 

\begin{equation}
F^\text{tang} = 
  \left(\begin{array}{@{}ccccccccccc}
    1      & \cdots & 1                                 & \cdots & 1                                     & \cdots & 1                  \\
    
    1      & \cdots & \frac{1}{i}                       & \cdots & \frac{1}{j}                           & \cdots & \frac{1}{d}        \\
    
    \vdots &        &                                   &        &                                       &        & \vdots             \\
    
    1      & \cdots & \frac{1}{i^{d-3}}                 & \cdots & \frac{1}{j^{d-3}}                     & \cdots & \frac{1}{d^{d-3}}  \\ \hline
    
    0      & \cdots & (-1)^{i+j+\delta_{i < j}+p(2)}*t  & \cdots & (-1)^{i+j+\delta_{i < j}+p(1)}*(1-t)  & \cdots & 0                  \\

    0      & \cdots & (-1)^{p(1)+1}*(1-t)               & \cdots & (-1)^{p(2)}*t                         & \cdots & 0                  \\
    
  \end{array}\right) \text{ for  (B) } \to \text{ (C)}
\end{equation} 

\begin{equation}
F^\text{tang} = 
  \left(\begin{array}{@{}ccccccccccc}
    1      & \cdots & 1                               & \cdots & 1                     & 1                  & \cdots & 1                  \\
    
    1      & \cdots & \frac{1}{i}                     & \cdots & \frac{1}{j-1}         & \frac{1}{j}        & \cdots & \frac{1}{d}        \\
    
    \vdots &        &                                 &        & \vdots                &                    &        & \vdots             \\
    
    1      & \cdots & \frac{1}{i^{d-3}}               & \cdots & \frac{1}{(j-1)^{d-3}} & \frac{1}{j^{d-3}}  & \cdots & \frac{1}{d^{d-3}}  \\ \hline
    
    0      & \cdots & (-1)^{i+j+\delta_{i < j}+p(2)}  & \cdots & 0                     & 0                  & \cdots & 0                  \\

    0      & \cdots & 0                               & \cdots & 0                     & (-1)^{p(2)}        & \cdots & 0                  \\
    
  \end{array}\right) \text{ for  (C) } \to \text{ (D)}
\end{equation} 

While this tangent framing is well-defined within each $\msquare_d$, we'd need to specify what the $(d-1)^{th}$ tangent framing's vector, $v^\text{tang}_{d-1}$ is going between adjacent hyper-cubes $\msquare_d(1) \to \msquare_d(2)$. We can just set $v^\text{tang}_{d-1} = t*v^\text{tang}_{d-1}(\msquare_d(1)) + (1-t)*v^\text{tang}_{d-1}(\msquare_d(2))$, where for $k=1,2$, $v^\text{tang}_{d-1}(\msquare_d(k))$ is the $(d-1)^{th}$ tangent frame vector at point (A) of $\msquare_d(2)$.

\subsection{Computing Windings}
First, let's recall what we mean by the windings. For some curve going around in a closed loop $\gamma(\tau)$ parameterized by a `time' variable $\tau$, we'll have that $F^\text{bckd} = F^{\text{rel}}(\tau) F^\text{tang}$ for some `relative framing' $F^{\text{rel}}(\tau) \in GL^+(d)$. Since the loop is closed $F^{\text{rel}}(\tau)$ will be some closed path representing a loop in $GL^+(d)$, which will lift to some path in $\Spin(d)$. At the end of the day, we want to compute whether this path in $\Spin(d)$ is closed or if the ending point is opposite to the starting point. If the start and end points in $\Spin(d)$ are different, then the induced spin structure on the framed loop will be non-bounding (i.e. `anti-periodic'), and if this curve in $\Spin(d)$ is a closed path, the induced spin structure will be bounding (i.e. `periodic'). We can compute this function $F^{\text{rel}}(\tau)$ explicitly and use the information of how it changes within each of the cubes $\msquare_d$ to compute piecewise contributions to the path in $\Spin(d)$. This will be possible since the framings share $(d-2)$ vectors, so restricting to the two vectors that actually vary with respect to each other will allow us to compute the endpoint of this path in terms of the `winding angle' those two vectors undergo with repsect to each other. In particular, this total winding angle will be $2\pi*\text{wind(loop)}$ for some integer `$\text{wind(loop)}$'. If `$\text{wind(loop)}$' is odd the induced spin structure is anti-periodic, and if `$\text{wind(loop)}$' is even it'll be periodic. And, the function $\sigma(\text{loop})$ is related as
\begin{equation}
    \sigma(\text{loop}) = -(-1)^{\text{wind(loop)}}
\end{equation}
which gives $1$ if the induced structure is anti-periodic and gives $-1$ if it's periodic.

To compute the windings, we should compare the last two vectors of $F^\text{bckd}$ and $F^\text{tang}$ and determine what angle in $2 \pi \Z$ they wind with respect to each other. This can be computed solely in terms of the information of $F^{\text{rel}}(\tau)$ on each leg of the journey within each $\msquare$. However, we won't compute $F^{\text{rel}}(\tau)$ directly - we'll instead compute the winding with respect to a homotopic background framing $\tilde{F}^\text{bckd}$ consisting of vector fields $\widehat{v}^{\text{bckd}}_1 \cdots \widehat{v}^{\text{bckd}}_d$. This modified framing is defined as follows:
\begin{equation}
\widehat{F}^\text{bckd} = 
  \left(\begin{array}{@{}ccccccccccc}
    1      & \cdots & 1                     & 1                  & 1                     & \cdots & 1                     & 1                          & 1                     & \cdots & 1                  \\
    
    \vdots &        & \vdots                &                    & \vdots                &        & \vdots                &                            & \vdots                &        & \vdots             \\
    
    1      & \cdots & \frac{1}{(i-1)^{d-3}} & \frac{1}{i^{d-3}}  & \frac{1}{(i+1)^{d-3}} & \cdots & \frac{1}{(j-1)^{d-3}} & \frac{1}{j^{d-3}}          & \frac{1}{(j+2)^{d-3}} & \cdots & \frac{1}{d^{d-3}}  \\ \hline
    
    0      & \cdots & 0                     & 0                  & 0                     & \cdots & 0                     & (-1)^{j+\delta_{i < j}+d}  & 0                     & \cdots & 0                  \\
    
    0      & \cdots & 0                     & (-1)^{(i+1+d)}     & 0                     & \cdots & 0                     & 0                          & 0                     & \cdots & 0                  \\
    
  \end{array}\right) \text{ for  (A)} \to \text{(B)}
\end{equation} 

\begin{equation}
\widehat{F}^\text{bckd} = 
  \left(\begin{array}{@{}ccccccccccc}
    1      & \cdots & 1                            & \cdots & 1                                & \cdots & 1                  \\
    
    \vdots &        &                              &        &                                  &        & \vdots             \\
    
    1      & \cdots & \frac{1}{i^{d-3}}            & \cdots & \frac{1}{j^{d-3}}                & \cdots & \frac{1}{d^{d-3}}  \\ \hline
    
    0      & \cdots & (-1)^{i+\delta_{i < j}+d+1}*t  & \cdots & (-1)^{j+\delta_{i < j}+d}*(1-t)  & \cdots & 0                  \\

    0      & \cdots & (-1)^{(i+1+d)}*(1-t)         & \cdots & (-1)^{(j+1+d)}*t                   & \cdots & 0                  \\
    
  \end{array}\right) \text{ for  (B)} \to \text{(C)}
\end{equation} 

\begin{equation}
\widehat{F}^{\text{bckd}} = 
  \left(\begin{array}{@{}ccccccccccc}
    1      & \cdots & 1                            & \cdots & 1                     & 1                  & \cdots & 1                  \\
    
    \vdots &        &                              &        & \vdots                &                    &        & \vdots             \\
    
    1      & \cdots & \frac{1}{i^{d-3}}            & \cdots & \frac{1}{(j-1)^{d-3}} & \frac{1}{j^{d-3}}  & \cdots & \frac{1}{d^{d-3}}  \\ \hline
    
    0      & \cdots & (-1)^{i+\delta_{i < j}+d+1}  & \cdots & 0                   & 0                  & \cdots & 0                \\

    0      & \cdots & 0                            & \cdots & 0                     & (-1)^{(j+1+d)}     & \cdots & 0                \\
    
  \end{array}\right) \text{ for  (C)} \to \text{(D)}
\end{equation} 
Note that $\widehat{F}^\text{bckd}$ can be obtained as a homotopy from $F^\text{bckd}$ in two steps. The first step is replacing the $d^{th}$ background vector field $v^\text{bckd}_d$ with the $d^{th}$ vector field $\widehat{v}^\text{bckd}_d$ in the above matrices $\widehat{v}^\text{bckd}_d$ via the homotopy $(1-s)*v^\text{bckd}_d + s*\widehat{v}^\text{bckd}_d$. This is okay to do since the determinant of the framing stays away from zero given the sign choices made above. The second step would be replacing the $(d-1)^{th}$ vector by a similar homotopy, which is again okay since the determinant will stay away from zero. 

For reference, the first step of only replacing $v^\text{bckd}_d \to \widehat{v}^\text{bckd}_d$ defines a framing $\tilde{F}^\textbf{bckd}$
\begin{equation}
\tilde{F}^\text{bckd} = 
  \left(\begin{array}{@{}ccccccccccc}
    1      & \cdots & 1                            & \cdots & 1                                & \cdots & 1                  \\
    
    \vdots &        &                              &        &                                  &        & \vdots             \\
    
    1      & \cdots & \frac{1}{i^{d-3}}            & \cdots & \frac{1}{j^{d-3}}                & \cdots & \frac{1}{d^{d-3}}  \\ \hline
    
    1      & \cdots & \frac{1}{i^{d-2}}            & \cdots & \frac{1}{j^{d-2}}                & \cdots & \frac{1}{d^{d-2}}  \\

    0      & \cdots & (-1)^{(i+1+d)}*(1-t)         & \cdots & (-1)^{(j+1+d)}*t                   & \cdots & 0                  \\
    
  \end{array}\right) \text{ for  (B)} \to \text{(C)}.
\end{equation} 
Note that this framing $\tilde{F}^\textbf{bckd}$ is independent of any specific path that taken on the dual 1-skeleton. In particular, away from the center of a $\msquare_d$ on a dual edge $\ihat_\pm$, this vector $\widehat{v}_{d}$ is parallel to the dual 1-skeleton and will point in the direction $(-1)^{i+1+d}$.

Computing $\widehat{F}^\text{bckd}(F^\text{tang})^{-1}$ on the `(B) $\to$ (C)' segment and examining the path that the bottom 2$\times$2 block takes in $GL^+(2)$ will tell us the winding. Write $\widehat{F}^\text{bckd}(\tau) = \widehat{F}^{\text{rel}}(\tau) F^\text{tang}(\tau)$ along the closed loop $\gamma(\tau)$ on the dual skeleton with $\tau$ being the time variable around the entire loop. The winding along a segment parameterized as $t_1 \to t_2$ would be determined by $\{\widehat{F}^{\text{rel}}(\tau) | t_1 \le \tau \le t_2\}$ . Using the windings above and a bit of algebra shows that the winding from $\ihat_{\pm_1} \to \jhat_{\pm_2}$ would be (recall we're assuming $i \neq j$):
\begin{equation}
2 \pi \, \text{wind}(\ihat_{\pm_1} \to \jhat_{\pm_2}) = (-1)^d
\begin{cases}
0    \quad &\text{ if } i+j+p(1)+p(2) = 1 \text{ (mod 2)}   \\
\pi  \quad &\text{ if } i+j+p(1)+p(2) = 0 \text{ (mod 2)}  \text{ AND } i < j \\
-\pi \quad &\text{ if } i+j+p(1)+p(2) = 0 \text{ (mod 2)}  \text{ AND } i > j
\end{cases}.
\end{equation}
Now, we can visualize these partial windings as in Fig.~\ref{fig:windingsGeneralDimensions}. The form of this figure comes from the fact that projecting away the shared framing directions are the same ones that determine the thickening and shifting. So, the pairs of dual edges that occur in the intersection between the dual 1-skeleton and its shifted versions are the pairs that appear in the $\cup_{d-2}$ product of $(d-1)$-cochains.
\begin{equation*}
(\alpha \cup_{d-2} \beta)(\msquare_d) = \sum_{1\le i<j \le d}(\alpha(\ihat_{(-1)^{i+1}})\beta(\jhat_{(-1)^{j+1}})+\alpha(\jhat_{(-1)^{j}})\beta(\ihat_{(-1)^{i}}))
\end{equation*}

\begin{figure}[h!]
  \centering
  \includegraphics[width=\linewidth]{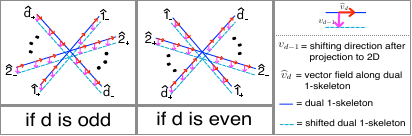}
  \caption{Vector fields $v_{d-1}, \widehat{v}_d$ that can be used to visualize the windings after projecting away the `shared framing directions' to 2D. Solid blue lines are the original dual 1-skeleton and the dashed light-blue lines are the shifted version. $\widehat{v}_d$ points along dual 1-skeleton except near the center. $v_{d-1}$ determines the shifting. (Left) The projected picture for $d$ odd. (Middle) The projected picture for $d$ even. (Right) Legend. }
  \label{fig:windingsGeneralDimensions}
\end{figure}

The total winding is given by
\begin{equation*}
\text{wind(loop)} = \sum_{\ihat_{\pm_1} \to \jhat_{\pm_2}} \text{wind}(\ihat_{\pm_1} \to \jhat_{\pm_2}),
\end{equation*}
so the induced spin structure on the framed curve would determine $\sigma(\text{loop})$ as a product over all such segments $\ihat_{\pm_1} \to \jhat_{\pm_2}$
\begin{equation}
    \sigma(\text{loop}) = - (-1)^{\sum{\ihat_{\pm_1} \to \jhat_{\pm_2}} \text{wind}(\ihat_{\pm_1} \to \jhat_{\pm_2})}.
\end{equation}

\bibliography{bibliography.bib}

\begin{thebibliography}{39}%
\makeatletter
\providecommand \@ifxundefined [1]{%
 \@ifx{#1\undefined}
}%
\providecommand \@ifnum [1]{%
 \ifnum #1\expandafter \@firstoftwo
 \else \expandafter \@secondoftwo
 \fi
}%
\providecommand \@ifx [1]{%
 \ifx #1\expandafter \@firstoftwo
 \else \expandafter \@secondoftwo
 \fi
}%
\providecommand \natexlab [1]{#1}%
\providecommand \enquote  [1]{``#1''}%
\providecommand \bibnamefont  [1]{#1}%
\providecommand \bibfnamefont [1]{#1}%
\providecommand \citenamefont [1]{#1}%
\providecommand \href@noop [0]{\@secondoftwo}%
\providecommand \href [0]{\begingroup \@sanitize@url \@href}%
\providecommand \@href[1]{\@@startlink{#1}\@@href}%
\providecommand \@@href[1]{\endgroup#1\@@endlink}%
\providecommand \@sanitize@url [0]{\catcode `\\12\catcode `\$12\catcode
  `\&12\catcode `\#12\catcode `\^12\catcode `\_12\catcode `\%12\relax}%
\providecommand \@@startlink[1]{}%
\providecommand \@@endlink[0]{}%
\providecommand \url  [0]{\begingroup\@sanitize@url \@url }%
\providecommand \@url [1]{\endgroup\@href {#1}{\urlprefix }}%
\providecommand \urlprefix  [0]{URL }%
\providecommand \Eprint [0]{\href }%
\providecommand \doibase [0]{https://doi.org/}%
\providecommand \selectlanguage [0]{\@gobble}%
\providecommand \bibinfo  [0]{\@secondoftwo}%
\providecommand \bibfield  [0]{\@secondoftwo}%
\providecommand \translation [1]{[#1]}%
\providecommand \BibitemOpen [0]{}%
\providecommand \bibitemStop [0]{}%
\providecommand \bibitemNoStop [0]{.\EOS\space}%
\providecommand \EOS [0]{\spacefactor3000\relax}%
\providecommand \BibitemShut  [1]{\csname bibitem#1\endcsname}%
\let\auto@bib@innerbib\@empty
\bibitem [{\citenamefont {Steenrod}(1947)}]{S47}%
  \BibitemOpen
  \bibfield  {author} {\bibinfo {author} {\bibfnamefont {N.~E.}\ \bibnamefont
  {Steenrod}},\ }\bibfield  {title} {\bibinfo {title} {Products of cocycles and
  extensions of mappings},\ }\href {http://www.jstor.org/stable/1969172}
  {\bibfield  {journal} {\bibinfo  {journal} {Annals of Mathematics}\ }\textbf
  {\bibinfo {volume} {48}},\ \bibinfo {pages} {290} (\bibinfo {year}
  {1947})}\BibitemShut {NoStop}%
\bibitem [{\citenamefont {{Thom}}(1950)}]{thom1950}%
  \BibitemOpen
  \bibfield  {author} {\bibinfo {author} {\bibfnamefont {R.}~\bibnamefont
  {{Thom}}},\ }\bibfield  {title} {\bibinfo {title} {Vari\'et\'es plong\'ees et
  \(i\)-carr\'es},\ }\href@noop {} {\bibfield  {journal} {\bibinfo  {journal}
  {Comptes Rendus Hebdomadaires des S\'eances de l'Acad\'emie des Sciences,
  Paris}\ }\textbf {\bibinfo {volume} {230}},\ \bibinfo {pages} {507} (\bibinfo
  {year} {1950})}\BibitemShut {NoStop}%
\bibitem [{\citenamefont {Gu}\ and\ \citenamefont
  {Wen}(2014)}]{guWen2014Supercohomology}%
  \BibitemOpen
  \bibfield  {author} {\bibinfo {author} {\bibfnamefont {Z.-C.}\ \bibnamefont
  {Gu}}\ and\ \bibinfo {author} {\bibfnamefont {X.-G.}\ \bibnamefont {Wen}},\
  }\bibfield  {title} {\bibinfo {title} {Symmetry-protected topological orders
  for interacting fermions: Fermionic topological nonlinear
  $\ensuremath{\sigma}$ models and a special group supercohomology theory},\
  }\href {https://doi.org/10.1103/PhysRevB.90.115141} {\bibfield  {journal}
  {\bibinfo  {journal} {Phys. Rev. B}\ }\textbf {\bibinfo {volume} {90}},\
  \bibinfo {pages} {115141} (\bibinfo {year} {2014})}\BibitemShut {NoStop}%
\bibitem [{\citenamefont {Thorngren}(2018)}]{T18}%
  \BibitemOpen
  \bibfield  {author} {\bibinfo {author} {\bibfnamefont {R.~G.}\ \bibnamefont
  {Thorngren}},\ }\emph {\bibinfo {title} {{Combinatorial Topology and
  Applications to Quantum Field Theory}}},\ \href@noop {} {Ph.D. thesis},\
  \bibinfo  {school} {UC, Berkeley (main)} (\bibinfo {year} {2018})\BibitemShut
  {NoStop}%
\bibitem [{\citenamefont {Tata}(2020)}]{T20}%
  \BibitemOpen
  \bibfield  {author} {\bibinfo {author} {\bibfnamefont {S.}~\bibnamefont
  {Tata}},\ }\bibfield  {title} {\bibinfo {title} {Geometrically interpreting
  higher cup products, and application to combinatorial pin structures},\
  }\href@noop {} {\bibfield  {journal} {\bibinfo  {journal} {arXiv preprint
  arXiv:2008.10170}\ } (\bibinfo {year} {2020})}\BibitemShut {NoStop}%
\bibitem [{\citenamefont {Chen}(2020)}]{chen2020BosonArbDimensions}%
  \BibitemOpen
  \bibfield  {author} {\bibinfo {author} {\bibfnamefont {Y.-A.}\ \bibnamefont
  {Chen}},\ }\bibfield  {title} {\bibinfo {title} {Exact bosonization in
  arbitrary dimensions},\ }\href
  {https://doi.org/10.1103/PhysRevResearch.2.033527} {\bibfield  {journal}
  {\bibinfo  {journal} {Phys. Rev. Research}\ }\textbf {\bibinfo {volume}
  {2}},\ \bibinfo {pages} {033527} (\bibinfo {year} {2020})}\BibitemShut
  {NoStop}%
\bibitem [{\citenamefont {Levin}\ and\ \citenamefont {Gu}(2012)}]{levinGu2012}%
  \BibitemOpen
  \bibfield  {author} {\bibinfo {author} {\bibfnamefont {M.}~\bibnamefont
  {Levin}}\ and\ \bibinfo {author} {\bibfnamefont {Z.-C.}\ \bibnamefont {Gu}},\
  }\bibfield  {title} {\bibinfo {title} {Braiding statistics approach to
  symmetry-protected topological phases},\ }\href
  {https://doi.org/10.1103/PhysRevB.86.115109} {\bibfield  {journal} {\bibinfo
  {journal} {Phys. Rev. B}\ }\textbf {\bibinfo {volume} {86}},\ \bibinfo
  {pages} {115109} (\bibinfo {year} {2012})}\BibitemShut {NoStop}%
\bibitem [{\citenamefont {Burnell}\ \emph {et~al.}(2014)\citenamefont
  {Burnell}, \citenamefont {Chen}, \citenamefont {Fidkowski},\ and\
  \citenamefont {Vishwanath}}]{BCFV14}%
  \BibitemOpen
  \bibfield  {author} {\bibinfo {author} {\bibfnamefont {F.~J.}\ \bibnamefont
  {Burnell}}, \bibinfo {author} {\bibfnamefont {X.}~\bibnamefont {Chen}},
  \bibinfo {author} {\bibfnamefont {L.}~\bibnamefont {Fidkowski}},\ and\
  \bibinfo {author} {\bibfnamefont {A.}~\bibnamefont {Vishwanath}},\ }\bibfield
   {title} {\bibinfo {title} {Exactly soluble model of a three-dimensional
  symmetry-protected topological phase of bosons with surface topological
  order},\ }\href@noop {} {\bibfield  {journal} {\bibinfo  {journal} {Physical
  Review B}\ }\textbf {\bibinfo {volume} {90}},\ \bibinfo {pages} {245122}
  (\bibinfo {year} {2014})}\BibitemShut {NoStop}%
\bibitem [{\citenamefont {Kapustin}\ and\ \citenamefont
  {Thorngren}(2017)}]{kapustinThorngren2017}%
  \BibitemOpen
  \bibfield  {author} {\bibinfo {author} {\bibfnamefont {A.}~\bibnamefont
  {Kapustin}}\ and\ \bibinfo {author} {\bibfnamefont {R.}~\bibnamefont
  {Thorngren}},\ }\bibfield  {title} {\bibinfo {title} {Fermionic spt phases in
  higher dimensions and bosonization},\ }\bibfield  {journal} {\bibinfo
  {journal} {Journal of High Energy Physics}\ }\textbf {\bibinfo {volume}
  {2017}},\ \href {https://doi.org/10.1007/jhep10(2017)080}
  {10.1007/jhep10(2017)080} (\bibinfo {year} {2017})\BibitemShut {NoStop}%
\bibitem [{\citenamefont {Haah}\ \emph {et~al.}(2018)\citenamefont {Haah},
  \citenamefont {Fidkowski},\ and\ \citenamefont {Hastings}}]{HFH18}%
  \BibitemOpen
  \bibfield  {author} {\bibinfo {author} {\bibfnamefont {J.}~\bibnamefont
  {Haah}}, \bibinfo {author} {\bibfnamefont {L.}~\bibnamefont {Fidkowski}},\
  and\ \bibinfo {author} {\bibfnamefont {M.~B.}\ \bibnamefont {Hastings}},\
  }\bibfield  {title} {\bibinfo {title} {Nontrivial quantum cellular automata
  in higher dimensions},\ }\href@noop {} {\bibfield  {journal} {\bibinfo
  {journal} {arXiv preprint arXiv:1812.01625}\ } (\bibinfo {year}
  {2018})}\BibitemShut {NoStop}%
\bibitem [{\citenamefont {Chen}\ \emph {et~al.}()\citenamefont {Chen},
  \citenamefont {Dua}, \citenamefont {Ellison}, \citenamefont {Shirley},
  \citenamefont {Tantivasadakarn},\ and\ \citenamefont
  {Williamson}}]{myfuturework}%
  \BibitemOpen
  \bibfield  {author} {\bibinfo {author} {\bibfnamefont {Y.-A.}\ \bibnamefont
  {Chen}}, \bibinfo {author} {\bibfnamefont {A.}~\bibnamefont {Dua}}, \bibinfo
  {author} {\bibfnamefont {T.}~\bibnamefont {Ellison}}, \bibinfo {author}
  {\bibfnamefont {W.}~\bibnamefont {Shirley}}, \bibinfo {author} {\bibfnamefont
  {N.}~\bibnamefont {Tantivasadakarn}},\ and\ \bibinfo {author} {\bibfnamefont
  {D.}~\bibnamefont {Williamson}},\ }\href@noop {} {\bibinfo {title} {In
  preparation}}\BibitemShut {NoStop}%
\bibitem [{\citenamefont {Freedman}\ and\ \citenamefont
  {Hastings}(2016)}]{FH16}%
  \BibitemOpen
  \bibfield  {author} {\bibinfo {author} {\bibfnamefont {M.~H.}\ \bibnamefont
  {Freedman}}\ and\ \bibinfo {author} {\bibfnamefont {M.~B.}\ \bibnamefont
  {Hastings}},\ }\bibfield  {title} {\bibinfo {title} {Double semions in
  arbitrary dimension},\ }\href@noop {} {\bibfield  {journal} {\bibinfo
  {journal} {Communications in Mathematical Physics}\ }\textbf {\bibinfo
  {volume} {347}},\ \bibinfo {pages} {389} (\bibinfo {year}
  {2016})}\BibitemShut {NoStop}%
\bibitem [{\citenamefont {Debray}(2019)}]{debray2019}%
  \BibitemOpen
  \bibfield  {author} {\bibinfo {author} {\bibfnamefont {A.}~\bibnamefont
  {Debray}},\ }\bibfield  {title} {\bibinfo {title} {The low-energy tqft of the
  generalized double semion model},\ }\href
  {https://doi.org/10.1007/s00220-019-03554-w} {\bibfield  {journal} {\bibinfo
  {journal} {Communications in Mathematical Physics}\ }\textbf {\bibinfo
  {volume} {375}},\ \bibinfo {pages} {1079–1115} (\bibinfo {year}
  {2019})}\BibitemShut {NoStop}%
\bibitem [{\citenamefont {Chen}\ \emph {et~al.}(2018)\citenamefont {Chen},
  \citenamefont {Kapustin},\ and\ \citenamefont {Radicevic}}]{CKR18}%
  \BibitemOpen
  \bibfield  {author} {\bibinfo {author} {\bibfnamefont {Y.-A.}\ \bibnamefont
  {Chen}}, \bibinfo {author} {\bibfnamefont {A.}~\bibnamefont {Kapustin}},\
  and\ \bibinfo {author} {\bibfnamefont {D.}~\bibnamefont {Radicevic}},\
  }\bibfield  {title} {\bibinfo {title} {Exact bosonization in two spatial
  dimensions and a new class of lattice gauge theories},\ }\href
  {https://doi.org/https://doi.org/10.1016/j.aop.2018.03.024} {\bibfield
  {journal} {\bibinfo  {journal} {Annals of Physics}\ }\textbf {\bibinfo
  {volume} {393}},\ \bibinfo {pages} {234} (\bibinfo {year}
  {2018})}\BibitemShut {NoStop}%
\bibitem [{\citenamefont {Chen}\ and\ \citenamefont {Kapustin}(2019)}]{CK18}%
  \BibitemOpen
  \bibfield  {author} {\bibinfo {author} {\bibfnamefont {Y.-A.}\ \bibnamefont
  {Chen}}\ and\ \bibinfo {author} {\bibfnamefont {A.}~\bibnamefont
  {Kapustin}},\ }\bibfield  {title} {\bibinfo {title} {Bosonization in three
  spatial dimensions and a 2-form gauge theory},\ }\href
  {https://doi.org/10.1103/PhysRevB.100.245127} {\bibfield  {journal} {\bibinfo
   {journal} {Phys. Rev. B}\ }\textbf {\bibinfo {volume} {100}},\ \bibinfo
  {pages} {245127} (\bibinfo {year} {2019})}\BibitemShut {NoStop}%
\bibitem [{\citenamefont {Hatcher}(2000)}]{H00}%
  \BibitemOpen
  \bibfield  {author} {\bibinfo {author} {\bibfnamefont {A.}~\bibnamefont
  {Hatcher}},\ }\href {https://cds.cern.ch/record/478079} {\emph {\bibinfo
  {title} {{Algebraic topology}}}}\ (\bibinfo  {publisher} {Cambridge Univ.
  Press},\ \bibinfo {address} {Cambridge},\ \bibinfo {year} {2000})\BibitemShut
  {NoStop}%
\bibitem [{\citenamefont {Halperin}\ and\ \citenamefont
  {Toledo}(1972)}]{HalperinToledo}%
  \BibitemOpen
  \bibfield  {author} {\bibinfo {author} {\bibfnamefont {S.}~\bibnamefont
  {Halperin}}\ and\ \bibinfo {author} {\bibfnamefont {D.}~\bibnamefont
  {Toledo}},\ }\bibfield  {title} {\bibinfo {title} {Stiefel-whitney homology
  classes},\ }\href {http://www.jstor.org/stable/1970823} {\bibfield  {journal}
  {\bibinfo  {journal} {Annals of Mathematics}\ }\textbf {\bibinfo {volume}
  {96}},\ \bibinfo {pages} {511} (\bibinfo {year} {1972})}\BibitemShut
  {NoStop}%
\bibitem [{mor()}]{morganHigherCupLecture}%
  \BibitemOpen
  \href {https://www.youtube.com/watch?v=mpbWQbkl8_g#t=20m15s} {\bibinfo
  {title} {John morgan - higher cup products - homotopy theory lectures at
  stony brook - 29 oct. 2018}}\BibitemShut {NoStop}%
\bibitem [{\citenamefont {Johnson}(1980)}]{johnson1980}%
  \BibitemOpen
  \bibfield  {author} {\bibinfo {author} {\bibfnamefont {D.}~\bibnamefont
  {Johnson}},\ }\bibfield  {title} {\bibinfo {title} {Spin structures and
  quadratic forms on surfaces},\ }\href
  {https://doi.org/https://doi.org/10.1112/jlms/s2-22.2.365} {\bibfield
  {journal} {\bibinfo  {journal} {Journal of the London Mathematical Society}\
  }\textbf {\bibinfo {volume} {s2-22}},\ \bibinfo {pages} {365} (\bibinfo
  {year} {1980})},\ \Eprint
  {https://arxiv.org/abs/https://londmathsoc.onlinelibrary.wiley.com/doi/pdf/10.1112/jlms/s2-22.2.365}
  {https://londmathsoc.onlinelibrary.wiley.com/doi/pdf/10.1112/jlms/s2-22.2.365}
  \BibitemShut {NoStop}%
\bibitem [{\citenamefont {Gaiotto}\ and\ \citenamefont
  {Kapustin}(2016)}]{GK16}%
  \BibitemOpen
  \bibfield  {author} {\bibinfo {author} {\bibfnamefont {D.}~\bibnamefont
  {Gaiotto}}\ and\ \bibinfo {author} {\bibfnamefont {A.}~\bibnamefont
  {Kapustin}},\ }\bibfield  {title} {\bibinfo {title} {Spin tqfts and fermionic
  phases of matter},\ }\href {https://doi.org/10.1142/S0217751X16450445}
  {\bibfield  {journal} {\bibinfo  {journal} {International Journal of Modern
  Physics A}\ }\textbf {\bibinfo {volume} {31}},\ \bibinfo {pages} {1645044}
  (\bibinfo {year} {2016})}\BibitemShut {NoStop}%
\bibitem [{\citenamefont {Tata}\ \emph {et~al.}(2021)\citenamefont {Tata},
  \citenamefont {Kobayashi}, \citenamefont {Bulmash},\ and\ \citenamefont
  {Barkeshli}}]{TKBB2021anomalies}%
  \BibitemOpen
  \bibfield  {author} {\bibinfo {author} {\bibfnamefont {S.}~\bibnamefont
  {Tata}}, \bibinfo {author} {\bibfnamefont {R.}~\bibnamefont {Kobayashi}},
  \bibinfo {author} {\bibfnamefont {D.}~\bibnamefont {Bulmash}},\ and\ \bibinfo
  {author} {\bibfnamefont {M.}~\bibnamefont {Barkeshli}},\ }\href@noop {}
  {\bibinfo {title} {Anomalies in (2+1)d fermionic topological phases and
  (3+1)d path integral state sums for fermionic spts}} (\bibinfo {year}
  {2021}),\ \Eprint {https://arxiv.org/abs/2104.14567} {arXiv:2104.14567
  [cond-mat.str-el]} \BibitemShut {NoStop}%
\bibitem [{\citenamefont {Turaev}\ and\ \citenamefont
  {Viro}(1992)}]{turaevViro1992}%
  \BibitemOpen
  \bibfield  {author} {\bibinfo {author} {\bibfnamefont {V.}~\bibnamefont
  {Turaev}}\ and\ \bibinfo {author} {\bibfnamefont {O.}~\bibnamefont {Viro}},\
  }\bibfield  {title} {\bibinfo {title} {State sum invariants of 3-manifolds
  and quantum 6j-symbols},\ }\href@noop {} {\bibfield  {journal} {\bibinfo
  {journal} {Topology}\ }\textbf {\bibinfo {volume} {31}},\ \bibinfo {pages}
  {865} (\bibinfo {year} {1992})}\BibitemShut {NoStop}%
\bibitem [{\citenamefont {Levin}\ and\ \citenamefont
  {Wen}(2005)}]{levinWen2005}%
  \BibitemOpen
  \bibfield  {author} {\bibinfo {author} {\bibfnamefont {M.~A.}\ \bibnamefont
  {Levin}}\ and\ \bibinfo {author} {\bibfnamefont {X.-G.}\ \bibnamefont
  {Wen}},\ }\bibfield  {title} {\bibinfo {title} {String-net condensation: A
  physical mechanism for topological phases},\ }\bibfield  {journal} {\bibinfo
  {journal} {Physical Review B}\ }\textbf {\bibinfo {volume} {71}},\ \href
  {https://doi.org/10.1103/physrevb.71.045110} {10.1103/physrevb.71.045110}
  (\bibinfo {year} {2005})\BibitemShut {NoStop}%
\bibitem [{\citenamefont {Crane}\ and\ \citenamefont
  {Yetter}(1993)}]{craneYetter1993}%
  \BibitemOpen
  \bibfield  {author} {\bibinfo {author} {\bibfnamefont {L.}~\bibnamefont
  {Crane}}\ and\ \bibinfo {author} {\bibfnamefont {D.}~\bibnamefont {Yetter}},\
  }\bibfield  {title} {\bibinfo {title} {{A categorical construction of 4D
  TQFTs}},\ }in\ \href@noop {} {\emph {\bibinfo {booktitle} {Quantum
  Topology}}},\ \bibinfo {editor} {edited by\ \bibinfo {editor} {\bibfnamefont
  {L.}~\bibnamefont {Kauffman}}\ and\ \bibinfo {editor} {\bibfnamefont
  {R.}~\bibnamefont {Baadhio}}}\ (\bibinfo  {publisher} {World Scientific},\
  \bibinfo {address} {Singapore},\ \bibinfo {year} {1993})\ \Eprint
  {https://arxiv.org/abs/arXiv:hep-th/9301062} {arXiv:hep-th/9301062}
  \BibitemShut {NoStop}%
\bibitem [{\citenamefont {Walker}\ and\ \citenamefont
  {Wang}(2012{\natexlab{a}})}]{walkerWang2012}%
  \BibitemOpen
  \bibfield  {author} {\bibinfo {author} {\bibfnamefont {K.}~\bibnamefont
  {Walker}}\ and\ \bibinfo {author} {\bibfnamefont {Z.}~\bibnamefont {Wang}},\
  }\href@noop {} {\bibfield  {journal} {\bibinfo  {journal} {Frontier of
  Physics}\ }\textbf {\bibinfo {volume} {7}},\ \bibinfo {pages} {150} (\bibinfo
  {year} {2012}{\natexlab{a}})}\BibitemShut {NoStop}%
\bibitem [{\citenamefont {Bhardwaj}\ \emph {et~al.}(2017)\citenamefont
  {Bhardwaj}, \citenamefont {Gaiotto},\ and\ \citenamefont {Kapustin}}]{BGK17}%
  \BibitemOpen
  \bibfield  {author} {\bibinfo {author} {\bibfnamefont {L.}~\bibnamefont
  {Bhardwaj}}, \bibinfo {author} {\bibfnamefont {D.}~\bibnamefont {Gaiotto}},\
  and\ \bibinfo {author} {\bibfnamefont {A.}~\bibnamefont {Kapustin}},\
  }\bibfield  {title} {\bibinfo {title} {State sum constructions of spin-tfts
  and string net constructions of fermionic phases of matter},\ }\bibfield
  {journal} {\bibinfo  {journal} {Journal of High Energy Physics}\ }\textbf
  {\bibinfo {volume} {2017}},\ \href {https://doi.org/10.1007/jhep04(2017)096}
  {10.1007/jhep04(2017)096} (\bibinfo {year} {2017})\BibitemShut {NoStop}%
\bibitem [{\citenamefont {Walker}\ and\ \citenamefont
  {Wang}(2012{\natexlab{b}})}]{WW12}%
  \BibitemOpen
  \bibfield  {author} {\bibinfo {author} {\bibfnamefont {K.}~\bibnamefont
  {Walker}}\ and\ \bibinfo {author} {\bibfnamefont {Z.}~\bibnamefont {Wang}},\
  }\bibfield  {title} {\bibinfo {title} {(3+ 1)-tqfts and topological
  insulators},\ }\href@noop {} {\bibfield  {journal} {\bibinfo  {journal}
  {Frontiers of Physics}\ }\textbf {\bibinfo {volume} {7}},\ \bibinfo {pages}
  {150} (\bibinfo {year} {2012}{\natexlab{b}})}\BibitemShut {NoStop}%
\bibitem [{\citenamefont {Kapustin}\ and\ \citenamefont
  {Thorngren}(2014)}]{KT14}%
  \BibitemOpen
  \bibfield  {author} {\bibinfo {author} {\bibfnamefont {A.}~\bibnamefont
  {Kapustin}}\ and\ \bibinfo {author} {\bibfnamefont {R.}~\bibnamefont
  {Thorngren}},\ }\bibfield  {title} {\bibinfo {title} {Topological field
  theory on a lattice, discrete theta-angles and confinement},\ }\href@noop {}
  {\bibfield  {journal} {\bibinfo  {journal} {Advances in Theoretical and
  Mathematical Physics}\ }\textbf {\bibinfo {volume} {18}},\ \bibinfo {pages}
  {1233} (\bibinfo {year} {2014})}\BibitemShut {NoStop}%
\bibitem [{\citenamefont {Hsin}\ \emph {et~al.}(2019)\citenamefont {Hsin},
  \citenamefont {Lam},\ and\ \citenamefont {Seiberg}}]{HLS19}%
  \BibitemOpen
  \bibfield  {author} {\bibinfo {author} {\bibfnamefont {P.-S.}\ \bibnamefont
  {Hsin}}, \bibinfo {author} {\bibfnamefont {H.~T.}\ \bibnamefont {Lam}},\ and\
  \bibinfo {author} {\bibfnamefont {N.}~\bibnamefont {Seiberg}},\ }\bibfield
  {title} {\bibinfo {title} {Comments on one-form global symmetries and their
  gauging in 3d and 4d},\ }\href@noop {} {\bibfield  {journal} {\bibinfo
  {journal} {SciPost Physics}\ }\textbf {\bibinfo {volume} {6}},\ \bibinfo
  {pages} {Art} (\bibinfo {year} {2019})}\BibitemShut {NoStop}%
\bibitem [{\citenamefont {Hsin}\ \emph {et~al.}(2021)\citenamefont {Hsin},
  \citenamefont {Ji},\ and\ \citenamefont {Jian}}]{HJJ21}%
  \BibitemOpen
  \bibfield  {author} {\bibinfo {author} {\bibfnamefont {P.-S.}\ \bibnamefont
  {Hsin}}, \bibinfo {author} {\bibfnamefont {W.}~\bibnamefont {Ji}},\ and\
  \bibinfo {author} {\bibfnamefont {C.-M.}\ \bibnamefont {Jian}},\ }\bibfield
  {title} {\bibinfo {title} {Exotic invertible phases with higher-group
  symmetries},\ }\href@noop {} {\bibfield  {journal} {\bibinfo  {journal}
  {arXiv preprint arXiv:2105.09454}\ } (\bibinfo {year} {2021})}\BibitemShut
  {NoStop}%
\bibitem [{\citenamefont {Barkeshli}\ \emph {et~al.}(2019)\citenamefont
  {Barkeshli}, \citenamefont {Bonderson}, \citenamefont {Cheng}, \citenamefont
  {Jian},\ and\ \citenamefont {Walker}}]{barkeshli2019_nonorientable}%
  \BibitemOpen
  \bibfield  {author} {\bibinfo {author} {\bibfnamefont {M.}~\bibnamefont
  {Barkeshli}}, \bibinfo {author} {\bibfnamefont {P.}~\bibnamefont
  {Bonderson}}, \bibinfo {author} {\bibfnamefont {M.}~\bibnamefont {Cheng}},
  \bibinfo {author} {\bibfnamefont {C.-M.}\ \bibnamefont {Jian}},\ and\
  \bibinfo {author} {\bibfnamefont {K.}~\bibnamefont {Walker}},\ }\bibfield
  {title} {\bibinfo {title} {Reflection and time reversal symmetry enriched
  topological phases of matter: Path integrals, non-orientable manifolds, and
  anomalies},\ }\href {https://doi.org/10.1007/s00220-019-03475-8} {\bibfield
  {journal} {\bibinfo  {journal} {Communications in Mathematical Physics}\
  }\textbf {\bibinfo {volume} {374}},\ \bibinfo {pages} {1021–1124} (\bibinfo
  {year} {2019})}\BibitemShut {NoStop}%
\bibitem [{\citenamefont {Wang}\ \emph {et~al.}(2019)\citenamefont {Wang},
  \citenamefont {Shirley},\ and\ \citenamefont {Chen}}]{wangEtAl2019}%
  \BibitemOpen
  \bibfield  {author} {\bibinfo {author} {\bibfnamefont {T.}~\bibnamefont
  {Wang}}, \bibinfo {author} {\bibfnamefont {W.}~\bibnamefont {Shirley}},\ and\
  \bibinfo {author} {\bibfnamefont {X.}~\bibnamefont {Chen}},\ }\bibfield
  {title} {\bibinfo {title} {Foliated fracton order in the majorana
  checkerboard model},\ }\href {https://doi.org/10.1103/PhysRevB.100.085127}
  {\bibfield  {journal} {\bibinfo  {journal} {Phys. Rev. B}\ }\textbf {\bibinfo
  {volume} {100}},\ \bibinfo {pages} {085127} (\bibinfo {year}
  {2019})}\BibitemShut {NoStop}%
\bibitem [{\citenamefont {Tantivasadakarn}\ and\ \citenamefont
  {Vishwanath}(2018)}]{TV18}%
  \BibitemOpen
  \bibfield  {author} {\bibinfo {author} {\bibfnamefont {N.}~\bibnamefont
  {Tantivasadakarn}}\ and\ \bibinfo {author} {\bibfnamefont {A.}~\bibnamefont
  {Vishwanath}},\ }\bibfield  {title} {\bibinfo {title} {Full commuting
  projector hamiltonians of interacting symmetry-protected topological phases
  of fermions},\ }\href@noop {} {\bibfield  {journal} {\bibinfo  {journal}
  {Physical Review B}\ }\textbf {\bibinfo {volume} {98}},\ \bibinfo {pages}
  {165104} (\bibinfo {year} {2018})}\BibitemShut {NoStop}%
\bibitem [{\citenamefont {Fidkowski}\ \emph {et~al.}(2020)\citenamefont
  {Fidkowski}, \citenamefont {Haah}, \citenamefont {Hastings},\ and\
  \citenamefont {Tantivasadakarn}}]{FHHT20}%
  \BibitemOpen
  \bibfield  {author} {\bibinfo {author} {\bibfnamefont {L.}~\bibnamefont
  {Fidkowski}}, \bibinfo {author} {\bibfnamefont {J.}~\bibnamefont {Haah}},
  \bibinfo {author} {\bibfnamefont {M.~B.}\ \bibnamefont {Hastings}},\ and\
  \bibinfo {author} {\bibfnamefont {N.}~\bibnamefont {Tantivasadakarn}},\
  }\bibfield  {title} {\bibinfo {title} {Disentangling the generalized double
  semion model},\ }\href@noop {} {\bibfield  {journal} {\bibinfo  {journal}
  {Communications in Mathematical Physics}\ }\textbf {\bibinfo {volume}
  {380}},\ \bibinfo {pages} {1151} (\bibinfo {year} {2020})}\BibitemShut
  {NoStop}%
\bibitem [{\citenamefont {Goldstein}\ and\ \citenamefont
  {Turner}(1976)}]{GT76}%
  \BibitemOpen
  \bibfield  {author} {\bibinfo {author} {\bibfnamefont {R.~Z.}\ \bibnamefont
  {Goldstein}}\ and\ \bibinfo {author} {\bibfnamefont {E.~C.}\ \bibnamefont
  {Turner}},\ }\bibfield  {title} {\bibinfo {title} {A formula for
  stiefel-whitney homology classes},\ }\href
  {https://doi.org/10.1090/s0002-9939-1976-0415643-5} {\bibfield  {journal}
  {\bibinfo  {journal} {Proceedings of the American Mathematical Society}\
  }\textbf {\bibinfo {volume} {58}},\ \bibinfo {pages} {339} (\bibinfo {year}
  {1976})}\BibitemShut {NoStop}%
\bibitem [{\citenamefont {Barkeshli}\ \emph {et~al.}(2013)\citenamefont
  {Barkeshli}, \citenamefont {Jian},\ and\ \citenamefont
  {Qi}}]{BJQ2013_abelianTopOrder}%
  \BibitemOpen
  \bibfield  {author} {\bibinfo {author} {\bibfnamefont {M.}~\bibnamefont
  {Barkeshli}}, \bibinfo {author} {\bibfnamefont {C.-M.}\ \bibnamefont
  {Jian}},\ and\ \bibinfo {author} {\bibfnamefont {X.-L.}\ \bibnamefont {Qi}},\
  }\bibfield  {title} {\bibinfo {title} {Theory of defects in abelian
  topological states},\ }\href {https://doi.org/10.1103/PhysRevB.88.235103}
  {\bibfield  {journal} {\bibinfo  {journal} {Phys. Rev. B}\ }\textbf {\bibinfo
  {volume} {88}},\ \bibinfo {pages} {235103} (\bibinfo {year}
  {2013})}\BibitemShut {NoStop}%
\bibitem [{\citenamefont {Thorngren}\ and\ \citenamefont
  {Else}(2018)}]{thorngrenElse2018_crystalline}%
  \BibitemOpen
  \bibfield  {author} {\bibinfo {author} {\bibfnamefont {R.}~\bibnamefont
  {Thorngren}}\ and\ \bibinfo {author} {\bibfnamefont {D.~V.}\ \bibnamefont
  {Else}},\ }\bibfield  {title} {\bibinfo {title} {Gauging spatial symmetries
  and the classification of topological crystalline phases},\ }\bibfield
  {journal} {\bibinfo  {journal} {Physical Review X}\ }\textbf {\bibinfo
  {volume} {8}},\ \href {https://doi.org/10.1103/physrevx.8.011040}
  {10.1103/physrevx.8.011040} (\bibinfo {year} {2018})\BibitemShut {NoStop}%
\bibitem [{\citenamefont {Manjunath}\ and\ \citenamefont
  {Barkeshli}(2020)}]{manjunath2020classification}%
  \BibitemOpen
  \bibfield  {author} {\bibinfo {author} {\bibfnamefont {N.}~\bibnamefont
  {Manjunath}}\ and\ \bibinfo {author} {\bibfnamefont {M.}~\bibnamefont
  {Barkeshli}},\ }\href@noop {} {\bibinfo {title} {Classification of fractional
  quantum hall states with spatial symmetries}} (\bibinfo {year} {2020}),\
  \Eprint {https://arxiv.org/abs/2012.11603} {arXiv:2012.11603
  [cond-mat.str-el]} \BibitemShut {NoStop}%
\bibitem [{\citenamefont {Debray}(2021)}]{debray2021invertible}%
  \BibitemOpen
  \bibfield  {author} {\bibinfo {author} {\bibfnamefont {A.}~\bibnamefont
  {Debray}},\ }\href@noop {} {\bibinfo {title} {Invertible phases for mixed
  spatial symmetries and the fermionic crystalline equivalence principle}}
  (\bibinfo {year} {2021}),\ \Eprint {https://arxiv.org/abs/2102.02941}
  {arXiv:2102.02941 [math-ph]} \BibitemShut {NoStop}%
\end{thebibliography}%

\end{document}